%% file: gdefj.tex
\documentclass{lmcs}
\pdfoutput=1

\usepackage{silence}
\usepackage{lastpage}
\lmcsdoi{15}{3}{29}
\lmcsheading{}{\pageref{LastPage}}{}{}%
{Aug.~01,~2018}{Sep.~10,~2019}{}

\usepackage{hyperref}
\usepackage{amssymb}
\usepackage{xspace}
\usepackage{graphicx}
\usepackage{subcaption}
\usepackage{tikz}

\input{macros}


\allowdisplaybreaks%

\begin{document}

\title{Definability and Interpolation \texorpdfstring{\\}{}within Decidable Fixpoint Logics\rsuper*}
\titlecomment{{\lsuper*}This is the journal version of material appearing in ICALP~2017~\cite{icalp17},
and extending work in LICS~2015~\cite{lics15-gnfpi}.}

\author[M. Benedikt]{Michael Benedikt\rsuper{a}}
\address{\lsuper{a}University of Oxford}
\email{michael.benedikt@cs.ox.ac.uk}

\author[P. Bourhis]{Pierre Bourhis\rsuper{b}}
\address{\lsuper{b}CNRS CRIStAL UMR 9189, University of Lille, INRIA Lille}
\email{pierre.bourhis@univ-lille.fr}

\author[M. Vanden Boom]{Michael Vanden Boom\rsuper{a}}


\input{abstract}

\maketitle

\input{intro}
\input{related}

\input{organ}
\input{prelims}
\input{highlevel}
\input{gfp}
\input{gnfp}
\input{interp}
\input{conc}
\input{ack}

\bibliographystyle{alpha}
\bibliography{guarded-def}

\end{document}

%% file: macros.tex
\newcommand{\powerset}[1]{\mathcal{P}(#1)} 
\newcommand{\LFPA}[2]{\LFP_{#1,#2}}

\newcommand{\nf}{\text{strict normal form}\xspace}
\newcommand{\logic}[1]{\textup{\small #1}\xspace}
\renewcommand{\vec}[1]{\boldsymbol{#1}}
\newcommand{\cA}{\mathcal{A}}

\newcommand{\cC}{\mathcal{C}}

\newcommand{\cG}{\mathcal{G}}
\newcommand{\cM}{\mathcal{M}}

\newcommand{\cT}{\mathcal{T}}
\newcommand{\cU}{\mathcal{U}}

\newcommand{\N}{\mathbb{N}}

\newcommand{\fA}{\mathfrak{A}}
\newcommand{\fB}{\mathfrak{B}}
\newcommand{\fU}{\mathfrak{U}}
\newcommand{\fM}{\mathfrak{M}}
\newcommand{\set}[1]{\left\{{#1}\right\}}
\newcommand{\sset}[1]{\{{#1}\}}
\newcommand{\size}[1]{\lvert#1\rvert}
\newcommand{\width}[1]{\operatorname{width}(#1)} 
\newcommand{\restrict}[2]{ #1\restriction_{#2} }

\renewcommand{\L}{\mathrm{L}}
\newcommand{\R}{\mathrm{R}}

\newcommand{\free}[1]{\operatorname{free}(#1)} 
\newcommand{\dom}[1]{\operatorname{dom}(#1)} 
\newcommand{\gn}{\textup{GN}\xspace}
\newcommand{\gnfo}{\logic{GNF}}
\newcommand{\gso}{\logic{GSO}}
\newcommand{\mso}{\logic{MSO}}
\newcommand{\gnf}{\gnfo}
\newcommand{\gf}{\logic{GF}}
\newcommand{\lfp}{\logic{LFP}}
\newcommand{\unfo}{\logic{UNF}}
\newcommand{\fo}{\logic{FO}}
\newcommand{\unf}{\unfo}
\newcommand{\Lmu}{\textup{L}_{\mu}}
\newcommand{\gnfpup}{\logic{GNFP-UP}}

\newcommand{\gfp}{\logic{GFP}}
\newcommand{\LFP}{\operatorname{\bf lfp}}

\newcommand{\gnfp}{\logic{GNFP}}
\newcommand{\unfp}{\logic{UNFP}}
\newcommand{\unfpk}{\unfp^k}
\newcommand{\gnfk}{\gnf^k}
\newcommand{\unfk}{\unf^k}
\newcommand{\lfpl}{\logic{LFP}}
\newcommand{\complexity}[1]{\textsc{#1}}
\newcommand\ptime{\complexity{PTime}\xspace}
\newcommand\exptime{\complexity{ExpTime}\xspace}

\newcommand\twoexptime{\complexity{2-ExpTime}\xspace}
\newcommand{\twoexp}{\twoexptime}

\newcommand{\dmodality}[1]{\langle#1\rangle}
\newcommand{\bmodality}[1]{[ #1 ]}
\newcommand{\valuation}{\ell}
\newcommand{\codom}[1]{\operatorname{rng}(#1)} 
\newcommand{\decode}[1]{\mathfrak{D}(#1)} 
\newcommand{\tree}{\cT}
\newcommand{\cL}{\mathcal{L}}
\newcommand{\gnfpk}{\text{$\gnfp^k$}\xspace}
\renewcommand{\varphi}{\phi}
\newcommand{\kw}[1]{\textsc{#1}}
\newcommand{\atypesig}[3]{\kw{AT}_{#2,#3}(#1)} 
\newcommand{\ExactLabel}[1]{\kw{ExactLabel}(#1)} 
\newcommand{\GNLabel}[1]{\kw{GNLabel}(#1)} 
\newcommand{\sigmag}{\sigma_g}

\newcommand{\gnkinvar}{\logic{GN}^k}
\newcommand{\sgnkinvar}{\logic{BGN}^k}

\newcommand{\gnlinvar}{\logic{GN}^l}
\newcommand{\gnninvar}{\logic{GN}^n}

\newcommand{\ginvar}{\logic{G}}
\newcommand{\bunravelk}[2]{\cU_{\textup{BGN}^k[#2]}(#1)} 
\newcommand{\unravelk}[2]{\cU_{\textup{GN}^k[#2]}(#1)} 
\newcommand{\unravelm}[2]{\cU_{\textup{GN}^m[#2]}(#1)} 
\newcommand{\sunravelk}[2]{\cU^{\textup{plump}}_{\textup{BGN}^k[#2]}(#1)} 
\newcommand{\gunravel}[2]{\cU_{\textup{G}[#2]}(#1)} 
\newcommand{\sigcode}[2]{\Sigma^{\text{code}}_{#1,#2}}
\newcommand{\dloglite}{\logic{DATALOG-LITE}}

\newcommand{\unravellingcode}{\cU_{\logictarget}(\fB)} 
\newcommand{\lemmafwd}{\hyperref[lemma:forward-gso]{Lemma~Fwd}\xspace}
\newcommand{\lemmagfpbwd}[1]{\hyperref[lemma:backwards-gfp]{#1}\xspace}
\newcommand{\lemmagnfpbwd}[1]{\hyperref[lemma:backwards-gnfp]{#1}\xspace}
\newcommand{\logictarget}{\cL_1}
\newcommand{\sigtarget}{\sigma'}
\newcommand{\sigoriginal}{\sigma}
\newcommand{\Bekic}{Beki\v{c}\xspace}
\newcommand{\convertnf}[1]{\operatorname{convert}(#1)} 
\newcommand{\indices}[1]{\textup{names}(#1)} 
\newcommand{\tuple}[1]{\langle#1\rangle}
\newcommand{\dup}{\uparrow}
\newcommand{\dstay}{0}
\newcommand{\ddown}{\downarrow}
\newcommand{\Pri}{\mathrm{Pri}}
\newcommand{\Dir}{\mathrm{Dir}}
\newcommand{\gameonstart}[3]{\cG(#1, #2, #3)}

\newcommand{\Sigmap}{\Sigma_p}
\newcommand{\Sigmaa}{\Sigma_a}
\newcommand{\QE}{Q_{E}}
\newcommand{\QA}{Q_{A}}

\newcommand{\propsguardedneg}[1]{\kw{GNProps}(#1)} 
\newcommand{\propsneg}[1]{\kw{NProps}(#1)} 
\newcommand{\bagextend}[1]{\kw{BagLabels}(#1)} 

\newcommand{\IntNodeLabels}{\kw{InterfaceLabels}}
\newcommand{\BagNodeLabels}{\kw{BagLabels}}
\newcommand{\EdgeLabels}{\kw{Edges}}
\newcommand{\EdgeLabelsRestrict}[1]{\kw{Edges}(#1)} 
\newcommand{\NodeLabels}{\kw{NodeLabels}}
\newcommand{\tforward}[1]{ {#1}^{\rightarrow} }
\newcommand{\tbackward}[1]{ {#1}^{\leftarrow} }
\newcommand{\struct}[1]{\mathfrak{D}(#1)} 
\newcommand{\readymode}{\text{select}}
\newcommand{\waitmode}{\text{move}}
\newcommand{\views}{\kw{Views}}
\newcommand{\view}{I}

\newcommand{\strategy}{\zeta}
\newcommand{\correct}[1]{\textup{correct}(#1)} 
\newcommand{\correctr}[1]{\textup{correct}_{r}(#1)} 

\newcommand{\consistencyform}{\phi_{\textup{consistent}}}
\newcommand{\elem}[1]{\operatorname{elem}(#1)} 
\newcommand{\guardg}[2]{\operatorname{guard}^{\sigmag}_{#1}(#2)} 

\newcommand{\GuardDom}[2]{\operatorname{guard-dom}(#1,#2)} 
\newcommand{\GuardRng}[2]{\operatorname{guard-rng}(#1,#2)} 
\newcommand{\phiL}{\phi_{\L}}
\newcommand{\phiR}{\phi_{\R}}
\newcommand{\cform}[2]{\text{consistent}_{#1,#2}}
\newcommand{\sigR}{\sigma_{\R}}
\newcommand{\fG}{\mathfrak{G}}
\newcommand{\myparagraph}[1]{\subsubsection*{#1}}
\newcommand{\ltrue}{\top}
\newcommand{\lfalse}{\bot}


%% file: abstract.tex
\begin{abstract}
We look at characterizing which formulas are expressible
in rich decidable logics such as guarded fixpoint logic,
unary negation fixpoint logic, and guarded negation fixpoint logic.
We consider semantic characterizations
of definability, as well as effective characterizations. Our algorithms revolve
around a finer analysis of the tree-model property and a refinement of
the method of moving back and forth between relational logics and logics over trees.
\end{abstract}


%% file: intro.tex
\section{Introduction}%
\label{sec:intro}

A major line of research in computational logic has focused
on obtaining extremely expressive decidable logics.
The guarded fragment  ($\gf$)~\cite{gforig},  the unary negation fragment ($\unf$)~\cite{unf},
and the guarded negation fragment ($\gnf$)~\cite{gnfj} are rich decidable
fragments of first-order logic.
Each of these has extensions with a fixpoint
operator that retain decidability: $\gfp$~\cite{gfp},  $\unfp$~\cite{unf},
and $\gnfp$~\cite{gnfj} respectively. In each case the argument for satisfiability relies on
``moving to trees''. This involves showing that the logic possesses the tree-like model property: whenever there is a satisfying
model for a formula, it can be taken to be of tree-width that can be effectively computed from the formula. Such models
can be coded by trees, thus reducing satisfiability of the logic to satisfiability of a corresponding formula  over
trees, which can be decided using automata-theoretic techniques. This method has been applied
for decades (e.g.~\cite{robustly, GradelHO02}).

A  question is how to recognize formulas in these logics, and more generally how to distinguish the properties
of the formulas in one logic from another.
Clearly if we start with a formula in an undecidable logic,
such as first-order logic or least fixed point logic ($\lfp$), we have no possibility for effectively recognizing any non-trivial
property. But we could still hope for an insightful semantic characterization of
the subset that falls within the decidable logic. One well-known example
of this is van Benthem's theorem~\cite{vbbook} characterizing modal logic within first-order logic:
a first-order sentence is equivalent to a modal logic sentence exactly when it is bisimulation invariant.
For fixpoint logics, an analogous characterization is the Janin-Walukiewicz theorem~\cite{janinw},
stating that the modal $\mu$-calculus ($\Lmu$) captures the bisimulation-invariant fragment of monadic second-order logic ($\mso$).
If we start in one decidable logic and look to characterize another decidable logic,  we could also hope for
a characterization that is effective. For example, Otto~\cite{ottoeliminating} showed that if we start
with a formula of $\Lmu$, we can determine whether it can be expressed in modal logic.

In this work we will investigate both kinds of characterizations.
We will begin with $\gfp$. Gr\"adel, Hirsch, and Otto~\cite{GradelHO02}
have already provided a characterization
of $\gfp$-definability within  a very rich logic extending $\mso$ called guarded second-order logic ($\gso$). The characterization
is   exactly analogous to the van Benthem and Janin-Walukiewicz results mentioned above: $\gfp$ captures
the ``guarded bisimulation-invariant'' fragment of $\gso$. The characterization
makes use of a refinement of the method used for decidability
of these logics, which moves \emph{back and forth between relational structures and trees}:
\begin{enumerate}
\item define a \emph{forward mapping} taking a formula $\phi_0$ in the larger logic $\cL_0$ (e.g.~$\gso$ invariant under guarded bisimulation) over
relational structures to a formula $\phi'_0$ that describes the trees that code structures satisfying $\phi_0$; and
\item define a \emph{backward mapping} based on the invariance
going back to some $\phi_1$ in the restricted logic $\cL_1$ (e.g.~$\gfp$).
\end{enumerate}
The method is shown in Figure~\ref{fig:bandfsem}.

Our first main theorem is an effective version
of the above result: if we start with a
formula in certain richer
decidable fixpoint logics, such as $\gnfp$,
we can decide whether the formula is in $\gfp$. At the same time we provide a refinement
of~\cite{GradelHO02}  which accounts for two signatures, the one allowed for arbitrary relations and the one
allowed for ``guard relations'' that play a key
role in the syntax of all guarded logics.
We extend this result to deciding membership in  the ``$k$-width fragment'', $\gnfp^k$; roughly speaking this consists of formulas
built up from guarded components and positive existential formulas with at most $k$ variables. We provide a semantic
characterization of this fragment within $\gso$, as the fragment closed under the corresponding notion
of bisimulation (essentially, the $\gn^k$-bisimulation of~\cite{gnfj}). As with $\gfp$, we show that
the characterization can be made effective, provided that one starts with a formula in certain larger  decidable logics.
The proof also gives an effective characterization for the $k$-width fragment of $\unfp$.

These effective characterizations also rely on a back-and-forth method. The revised method is shown schematically
in Figure~\ref{fig:bandfeffsem}. We apply a forward mapping to
move from a formula $\phi_0$ in a larger logic $\cL_0$ (e.g.~$\gnfpk$) on relational structures to a formula $\phi'_0$ on tree encodings.
But then we can apply a \emph{different backward mapping}, tuned towards the smaller logic $\logictarget$ (e.g.~$\gfp$)
and the special properties of its tree-like models.
The backward mapping of a tree property $\phi'_0$ is always a formula $\phi_1$ in the smaller logic $\logictarget$.
But it is no longer guaranteed to be ``correct'' unconditionally---i.e.~to always characterize structures whose
codes satisfy $\phi'_0$. Still, we show that \emph{if} the original formula $\phi_0$ is definable in the
smaller logic $\logictarget$,
 then the backward mapping applied to the forward mapping   gives
such a definition. Since we can check the equivalence of two sentences in our logic effectively,
this property suffices to get decidability of definability.

\input{diagram}

The technique above has a few inefficiencies; first, it
translates  forward to sentences in a rich logic on trees, for which
analysis is non-elementary. Secondly,
it implicitly moves between relational  structures and
tree structures \emph{twice}: once to construct the formula $\phi'_0$, and a second
time to check that $\phi_0$ is equivalent to $\phi_1$, which in turn
requires first forming
a formula over trees $\phi'_1$ via a  forward mapping
and then checking its equivalence with $\phi'_0$.
 We show that in some cases we can perform an optimized version
of the process,
 allowing us to get tight bounds on the equivalence problem.

We show that our results ``restrict'' to fragments of these guarded logics, including
 their first-order fragments.
In particular, our results
give effective characterizations of $\gf$ definability.
They can be thus seen as a generalization of well-known effective characterizations
of the conjunctive existential formulas in $\gf$, the \emph{acyclic queries}. We show that
we can apply our techniques to the problem of transforming conjunctive formulas
to a well-known efficiently-evaluable form (acyclic formulas) relative
to $\gf$ theories.  These results complement previous results
on query evaluation with constraints from~\cite{acyclicmodpods, acyclicmodlics}.

This refined back-and-forth method can be tuned in a number of ways, allowing us to control
the signature as well as the sublogic. We show that this machinery can be adapted to
give an approximation of the formula $\phi_0$ within the logic $\logictarget$,
which is a kind of \emph{uniform interpolant}.


%% file: diagram.tex
\begin{figure}
\centering
\begin{subfigure}{.45\textwidth}
\centering
\begin{tikzpicture}[]
\node[align=center,font=\small\bf,text width=2cm,text badly centered] (aa) at (0,.75) {Relational \\ structures};
\node[align=center,font=\small\bf,text width=2cm,text badly centered] (bb) at (2.5,.75) {Coded \\ structures};
\node[align=center,font=\small] (a) at (0,-.25) {$\phi_0 \in \cL_0$};
\node[align=center,font=\small] (bb) at (2.5,-.25) {$\phi'_0 \in \Lmu$ };
\node[align=center,font=\small] (d) at (0,-1.25) {$\phi_1 \in \logictarget$};
\path[->] (a) edge node[font=\footnotesize,above] {(1)} (bb);
\path[->] (bb) edge node[font=\footnotesize,below] {(2)} (d);
\node[align=center,font=\small] (e) at (.5,-2) {\phantom{Test $\phi_0 \leftrightarrow \phi_1$}};
\end{tikzpicture}
\caption{Semantic Characterization}%
\label{fig:bandfsem}
\end{subfigure}
\qquad
\begin{subfigure}{.45\textwidth}
\centering
\begin{tikzpicture}[]
\node[align=center,font=\small\bf,text width=2cm,text badly centered] (aa) at (0,.75) {Relational \\ structures};
\node[align=center,font=\small\bf,text width=2cm,text badly centered] (bb) at (2.5,.75) {Coded \\ structures};
\node[align=center,font=\small] (a) at (0,-.25) {$\phi_0 \in \cL_0$};
\node[align=center,font=\small] (bb) at (2.5,-.25) {$\phi'_0 \in \Lmu$ };
\node[align=center,font=\small] (d) at (0,-1.25) {$\phi_1 \in \logictarget$};
\path[->] (a) edge node[font=\footnotesize,above] {(1)} (bb);
\path[->] (bb) edge node[font=\footnotesize,below] {(2)} (d);
\node[align=center,font=\small] (e) at (.5,-2) {Test $\phi_0 \leftrightarrow \phi_1$};
\end{tikzpicture}
\caption{Effective Characterization}%
\label{fig:bandfeffsem}
\end{subfigure}
\caption{Using forward and backward mappings for characterizations}%
\label{fig:both}
\end{figure}


%% file: related.tex
\subsection*{Related work}%
\label{sec:related}
The immediate inspiration for our work are characterizations of
definability in the guarded fragment within first-order logic~\cite{gforig},
 and characterization of definability in guarded fixpoint logic within guarded
second-order logic~\cite{GradelHO02}. Neither of these characterizations can be effective,
since the larger logics in question are too expressive.

Identifying  formulas in definable sublogics has been studied extensively
in the context of regular word and tree languages (\cite{place,placesegoufinfo2}), and
the corresponding characterizations are effective. These techniques
do not lift easily to the setting of relational languages, even those with tree-like models,
since one would require decidability over infinite trees, and the few results there (e.g.~\cite{bplaceinftree}) do not
map back to natural logics over decodings.
Although we know of no work on effectively identifying  formulas definable in a  fixpoint logic,
there are a number of works on identifying sufficient conditions for a decidable fixpoint logic
formula to be convertible into a formula without recursion (e.g.~\cite{ottoeliminating, boundedness}).

Our work is also inspired by prior automata-theoretic  approaches to uniform interpolation.
The key result here is D'Agostino and Hollenberg's~\cite{interpolationmucalc}, which
shows uniform interpolation for the modal $\mu$-calculus. We make use of this
result in our proofs. Craig interpolation for
guarded logics has been considered in the past (e.g.~\cite{effectiveinterp}), but
we know of no other work considering uniform interpolation for logics on arbitrary
arity signatures.


%% file: organ.tex
\subsection*{Organization}
Section~\ref{sec:prelims} defines the logics we study in this paper and reviews
their properties.
It also introduces tree encodings, bisimulation games, and unravelling constructions that will be the basis for several of our proofs.
It concludes with a description of automata that can operate on the tree codes.
We would encourage readers to consult this section as needed, particularly the section on automata.

Section~\ref{sec:high} presents an overview of the back-and-forth technique, and how it can be used to answer definability questions.
Section~\ref{sec:gfp} presents characterization results for $\gfp$, which provides a first example of the technique in the action.
Section~\ref{sec:gnfp} extends this technique to $\gnfpk$ and $\unfpk$.
Section~\ref{sec:interp} presents applications of the technique to interpolation, while
Section~\ref{sec:conc} gives conclusions.


%% file: prelims.tex
\section{Preliminaries}\label{sec:prelims}

We work with finite relational signatures $\sigma$.
We use $\vec{x},\vec{y},\dots$
(respectively, $\vec{X},\vec{Y},\dots$)
to denote vectors of first-order
(respectively, second-order)
variables.
For a formula $\phi$,
we write $\free{\phi}$ to denote the free first-order variables of~$\phi$,
and write $\phi(\vec{x})$ to indicate that these free variables are among $\vec{x}$.
If we want to emphasize that there are also free second-order variables $\vec{X}$,
we write $\phi(\vec{x},\vec{X})$.
We often use $\alpha$ to denote atomic formulas,
and if we write $\alpha(\vec{x})$
then we assume that the free first-order variables in $\alpha$ are precisely $\vec{x}$.
The \emph{width} of $\phi$, denoted $\width{\phi}$, is the maximum number of free variables
in any subformula of $\phi$,
and the width of a signature $\sigma$ is the maximum arity of its relations.

\subsection{Guardedness}\label{sec:guardedness}

An atomic formula $\alpha$ is a \emph{guard} for variables $\vec{x}$ if $\alpha$ uses every variable in $\vec{x}$. We say $\alpha$ is a guard for a formula $\phi$ if it is a guard for the free variables in $\phi$. This means $\free{\alpha} \supseteq \free{\phi}$. Guards can take the form $\top$ (if $\phi$ is a sentence) or the form $x=x$ (if $\phi$ has one free variable~$x$). A \emph{strict guard} for a formula $\phi$ is a guard such that the free variables of $\alpha$ are identical to the free variables in $\phi$; that is $\free{\alpha} = \free{\phi}$. For example, $Rxy$ could serve as a strict guard for $\exists z . (R y z \wedge R z x)$.

We can also talk about guardedness within a structure $\fA$.
Any set of elements of size at most 1 is considered to be both guarded and strictly guarded.
Otherwise, we say a set $U$ of elements in the domain of $\fA$ is \emph{guarded} in $\fA$ if there is some atom $\alpha(\vec{a})$ such that every element in $U$ appears in $\vec{a}$. In the special case that this atom uses precisely the elements in~$U$ and no more, then we say $U$ is \emph{strictly guarded} in $\fA$.

If we want to emphasize that the guards come from a certain signature $\sigmag$, then we will say $\sigmag$-guarded or strictly $\sigmag$-guarded.

\subsection{Basics of guarded logics}
The \emph{Guarded Negation Fragment} of \fo~\cite{gnfj} (denoted $\gnf$) is built up inductively
according to the grammar
$\phi ::= R \, \vec x ~ | ~ \exists x . \phi ~ |
~ \phi \vee \phi ~ | ~ \phi \wedge \phi ~ | ~ \alpha(\vec x) \wedge \neg \phi(\vec x)$
where $R$ is either a relation symbol or the equality relation,
and $\alpha$ is a guard for $\phi$.
If we restrict $\alpha$ to be an equality, then each negated formula can be rewritten to use
at most one free variable; this is the \emph{Unary Negation Fragment},
$\unf$~\cite{unf}.
$\gnf$ is also related to the \emph{Guarded Fragment}~\cite{gforig} ($\gf$),
typically defined via the grammar
$\phi ::= R \, \vec x ~ | ~ \exists \vec x . \big( \alpha(\vec x \vec y) \wedge \phi(\vec x \vec y) \big) ~ |
~ \phi \vee \phi ~ | ~ \phi \wedge \phi ~ | ~ \neg \phi(\vec x)$
where $R$ is either a relation symbol or the equality relation,
and $\alpha$ is a guard for $\phi$.
Here it is the quantification that is guarded, rather than negation.
\gnf subsumes \gf sentences and \unf formulas.
\gnf also subsumes \gf formulas in which the free variables are guarded.

The fixpoint extensions of these logics
(denoted $\gnfp$, $\unfp$, and $\gfp$)
extend the base logic
with formulas
$[\LFPA{X}{\vec{x}} . \alpha(\vec{x}) \wedge \phi(\vec{x},X,\vec{Y}) ](\vec{x})$
where (i) $\alpha(\vec{x})$ is a guard for $\vec{x}$,
(ii) $X$ only appears positively in~$\phi$,
(iii) second-order variables like $X$ cannot be used as guards.
Some alternative (but equi-expressive) ways to define the fixpoint extension are discussed in~\cite{BaranyBC13};
in all of the definitions, the important feature is that tuples in the fixpoint are guarded
by an atom in the original signature.
In \unfp,
there is an additional requirement that only unary or 0-ary predicates can be defined
using the fixpoint operators.
$\gnfp$ subsumes both $\gfp$ sentences and \unfp formulas.
These logics are all contained in \lfpl, the fixpoint extension of \fo.

In this work we will be interested in varying the signatures considered, and
in distinguishing more finely which relations can be used in guards.
If we want to emphasize the relational signature $\sigma$ being used,
then we will write, e.g., $\gnfp[\sigma]$.
For $\sigmag \subseteq \sigma$,
we let $\gnfp[\sigma,\sigmag]$ denote the logic built up as in $\gnfp$
but allowing only equality or relations $R \in \sigma$ at the atomic step and
only guards $\alpha$ using equality or relations $R \in \sigmag$.
We define $\gfp[\sigma,\sigmag]$ similarly.
Note that $\unfp[\sigma]$ is equivalent to $\gnfp[\sigma,\emptyset]$,
since if the only guards are equality guards,
then the formula can be rewritten to use only unary negation and monadic fixpoints.

\input{fp}

\input{expressiveness}

\myparagraph{Normal form}
It is often helpful to consider the formulas
in a normal form. \emph{Strict normal form} $\gnfp[\sigma, \sigmag]$ formulas
can be generated using the following grammar:
\begin{align*}
\phi &::= \bigvee_i \exists \vec{x}_i . \Big( \bigwedge_j \psi_{ij} \Big) \\
\psi &::= \ltrue ~ | ~ \lfalse ~ | ~ R \, \vec{x} ~ | ~ X \, \vec{x} ~ | ~ \alpha(\vec{x}) \wedge \phi(\vec{x}) ~ | ~ \alpha(\vec x) \wedge \neg \phi(\vec x) ~ | ~
[\LFPA{X}{\vec{x}} . \alpha(\vec{x}) \wedge \phi(\vec{x}, X,\vec{Y})](\vec{x})
\end{align*}
where $R$ is either a relation symbol in $\sigma$ or the equality relation, and $\alpha$ is a strict $\sigmag$-guard for $\phi$.

We will sometimes refer to a formula like $\bigvee_i \exists \vec{x}_i .  ( \bigwedge_j \psi_{ij} )$
as a \emph{UCQ-shaped formula},
and each disjunct $\exists \vec{x}_i . ( \bigwedge_j \psi_{ij} )$ as a \emph{CQ-shaped formula}.
If $\vec{x}_i$ is non-empty, then we say $\exists \vec{x}_i . ( \bigwedge_j \psi_{ij} )$ is a \emph{CQ-shaped formula with projection}.
Note that UCQ-shaped and CQ-shaped formulas generalize UCQs and CQs, respectively.
The normal form captures the notion that formulas in the logic are built
up by combining the usual guarded logic constructors and UCQ-constructors.

Every $\gnfp$ formula
can be converted to this normal form.

\begin{prop}\label{prop:nf}
Let $\theta$ be a formula in $\gnfp[\sigma,\sigmag]$.
We can construct an equivalent formula $\convertnf{\theta} \in \gnfp[\sigma,\sigmag]$ in $\nf$
such that $\size{\convertnf{\theta}} \leq 2^{f(\size{\theta})}$
and $\width{\convertnf{\theta}} \leq \size{\theta}$,
where $f$ is a polynomial function independent of $\theta$.
\end{prop}

\begin{proof}
We proceed by induction on $\theta$.
The output $\convertnf{\theta}$ is a UCQ-shaped formula in $\nf$, with the same free variables as $\theta$.
\begin{itemize}
\item If $\theta$ is atomic or is an equality, then $\convertnf{\theta} := \theta$.
\item Suppose $\theta = \alpha \wedge \neg \psi$ where $\alpha$ is a $\sigmag$-guard for $\free{\psi}$.
Then $\convertnf{\theta} := \alpha \wedge \neg (\convertnf{\alpha \wedge \psi})$.
Note that the resulting formula is strictly $\sigmag$-guarded.

\item Suppose $\theta = \exists y . (\psi)$.
If $\convertnf{\psi}$ is a UCQ-shaped formula of the form $\bigvee_{i} \exists \vec{z}_i . ( \bigwedge_j \psi_{ij} )$,
then $\convertnf{\theta} := \bigvee_i \exists y \vec{z}_i . (\bigwedge_j \psi_{ij})$.

\item Suppose $\theta = [\LFPA{Y}{\vec{y}} . \alpha(\vec{y}) \wedge \psi(\vec{y})](\vec{x})$.
Then $\convertnf{\theta} := [\LFPA{Y}{\vec{y}} . \alpha(\vec{y}) \wedge \convertnf{\psi(\vec{y})}](\vec{x})$.

\item Suppose $\theta = \psi_1 \vee \psi_2$.
Then $\convertnf{\theta}$ is the UCQ-shaped formula $\convertnf{\psi_1} \vee \convertnf{\psi_2}$.

\item Suppose $\theta = \psi_1 \wedge \psi_2$.
Assume that $\convertnf{\psi_1} = \bigvee_{i} \exists \vec{x}_i . \chi_i$
and $\convertnf{\psi_2} = \bigvee_{i'} \exists \vec{x}'_{i'} . \chi'_{i'}$.
Then
$\convertnf{\theta} := \bigvee_{i,i'} \exists \vec{y}_i \vec{y}'_{i'} . (\chi_i[\vec{y}_i / \vec{x}_i] \wedge \chi'_{i'}[\vec{y}'_{i'} / \vec{x}'_{i'}])$
where the variables in every $\vec{y}_i$ and $\vec{y}'_{i'}$ are fresh.
\end{itemize}

\noindent
It is straightforward to check that
the new formula $\convertnf{\theta}$ is of size at most
$2^{f(k)}$ for $k = \size{\theta}$ and $f$ some polynomial function independent of $\theta$.
Moreover, the number of free variable names needed in any subformula is at most $k$,
so $\width{\convertnf{\theta}} \leq k$, and hence $\convertnf{\theta} \in \gnfpk[\sigma,\sigmag]$.
\end{proof}

Later,
we will need another version of this conversion process that preserves the width,
assuming the input satisfies some additional properties
(this is not needed until the proof of Lemma~\ref{lemma:gnfp-backward}).
We say a formula starting with a block of existential quantifiers
is \emph{strictly $\sigmag$-answer-guarded} if
it is of the form $\exists \vec{y} . \big(\alpha(\vec{x}) \wedge \chi(\vec{x},\vec{y}) \big)$.

\begin{prop}\label{prop:nfwidth}
Let $\theta$ be a formula in $\gnfp[\sigma,\sigmag]$
such that any subformula starting with an existential quantifier
and not directly below another existential quantifier is strictly $\sigmag$-answer-guarded
and any negation is strictly $\sigmag$-guarded.

Then we can construct an equivalent formula $\convertnf{\theta} \in \gnfp[\sigma,\sigmag]$ that is in $\nf$
and satisfies $\size{\convertnf{\theta}} \leq 2^{f(\size{\theta})}$
and $\width{\convertnf{\theta}} = \width{\theta}$,
where $f$ is a polynomial function independent of $\theta$.
\end{prop}

\begin{proof}
We assume that each subformula in $\theta$
that starts with an existential quantifier
and is not directly below another existential quantifier is a strictly $\sigmag$-answer-guarded formula,
and every negation is strictly $\sigmag$-guarded.
We proceed by induction on the structure of the formula~$\theta$,
ensuring that the output $\convertnf{\theta}$ is a
UCQ-shaped formula in \nf,
with the same free variables as $\theta$,
and where every CQ-shaped formula with projection (i.e.~every CQ-shaped formula that uses  existential quantification) is strictly $\sigmag$-answer-guarded.

\begin{itemize}
\item If $\theta$ is atomic or is an equality, then $\convertnf{\theta} := \theta$.
\item Suppose $\theta = \alpha \wedge \neg \psi$ where $\alpha$ is a strict $\sigmag$-guard for $\free{\psi}$.
Then $\convertnf{\theta} := \alpha \wedge \neg \convertnf{\psi}$
which is strictly $\sigmag$-answer guarded since $\free{\psi} = \free{\convertnf{\psi}}$.

\item Suppose $\theta = \exists \vec{y} . ( \beta(\vec{x}) \wedge \psi )$,
a strictly $\sigmag$-answer-guarded formula starting with a block of existential quantifiers.
Let $\convertnf{\psi}$ be the UCQ-shaped formula $\bigvee_{i} \exists \vec{z}_i . (\bigwedge_j \psi_{ij})$,
and let $\vec{x}_i$ and $\vec{y}_i$ be the subset of $\vec{x}$ and $\vec{y}$ used in $\bigwedge_j \psi_{ij}$.
Then $\convertnf{\theta} := \bigvee_i \exists \vec{y}_i \vec{z}_i . (\alpha_i \wedge \bigwedge_j \psi_{ij})$
where
$\alpha_i$ is the strict $\sigmag$-guard for $\vec{x}_i$ in $\bigwedge_j \psi_{ij}$ if $\vec{z}_i$ is non-empty,
and $\alpha_i = \beta(\vec{x})$ if $\vec{z}_i$ is empty
(since we need to add a $\sigmag$-guard to ensure strict $\sigmag$-answer-guardedness for this new CQ-shaped formula with projection).
Note that this process does not increase the width.

\item Suppose $\theta = [\LFPA{Y}{\vec{y}} . \alpha(\vec{y}) \wedge \psi(\vec{y})](\vec{x})$.
Then we have $\convertnf{\theta} := [\LFPA{Y}{\vec{y}} . \alpha(\vec{y}) \wedge \convertnf{\psi(\vec{y})}](\vec{x})$.

\item Suppose $\theta = \psi_1 \vee \psi_2$.
Then $\convertnf{\theta} := \convertnf{\psi_1} \vee \convertnf{\psi_2}$.

\item Suppose $\theta = \psi_1 \wedge \psi_2$.
Let $\convertnf{\psi_1} = \bigvee_{i} \chi_i$
and $\convertnf{\psi_2} = \bigvee_{i'} \chi'_{i'}$.
Let $\alpha_i$ be the strict $\sigmag$-answer-guard for $\chi_i$ if $\chi_i$ is a CQ-shaped formula with projection, and $\top$ otherwise.
Similarly for $\alpha'_{i'}$.
Then
$\convertnf{\theta} := \bigvee_{i} \bigvee_{i'} ((\alpha_i \wedge \chi_i) \wedge (\alpha'_{i'} \wedge \chi'_{i'}))$.
The outer level UCQ now only has CQ-shaped formulas without projection of the form
$(\alpha_i \wedge \chi_i) \wedge (\alpha'_{i'} \wedge \chi'_{i'})$.\qedhere
\end{itemize}
\end{proof}

\myparagraph{Second-order logic}
\emph{Guarded second-order logic}
over a signature~$\sigma$ (denoted $\gso[\sigma]$)
is a fragment of second-order logic
in which second-order quantification is interpreted only over guarded relations,
i.e.\ over relations where every tuple in the relation is guarded by some predicate from $\sigma$.
We refer the interested reader to~\cite{GradelHO02}
for more background and some equivalent definitions of this logic.
The logics $\unfp$, $\gnfp$, and $\gfp$ can all be translated into $\gso$.

\begin{prop}\label{prop:gnfp-to-gso}
Given $\phi \in \gnfp[\sigma]$,
we can construct an equivalent $\phi' \in \gso[\sigma]$.
\end{prop}

\begin{proof}
The translation is straightforward.
The interesting case is for the least fixpoint.
If $\phi(\vec{y}) = [\LFPA{X}{\vec{x}} . \alpha(\vec{x}) \wedge \psi(X,\vec{x})](\vec{y})$
then
\[
\phi'(\vec{y}) :=
\forall X . [ (\forall \vec{x} . ((\alpha(\vec{x}) \wedge \psi'(X,\vec{x})) \rightarrow X \vec{x}))
\rightarrow X \vec{y} ]
\]
where second-order quantifiers range over guarded relations.
\end{proof}

\subsection{Transition systems and their logics}\label{sec:transitionsys}
A special kind of signature is a \emph{transition system signature} $\Sigma$ consisting of a finite set of
 unary predicates (corresponding to a set of propositions)
and binary predicates (corresponding to a set of actions).
A structure for such a signature is a \emph{transition system}.
Trees allowing both edge labels and node labels have a natural interpretation as transition systems.

We will be interested in two logics over transition system signatures.
One is \emph{monadic second-order logic} (denoted \mso)---where second-order quantification is only over unary relations.
\mso is contained in \gso, because unary relations are trivially guarded.
While  {\mso} and \gso can be interpreted  over arbitrary signatures, there are logics like \emph{modal logic}
that have syntax specific to transition system signatures.
Another is the \emph{modal $\mu$-calculus} (denoted $\Lmu$), an extension of modal logic with fixpoints.
Given a transition system signature $\Sigma$,
formulas $\phi \in \Lmu[\Sigma]$ can be generated using the grammar
$\phi ::= P ~ | ~ X ~ | ~ \phi \wedge \phi ~ | ~ \neg \phi ~ | ~  \dmodality{\rho}{\phi} ~|~ \mu X . \phi$
where $P$ is a unary relation in $\Sigma$ and $\rho$ is a binary relation in $\Sigma$.
The formulas $\mu X . \phi$ are required to use the variable $X$ only positively in $\phi$.
The meaning of a $\mu$-calculus formula can be seen by the following straightforward conversion into $\lfp$:
\begin{prop}
For every $\phi \in \Lmu[\Sigma]$ there is a formula $\psi(x) \in \lfp[\Sigma]$ such that
for every transition system $\fM$ over $\Sigma$ and every node $v$ in $\fM$,
the formula $\phi$ holds at $v$ in $\fM$ iff
$\fM, v \models \psi(x)$.
\end{prop}

\begin{proof}
We proceed by induction on the structure of $\phi$.
We let $\psi'$ denote the inductive translation of $\phi'$.

\begin{itemize}
\item If $\phi = P$, then $\psi(x) = P(x)$.
\item If $\phi = X$, then $\psi(x) = X(x)$.
\item The translation commutes with $\vee$, $\wedge$, and $\neg$.
\item If $\phi = \dmodality{\rho}{\phi'}$, then $\psi(x) = \exists y . \big( \rho(x,y) \wedge \psi'(y) \big)$.
\item If $\phi = \bmodality{\rho}{\phi'}$, then $\psi(x) = \forall y . \big(\rho(x,y) \rightarrow \psi'(y) \big)$.
\item If $\phi = \mu X . \phi'$, then $\psi(x) = [\LFPA{X}{z} . \psi'(z)](x)$.\qedhere
\end{itemize}
\end{proof}

\noindent
We say a $\Lmu$-formula $\phi$ holds from a position $v$ in a transition system $\fM$ if
$\fM, v \models \psi(x)$, where $\psi(x)$ is the $\lfp$ formula given by the previous proposition. We say that an $\Lmu$-formula holds in a tree iff it holds at the root of the tree.

It is well-known that $\Lmu$ can also be translated into {\mso}~\cite{ArnoldN01}.

\input{encode}

\input{bisimgames}

\input{automata}


%% file: fp.tex
\myparagraph{Fixpoint semantics and notation}

We briefly review the semantics of the fixpoint operator.
Take some $\alpha(\vec{x}) \wedge \phi(\vec{x},X,\vec{Y})$ where $X$ appears only positively.
Then it induces a monotone operator
$U \mapsto \mathcal{O}^{\fA,\vec{V}}_{\phi}(U) := \set{ \vec{a} : \fA, U, \vec{V} \models \alpha(\vec{a}) \wedge \phi(\vec{a},X,\vec{Y}) }$
on every structure $\fA$ with valuation $\vec{V}$ for $\vec{Y}$,
and this operator has a unique least fixpoint.

One way to obtain this least fixpoint is based on fixpoint approximants.
Given some ordinal~$\beta$,
the \emph{fixpoint approximant} $\phi^{\beta}(\fA,\vec{V})$ of $\phi$ on $\fA,\vec{V}$ is defined such that
\begin{align*}
\phi^0(\fA,\vec{V}) &:= \emptyset \\
\phi^{\beta+1}(\fA,\vec{V}) &:= \mathcal{O}^{\fA,\vec{V}}_{\phi}(\phi^{\beta}(\fA,\vec{V})) \\
\phi^{\beta}(\fA,\vec{V}) &:= \bigcup_{\beta' < \beta} \phi^{\beta'}(\fA,\vec{V})
\quad \text{where $\beta$ is a limit ordinal.}
\end{align*}
We let $\phi^{\infty}(\fA,\vec{V}) := \bigcup_{\beta} \phi^{\beta}(\fA,\vec{V})$
denote the least fixpoint based on this operation.
Thus, $[\LFPA{X}{\vec{x}} . \alpha(\vec{x}) \wedge \phi(\vec{x},X,\vec{Y})]$
defines a new predicate named $X$ of arity $\size{\vec{x}}$,
and $\fA,\vec{V},\vec{a} \models [\LFPA{X}{\vec{x}} . \alpha(\vec{x}) \wedge \phi(\vec{x},X,\vec{Y}) ](\vec{x})$
iff $\vec{a} \in \phi^{\infty}(\fA,\vec{V})$.
If $\vec{V}$ is empty or understood in context, we just write $\phi^{\infty}(\fA)$.

It is often convenient to allow \emph{simultaneous fixpoints} (also known as \emph{vectorial fixpoints}).
These are fixpoints of the form $[\LFPA{X_i}{\vec{x}_i} . S](\vec{x})$
where $S$ is a system of equations
\[
\begin{cases}
X_1, \vec{x}_1 := \alpha_1(\vec{x}_1) \wedge \phi_1(\vec{x}_1,X_1,\dots,X_j,\vec{Y}) \\
\vdots \\
X_j, \vec{x}_j := \alpha_j(\vec{x}_j) \wedge \phi_j(\vec{x}_j,X_1,\dots,X_j,\vec{Y})
\end{cases}
\]
where $X_1,\dots,X_j$ occur only positively.
Such a system can be viewed as defining a monotone operation on a vector of $j$ valuations,
where the $i$-th component in the vector is the set of tuples satisfying $X_i$ (i.e.~the $i$-th component is the valuation for $X_i$).
The formula $[\LFPA{X_i}{\vec{x}_i} . S](\vec{x})$
expresses that $\vec{x}$ is a tuple in the $i$-th component of the least fixpoint
defined by this operation.
Simultaneous fixpoints can be eliminated in favor of traditional fixpoints
using what is known as the \Bekic principle~\cite{ArnoldN01}.
This can be done using a recursive procedure that eliminates a component of the simultaneous fixpoint by in-lining this formula in the other expressions.
This in-lining process preserves any guardedness properties of the fixpoints,
so we can allow simultaneous fixpoints in $\gnfp$, $\unfp$, and $\gfp$ without changing the expressivity of these logics.

%% file: expressiveness.tex
\myparagraph{Expressivity}
These guarded fixpoint logics are expressive:
the $\mu$-calculus (see Section~\ref{sec:transitionsys}) is contained in each of these logics, and so are many description logics~\cite{dlhandbook}.
$\gf$, and hence all of the logics defined previously,
can express many standard integrity constraints on database tables, such
as \emph{inclusion dependencies}, first-order sentences of the form
$\forall \vec x.\left[ R(\vec x) \rightarrow \exists \vec y. \left( S(\vec x) \right) \right]$.
Furthermore, every positive existential formula is expressible in $\unf$ and $\gnf$.
More specifically, $\unf$ and $\gnf$ can express conjunctive queries and unions of conjunctive queries. A \emph{conjunctive query} (CQ) is a formula of the form
$
\exists x_1 \ldots x_j. ~ ( \bigwedge_{i \leq n} A_i )
$
where each $A_i$ is an atomic formula, and a \emph{Union of Conjunctive Queries} (UCQ)
 is a disjunction of CQs. A \emph{Boolean CQ} is
a CQ that is a sentence.
Further, since $\unf$ and $\gnf$ are  closed under boolean combinations of sentences,
they  can express that a Boolean CQ $Q_2$ is implied
by  another CQ $Q_1$ conjoined with a set of sentences $\Sigma$ of
the logic.
This allows satisfiability of these guarded logics to be utilized
to solve implication problems of interest in database theory, such
as the \emph{certain answer problem}. For more details on these
applications, the reader can check~\cite{bgo,bco}. The fixpoint extensions such as $\gfp$ and $\gnfp$
allow one to express many queries involving reachability; some examples
of this can be found in Section~\ref{sec:failure}.

Nevertheless, these logics are decidable and have nice model theoretic properties. In particular
satisfiability and finite satisfiability
is $\twoexptime$-complete for $\gnf$ and $\gnfp$~\cite{gnf}. The same holds for $\unfp$ and $\gfp$~\cite{unf,gfp}.


%% file: encode.tex
\subsection{Tree-like model property and tree codes}\label{sec:treecodes}

The guarded logics that we consider in this paper exhibit interesting model theoretic properties. For example, $\gnf$ (and hence $\unf$ and $\gf$) has the finite-model property~\cite{gnf}:
if $\phi$ is satisfiable, then $\phi$ is satisfiable in a finite structure.
This finite model property does not hold for the fixpoint extensions of these logics.
In this paper we will be concerned only with equivalence over all structures, not just finite structures.

Our work relies heavily on a different model theoretic property, called the tree-like model property. We review now what it means for a relational structure to be ``tree-like''.
Roughly speaking, these are structures that can be decomposed into a tree form.
Formally, a \emph{tree decomposition} of a structure $\fM$ consists of a tree $(V,E)$ and a function $\lambda$
assigning to each vertex $v \in V$ a subset $\lambda(v)$ of elements in the domain of $\fM$,
so that the following hold:
\begin{itemize}
\item
For each atom $R c_1 \ldots c_n$ that holds in $\fM$, there is a $v$ such that $\lambda(v)$
includes each element of $c_1 \ldots c_n$.
\item For each domain element $e$ in the domain of $\fM$, the set of nodes
\[
\{ v \in V : e \in \lambda(v) \}
\]
is a connected subset of the tree. In other words, for any two vertices $v_1, v_2$ such that $e \in \lambda(v_1)$ and $e \in \lambda(v_2)$, there is a path between $v_1$ and $v_2$ such that $e \in \lambda(u)$ for every node $u$ on this path.
\end{itemize}
The \emph{width} of a decomposition is one less than the maximum size of $\lambda(v)$
over any element $v \in V$. The subsets $\lambda(v)$ of $\fM$ are called \emph{bags}
of the decomposition, so structures of tree-width $k-1$ have bags of size at most $k$.

$\gnfp$ (and hence $\unfp$ and $\gfp$) has the \emph{tree-like model property}~\cite{gnf}:
if $\phi$ is satisfiable, then $\phi$ is satisfiable
over structures with tree decompositions of some bounded tree-width. In fact satisfiable $\gnfpk$ formulas have
satisfying structures of tree-width $k-1$. Satisfiable $\gfp$ sentences have an even stronger property: each bag in the decomposition describes a guarded set of elements, so the width of the tree decomposition is bounded by the maximum arity of the relations.

It is well-known that structures of tree-width $k-1$ can be encoded by
labelled trees over an alphabet that depends only on the signature $\sigma$ of the structure
and $k$.
Our encoding scheme will make use
of trees with both node and edge labels, i.e.~trees over a transition system signature
$\sigcode{\sigma}{k}$.
Each node in a tree code represents atomic information over at most $k$ elements,
so the signature $\sigcode{\sigma}{k}$ includes unary predicates to indicate
the number of elements represented at that node,
and the atomic relations that hold of those elements.
The signature includes binary predicates that
indicate the overlap and relationship
between the names of elements encoded in neighboring nodes of the tree.
Formally,  $\sigcode{\sigma}{k}$ contains the following relations:
\begin{itemize}
\item There are unary relations $D_n \in \sigcode{\sigma}{k}$ for $n \in \set{0,\dots,k}$,
to indicate the number of elements represented at each node. We call these \emph{domain predicates} since they are used to specify the number of domain elements encoded at a given node.
\item For every relation $R \in \sigma$ of arity $n$
and every sequence $\vec{i} = i_1 \dots i_n$ over $\set{1,\dots,k}$,
there is a unary relation $R_{\vec{i}} \in \sigcode{\sigma}{k}$
to indicate that the tuple of elements coded by $\vec{i}$ is a tuple of elements in $R$.
For example, if $T$ is a ternary relation in $\sigma$
and $a_i$ is the element coded by name $i$ in some node,
then
$T_{3,1,3}$ indicates that $T(a_3,a_1,a_3)$ holds.
\item For every partial 1--1 map $\rho$ from $\set{1,\dots,k}$ to $\set{1,\dots,k}$,
there is a binary relation $E_{\rho} \in \sigcode{\sigma}{k}$ to indicate the relationship
between the names of elements in neighboring nodes.
For example,
if $(u,v) \in E_{\rho}$ and $\rho(3) = 1$, then
the element with name $3$ in $u$ is the same as the element with name $1$ in $v$.
\end{itemize}

\noindent
For a unary relation $R_{\vec{i}}$, we write $\indices{R_{\vec{i}}}$
to denote the set of elements from $\set{1,\dots,k}$ appearing in $\vec{i}$.
We will refer to the elements of $\set{1, \ldots , k}$
as \emph{indices} or \emph{names}.

For nodes $u,v$ in a $\sigcode{\sigma}{k}$-tree $\tree$ and names $i,j$,
we will say $(u,i)$ is \emph{equivalent} to $(v,j)$ if
there is a simple undirected path $u = u_1 u_2 \dots u_n = v$ in $\tree$,
and $\rho_1, \dots, \rho_{n-1}$ such that
$(u_i,u_{i+1}) \in E_{\rho_i}^{\cT}$ or $(u_{i+1},u_i) \in E_{\rho_i^{-1}}^{\cT}$,
and $(\rho_{n-1} \circ \dots \circ \rho_1)(i) = j$.
In words, the $i$-th element in node $u$
corresponds to the $j$-th element in node $v$,
based on the composition of edge labels (or their inverses) on the simple path between $u$ and $v$.
We write $[u,i]$ for the equivalence class based on this equivalence relation.

\label{encodedguardedness} 
Given some subsignature $\sigmag \subseteq \sigma$
and some set of indices $I \subseteq \set{1,\dots,k}$,
we say that $R_{\vec{i}} \in \sigcode{\sigma}{k}$ is a \emph{$\sigmag$-guard} for $I$ if
$\indices{R_{\vec{i}}} \supseteq I$ and $R \in \sigmag$.
Likewise, $R_{\vec{i}} \in \sigcode{\sigma}{k}$ is a \emph{strict $\sigmag$-guard} for $I$ if
$\indices{R_{\vec{i}}} = I$ and $R \in \sigmag$.
Given a set $\tau$ of unary relations from $\sigcode{\sigma}{k}$
we say \emph{$I$ is $\sigmag$-guarded in $\tau$}
if $\size{I} \leq 1$ or there is some $R_{\vec{i}} \in \tau$ that is a $\sigmag$-guard for $I$.
Similarly, we say \emph{$I$ is strictly $\sigmag$-guarded in $\tau$}
if $\size{I} \leq 1$ or there is some $R_{\vec{i}} \in \tau$ that is a strict $\sigmag$-guard for $I$.
These definitions are analogous to the definitions of guardedness and strict guardedness that were given in Section~\ref{sec:guardedness}, but adapted to encodings. For example, if $I$ is strictly $\sigmag$-guarded in $\tau$, then the set $\tau$ is encoding some relation that would strictly guard the elements encoded by $I$.

Given some $\sigcode{\sigma}{k}$-tree $\cT$,
we say $\cT$ is \emph{consistent}
if it satisfies certain natural conditions
that ensure that the tree actually corresponds to a code of some tree decomposition
of a $\sigma$-structure:
\begin{enumerate}
\item there is exactly one domain predicate $D_i$
that holds at each node, and the root $v_0$ is in $D_0^{\cT}$;
\item edge labels respect the domain predicates:
if $u \in D_m^{\cT}$, $v \in D_n^{\cT}$, and $(u,v) \in E_\rho^{\cT}$,
then $\dom{\rho} \subseteq \set{1,\dots,m}$
and $\codom{\rho} \subseteq \set{1,\dots,n}$;
\item node labels respect the domain predicates:
if $v \in D_n^{\cT}$ and $v \in R_{\vec{i}}^{\cT}$,
then $\indices{R_{\vec{i}}} \subseteq \set{ 1, \dots, n}$;
\item neighboring node labels agree on shared names:
if $u \in R_{\vec{i}}^{\cT}$, $(u,v) \in E_\rho^{\cT}$,
and $\indices{R_{\vec{i}}} \subseteq \dom{\rho}$,
then $v \in R_{\rho(\vec{i})}^{\cT}$;
similarly, if $v \in R_{\vec{i}}^{\cT}$, $(u,v) \in E_\rho^{\cT}$,
and $\indices{R_{\vec{i}}} \subseteq \codom{\rho}$,
then $u \in R_{\rho^{-1}(\vec{i})}^{\cT}$;
\end{enumerate}
where $P^{\cT}$ denotes the interpretation of relation $P$ in $\cT$.

It is now easy to verify the fact mentioned at the beginning of this subsection: tree decompositions
of every $\sigma$-structure of tree-width $k-1$
can be encoded in consistent $\sigcode{\sigma}{k}$-trees.

The next step is to describe how
a consistent $\sigcode{\sigma}{k}$-tree can be decoded to an actual $\sigma$-structure.
The \emph{decoding} of $\cT$ is the
$\sigma$-structure $\struct{\cT}$
where the universe is the set
\[\set{ [v,i] : \text{$v \in \dom{\cT}$ and $i \in \set{1,\dots,k}$} }\]
and a tuple $([v_1,i_1],\dots,[v_r,i_r])$ is in
$R^{\struct{\cT}}$
iff there is some node $w \in \dom{\cT}$ such that $w \in R_{j_1 \dots j_r}$
and $[w,j_m] = [v_m,i_m]$ for all $m \in \set{1,\dots,r}$.

Finally, we introduce some notation related to $\sigcode{\sigma}{k}$.
We often use $\tau$ to denote a node label, and $\tau(v)$ to denote the label at some node $v$ in a tree.
We write $\EdgeLabels$ for the set of functions $\rho$ such that the binary predicate $E_\rho$ is in $\sigcode{\sigma}{k}$.
We write $\NodeLabels$ for the set of \emph{internally consistent} node labels,
i.e.~the set consisting of sets of unary predicates from $\sigcode{\sigma}{k}$
that satisfy properties (1) and (3) in the definition of consistency above.


%% file: bisimgames.tex
\subsection{Bisimulations and unravellings}\label{sec:bisim}
The logic $\Lmu$ over transition system signatures
lies within {\mso}. Similarly the guarded logics $\gfp$, $\unfp$, and $\gnfp$
all lie within \gso and apply to
arbitrary-arity signatures.
It is easy to
see that these containments are proper.
 In each case, what distinguishes the smaller logic from the larger
is \emph{invariance} under certain equivalences called
\emph{bisimulations}, each of which is defined by a certain player
having a winning strategy in a two-player infinite game played between players
Spoiler and Duplicator.

For $\Lmu$, the appropriate game is the classical \emph{bisimulation game} between transition systems $\fA$ and $\fB$: the definition of the game and the basic results about it can be found  in~\cite{freedom}.
It is straightforward to check that $\Lmu[\Sigma]$-formulas
are $\Sigma$-bisimulation-invariant, i.e. $\Lmu[\Sigma]$-formulas cannot distinguish between $\Sigma$-bisimilar transition systems.
We will make use of a stronger result of  Janin and Walukiewicz~\cite{janinw} that
the $\mu$-calculus is the bisimulation-invariant fragment of {\mso}
(we state it here for trees because of how we use this later):
A class of trees is definable in $\Lmu[\Sigma]$ iff it is definable in $\mso[\Sigma]$ and
closed under $\Sigma$-bisimulation within the class of all $\Sigma$-trees.
The proof of this result is effective in the following sense:
 given an $\mso[\Sigma]$ sentence $\phi$ it is possible to
construct a $\mu$-calculus formula $\phi'$ such that,
if $\phi$ is bisimulation-invariant, $\phi'$ holds from the root of a tree  $\tree$
iff $\tree$ satisfies $\phi$.

We now describe a generalization of these games
between structures $\fA$ and $\fB$ over a signature $\sigoriginal$ with arbitrary arity relations, parameterized
by some subsignature $\sigtarget$ of the structures.
Each position in the game is a partial $\sigtarget$ homomorphism $h$ from $\fA$ to $\fB$, or vice versa.
The \emph{active structure} in position $h$ is the structure containing the domain of $h$.
The game starts from the empty partial map from $\fA$ to $\fB$.
In each round of the game, Spoiler chooses between one of the following moves:
\begin{itemize}
\item \textit{Extend:}
Spoiler chooses some set $X$ of elements in the active structure such that $X \supseteq \dom{h}$,
and Duplicator must then choose $h'$ extending $h$ (i.e.\ such that $h(c) = h'(c)$ for all $c \in \dom{h}$)
such that $h'$ is a partial $\sigtarget$ homomorphism;
Duplicator loses if this is not possible.
Otherwise, the game proceeds from the position $h'$.
\item \textit{Switch:}
Spoiler chooses to switch active structure.
If $h$ is not a partial $\sigtarget$ isomorphism, then Duplicator loses.
Otherwise, the game proceeds from the position $h^{-1}$.
\item \textit{Collapse:}
Spoiler selects some $X \subseteq \dom{h}$
and the game continues from position $\restrict{h}{X}$.
\end{itemize}
Duplicator wins if she can continue to play indefinitely.

We will consider several variants of this game. For $k \in \N$ and $\sigmag \subseteq \sigtarget$:
\begin{itemize}
\item {\bf $k$-width guarded negation bisimulation game}:
The $\gnkinvar[\sigtarget,\sigmag]$-game
is the version of the game where the domain of every position $h$ is of size at most $k$,
and Spoiler can only make a switch move at $h$ if $\dom{h}$ is strictly $\sigmag$-guarded
in the active structure.
\item {\bf block $k$-width guarded negation bisimulation game}:
The $\sgnkinvar[\sigtarget,\sigmag]$-game
is like the $\gnkinvar[\sigtarget,\sigmag]$-game, but additionally
Spoiler is required to alternate between extend/switch moves and moves where he collapses to a strictly $\sigmag$-guarded set.
We call it the ``block'' game since Spoiler must select all of the new extension elements in a single block,
rather than as a series of small extensions.
The key property is that the game alternates between positions with
a strictly $\sigmag$-guarded domain, and positions of size at most~$k$.
The restriction mimics the alternation between formulas of width $k$
and strictly $\sigmag$-guarded formulas within normalized $\gnfpk$ formulas.
\item {\bf guarded bisimulation game}:
The $\ginvar[\sigtarget,\sigmag]$-game
is the version of the game
where the domain of every position must be strictly $\sigmag$-guarded in the active structure.
Note that in such a game, every position $h$ satisfies $\size{\dom{h}} \leq \width{\sigma_g}$.
\end{itemize}
We say $\fA$ and $\fB$ are \emph{$\gnkinvar[\sigtarget,\sigmag]$-bisimilar}
if Duplicator has a winning strategy in the $\gnkinvar[\sigtarget,\sigmag]$-game
starting from the empty position.
We say a sentence $\phi$ is \emph{$\gnkinvar[\sigtarget,\sigmag]$-invariant}
if for any pair of $\gnkinvar[\sigtarget,\sigmag]$-bisimilar $\sigtarget$-structures,
$\fA \models \phi$ iff $\fB \models \phi$.
A logic $\cL$ is $\gnkinvar[\sigtarget,\sigmag]$-invariant if
every sentence in $\cL$ is $\gnkinvar[\sigtarget,\sigmag]$-invariant.
When the guard signature is the entire signature,
we will write, e.g., $\gnkinvar[\sigtarget]$ instead of
$\gnkinvar[\sigtarget,\sigtarget]$. We similarly talk about \emph{$\ginvar[\sigtarget,\sigmag]$-invariance}
where we replace  the $\gnkinvar[\sigtarget,\sigmag]$-game by the $\ginvar[\sigtarget,\sigmag]$-game.

It is known that the  bisimulation games characterize certain fragments of $\fo$:
$\gf[\sigtarget]$ is the $\ginvar[\sigtarget]$-invariant fragment of $\fo[\sigtarget]$~\textup{\cite{gforig}}
and $\gnfk[\sigtarget]$ can be characterized as either the $\sgnkinvar[\sigtarget]$-invariant or
the $\gnkinvar[\sigtarget]$-invariant fragment of $\fo[\sigtarget]$~\textup{\cite{gnf}}.
Likewise, for fixpoint logics and fragments of $\gso$,
$\gfp[\sigtarget]$ is the $\ginvar[\sigtarget]$-invariant fragment of $\gso[\sigtarget]$~\textup{\cite{GradelHO02}}, while
$\unfpk[\sigtarget]$ is the $\sgnkinvar[\sigtarget,\emptyset]$-invariant fragment of $\gso[\sigtarget]$ \textup{\cite{lics15-gnfpi}}.

In this paper,
we will prove a corresponding characterization for $\gnfpk[\sigtarget]$
in terms of $\sgnkinvar[\sigtarget]$-invariance:
$\gnfpk[\sigtarget]$ is the $\sgnkinvar[\sigtarget]$-invariant fragment of $\gso[\sigtarget]$
(in fact we will refine this to also talk about the guard signature; see Theorem~\ref{thm:gnfpk-characterization}).
Note that for fixpoint logics, $\gnkinvar[\sigtarget]$-invariance is strictly weaker than
$\sgnkinvar[\sigtarget]$-invariance. For example,~\cite{gnfpup} gives another decidable logic within $\gso$ which is $\gnkinvar[\sigtarget]$-invariant but not $\sgnkinvar[\sigtarget]$ invariant.

\myparagraph{Unravellings}\label{unravellings}
Given a $\sigoriginal$-structure $\fA$ and $k \in \N$
and $\sigmag \subseteq \sigtarget \subseteq \sigoriginal$,
we would like to construct a structure
that is $\gnkinvar[\sigtarget,\sigmag]$-bisimilar to $\fA$
but has a tree-decomposition of bounded tree-width.
A standard construction achieves this,
called the $\gnkinvar[\sigtarget,\sigmag]$-unravelling of $\fA$.
Let $\Pi_k$ be the set of finite sequences
of the form $Y_0 Y_1 \dots Y_m$
such that
$Y_0 = \emptyset$ and
$Y_i$ is a set of elements from $\fA$ of size at most $k$.
Each such sequence can be seen as the projection to $\fA$
of a play in the $\gnkinvar[\sigtarget,\sigmag]$-bisimulation game
between $\fA$ and some other structure.
For $Y$ a set of elements from $\fA$,
let $\atypesig{Y}{\fA}{\sigtarget}$
be the set of atoms that hold of the elements in $Y$:
$\set{ R(a_1,\dots,a_l) : \text{$R \in \sigtarget$, $\set{a_1,\dots, a_l} \subseteq Y$, $\fA \models R(a_1,\dots, a_l)$}}$.
Now define a $\sigcode{\sigtarget}{k}$-tree $\unravelk{\fA}{\sigtarget,\sigmag}$
where each node corresponds to a sequence in $\Pi_k$,
and the sequences are arranged in prefix order.
The node label of every $v = Y_0 \dots Y_{m-1} Y_m$ is an encoding of $\atypesig{Y_m}{\fA}{\sigtarget}$,
and the edge label between its parent $u$ and $v$
indicates the relationship between the shared elements $Y_{m-1} \cap Y_{m}$ encoded in $u$ and $v$.
We define $\decode{\unravelk{\fA}{\sigtarget,\sigmag}}$ to be the $\gnkinvar[\sigtarget,\sigmag]$-unravelling of $\fA$.
By restricting the set $\Pi_k$ to reflect the possible moves in the games,
we can define unravellings based on the other bisimulation games in a similar fashion.
We summarize the  two unravellings that will be most relevant later on:
\begin{itemize}
\item {\bf block $k$-width guarded negation unravelling}:
The $\sgnkinvar[\sigtarget,\sigmag]$-unravelling is denoted $\decode{\bunravelk{\fA}{\sigtarget,\sigmag}}$.
Its encoding $\bunravelk{\fA}{\sigtarget,\sigmag}$
is obtained by considering only sequences $Y_0 \dots Y_m \in \Pi_k$
such that
for all \emph{even} $i$, $Y_{i-1} \supseteq Y_i \subseteq Y_{i+1}$ and $Y_i$ is strictly $\sigmag$-guarded in $\fA$.
The tree $\bunravelk{\fA}{\sigtarget,\sigmag}$ is consistent and is called a \emph{$\sigmag$-guarded-interface tree}\label{guardedinterfacetree} since
it alternates between \emph{interface nodes}\label{interfacenodes} with strictly~$\sigmag$-\nolinebreak guarded domains---these correspond to collapse moves in the game---and
\emph{bag nodes} with domain of size at most $k$ that are not necessarily $\sigmag$-guarded.
\item {\bf guarded unravelling}:
The $\ginvar[\sigtarget,\sigmag]$-unravelling is denoted $\decode{\gunravel{\fA}{\sigtarget,\sigmag}}$
and its encoding $\gunravel{\fA}{\sigtarget,\sigmag}$
is obtained by considering only sequences $Y_0 \dots Y_m \in \Pi_k$
such that for all $i$,
$Y_i$ is strictly $\sigmag$-guarded in $\fA$.
The tree $\gunravel{\fA}{\sigtarget,\sigmag}$ is consistent and is called a \emph{$\sigmag$-guarded} tree
since the domain of every node in the tree is strictly $\sigmag$-guarded.
\end{itemize}

\noindent
It is straightforward to check that:
\begin{prop}%
\label{prop:unravelling}
Let $\fA$ be a $\sigma$-structure, and let $k \in \N$ and $\sigmag \subseteq \sigtarget \subseteq \sigma$. Then
\begin{itemize}
\item $\fA$ is $\sgnkinvar[\sigtarget,\sigmag]$-bisimilar to $\decode{\bunravelk{\fA}{\sigtarget,\sigmag}}$, the block $k$-width guarded negation unravelling;
\item $\fA$ is $\ginvar[\sigtarget,\sigmag]$-bisimilar to the guarded unravelling $\decode{\gunravel{\fA}{\sigtarget,\sigmag}}$.
\end{itemize}
\end{prop}
Because these unravellings have tree codes of some bounded tree-width,
this implies that these guarded logics have tree-like models.
The structural differences in the tree decompositions
will be exploited for our definability decision procedures.

Note that it is also possible (and more standard) to define the bisimulation games
and unravellings by replacing every occurrence of ``strictly guarded'' by ``guarded''.
The games would still preserve the corresponding logic, and the analog of Proposition~\ref{prop:unravelling}
would still hold. Further, our forward mapping results, saying that we can translate
 a formula in the logic into a formula running over the encoding of the unravelling
(see e.g. Lemma~\ref{lemma:forward-gso}), would still hold. Our use of strict guards will come into
play only in simplifying  the definition of the backward mappings (e.g. Lemma~\ref{lemma:backwards-gnfp}).


%% file: automata.tex
\subsection{Automata}\label{sec:automata}

We will make use of automata on trees
for the optimized decision procedures in Section~\ref{sec:gfp}.
We suggest that readers skip this section until it is needed there.

Our goal in this section is to define two automaton models
that can function on trees that have unbounded (possibly infinite) branching degree.
This is because the tree codes derived from the unravellings described earlier may have this unbounded branching.
We describe these automata below, but will assume familiarity with standard automata theory
over infinite structures (see, e.g.,~\cite{Thomas97}).

Fix a transition system signature $\Sigma$ consisting of
unary relations $\Sigmap$ and binary relations $\Sigmaa$
(for the node labels and edge labels, respectively).

A \emph{2-way alternating $\mu$-automaton} $\cA$
is a tuple $\tuple{\Sigma,\QE,\QA,q_0,\delta,\Omega}$
where
$Q := \QE \cup \QA$ is a finite set of states partitioned into states $\QE$
controlled by Eve and states $\QA$ controlled by Adam,
and $q_0 \in Q$ is the initial state.
The transition function has the form
\[
\delta : Q \times \powerset{\Sigmap} \to \powerset{ \Dir \times \Sigmaa \times Q }
\]
where $\Dir = \set{\dup,\dstay,\ddown}$ is the set of possible directions
(up $\dup$, stay $\dstay$, down $\ddown$).
The acceptance condition is a parity condition specified by $\Omega : Q \to \Pri$,
which maps each state to a priority in a finite set of priorities $\Pri$.

Let $\tree$ be a tree over $\Sigma$, and let $\tree(v)$ denote the set of unary propositions in $\Sigmap$ that hold at $v$.

The notion of acceptance of $\tree$ by $\cA$ starting at node $v_0 \in \dom{\cT}$
is defined in terms of a game $\gameonstart{\cA}{\cT}{v_0}$.
The arena is $Q \times \dom{\cT}$,
and the initial position is $(q_0,v_0)$.
From a position $(q,v)$ with $q \in \QE$ (respectively, $q \in \QA$),
Eve (respectively Adam) selects $(d,a,r) \in \delta(q,\cT(v))$,
and an $a$-neighbor $w$ of $v$ in direction $d$
(note if $d = \dstay$, then $v$ is considered the only option, and we sometimes write just $(\dstay,r)$ instead of $(\dstay,a,r)$).
The game continues from position $(r,w)$.

A \emph{play} in $\gameonstart{\cA}{\cT}{v_0}$
is a sequence $(q_0,v_0), (q_1,v_1), (q_2,v_2),\dots$
of moves in the game.
Such a \emph{play is winning} for Eve if the \emph{parity condition} is satisfied:
the maximum priority that occurs infinitely often in $\Omega(q_0) , \Omega(q_1), \dots$
is even.

A \emph{strategy} for one of the players is a function that
returns the next choice for that player given the history of the play.
If the function depends only on the current position (rather than the full history),
then it is \emph{positional}.
Choosing a strategy for both players fixes a play in $\gameonstart{\cA}{\cT}{v_0}$.
A play $\pi$ is \emph{compatible} with a strategy $\strategy$
if there is a strategy for the other player such that $\strategy$ and $\strategy'$
yield $\pi$.
A \emph{strategy is winning} for Eve if every play compatible with it is winning.

We write $L_{v_0}(\cA)$ for the set of trees $\cT$
such that Eve has a winning strategy in $\gameonstart{\cA}{\cT}{v_0}$.
If $v_0$ is the root of $\cT$,
then we just write $L(\cA)$ to denote the \emph{language} of $\cA$.

The \emph{dual} of a 2-way alternating $\mu$-automaton $\cA$
is the automaton $\cA'$ obtained from $\cA$ by
switching $\QA$ and $\QE$,
and incrementing each priority by 1
(i.e. $\Omega'(q) := \Omega(q) + 1$).
This has the effect of switching the roles of the two players,
so the resulting automaton accepts the complement of $L(\cA)$.

These 2-way alternating $\mu$-automata
are essentially the same as the automata used in~\cite{gfp};
we use slightly different notation here and allow directions stay, up, and down,
rather than just stay and `move to neighbor'.

We are also interested in a type of automaton
on trees with arbitrary branching
that operates in a 1-way, nondeterministic fashion.
These automata were introduced by Janin-Walukiewicz~\cite{janinw,JaninW96};
we follow the presentation given in~\cite{interpolationmucalc}.
A \emph{$\mu$-automaton} $\cM$ is a tuple
$\tuple{\Sigmap,\Sigmaa,Q,q_0,\delta,\Omega}$.
where the transition function now has the form
\[
\delta : Q \times \powerset{\Sigmap} \to \powerset{\powerset{\Sigmaa \times Q}} .
\]
Again, the acceptance condition is a parity condition
specified by $\Omega$.
As before, we define acceptance of $\cT$ from a node $v_0 \in \dom{\cT}$
based on a game $\gameonstart{\cA}{\cT}{v_0}$.
The arena is $Q \times \dom{\cT}$,
and the initial position is $(q_0,v_0)$.
From a position $(q,v)$,
Eve selects some $S \in \delta(q,\cT(v))$,
and a marking of every successor of $v$ with a set of states
such that (i) for all $(a,r) \in S$, there is some $a$-successor whose marking includes $r$,
and (ii) for all $a$-successors $w$ of $v$, if $r$ is in the marking of $w$,
then there is some $(a,r) \in S$.
Adam then selects some successor $w$ of $v$
and a state $r$ in the marking of $w$ chosen by Eve,
and the game continues from position $(r,w)$.
A winning play and strategy is defined as above.

\myparagraph{Properties of $\mu$-automata}

These automata
are bisimulation invariant on trees.

\begin{propC}[\cite{janinw}]\label{prop:mu-automata-bisim}
Let $\cA$ be a 2-way alternating $\mu$-automaton or a $\mu$-automaton.
For all trees $\cT$,
if $\cT \in L(\cA)$ and $\cT'$ is bisimilar to $\cT$,
then $\cT' \in L(\cA)$.
\end{propC}

These automata models also have nice closure properties.

\begin{prop}\label{prop:2waymuaut-closure}
2-way alternating $\mu$-automata are closed under:
\begin{itemize}
\item Intersection:
Let $\cA_1$ and $\cA_2$ be 2-way alternating $\mu$-automata.
Then we can construct a 2-way alternating $\mu$-automaton $\cA$
such that
$L(\cA) = L(\cA_1) \cap L(\cA_2)$,
and the size of $\cA$ is linear in $\size{\cA_1} + \size{\cA_2}$.
\item Union:
Let $\cA_1$ and $\cA_2$ be 2-way alternating $\mu$-automata.
Then we can construct a 2-way alternating $\mu$-automaton $\cA$
such that
$L(\cA) = L(\cA_1) \cup L(\cA_2)$,
and the size of $\cA$ is linear in $\size{\cA_1} + \size{\cA_2}$.
\item Complement:
Let $\cA$ be a 2-way alternating $\mu$-automaton.
Then we can construct a 2-way alternating $\mu$-automaton $\cA'$
of size at most $\size{\cA}$ such that $L(\cA')$ is the complement of $L(\cA)$.
\end{itemize}
\end{prop}

\begin{proof}
These are standard constructions for alternating automata.

For intersection,
we can just take the disjoint union of the two automata,
and create a new initial state $q_0$ controlled by Adam
with moves to stay in the same position and go to state $q_0^{\cA_1}$,
or stay in the same position and go to state $q_0^{\cA_2}$.
Depending on this initial choice, the automaton then simulates either $\cA_1$ or $\cA_2$.
The construction for the union is similar,
but the initial choice is given to Eve, rather than Adam.

For complement, we use the dual automaton,
which requires
switching $\QA$ and $\QE$, and incrementing the priority mapping by 1.
\end{proof}

It is straightforward to construct a 2-way alternating $\mu$-automaton that is equivalent to a given $\mu$-automaton.
Moreover, it is known that
$\mu$-automata, 2-way alternating $\mu$-automata and the $\mu$-calculus
are equivalent over trees (this follows from~\cite{janinw}).

\begin{thmC}[\cite{janinw}]\label{thm:muaut}
Given $\phi \in \Lmu[\Sigma]$,
we can construct a $\mu$-automaton
$\cA$ such that $L(\cA)$ is the set of $\Sigma$-trees such that $\tree \models \phi$.

Likewise, given a $\mu$-automaton or 2-way alternating $\mu$-automaton $\cA$ over signature $\Sigma$,
we can construct $\phi \in \Lmu[\Sigma]$
such that $L(\cA)$ is the set of $\Sigma$-trees such that $\tree \models \phi$.
\end{thmC}


%% file: highlevel.tex
\section{Decidability via back-and-forth method and equivalence}%
\label{sec:high}

We now describe the main components of our approach, and explain how they fit together.

\subsection{Forward mapping}
The first component is a \emph{forward mapping},
translating an input $\gso$ formula $\phi$
to a formula over tree codes, holding on
precisely the codes that correspond to tree-like models of~$\phi$.
We start in the most general way we can,
with $\gso$ formulas $\phi$ that are $\gnlinvar$-invariant for some~$l$.
In this case, we can define a forward mapping that produces a $\mu$-calculus
formula that holds in a tree code iff $\phi$ holds in its decoding.

\begin{lem}[Fwd, adapted from~\cite{GradelHO02}]\label{lemma:forward-gso}
Given a $\gnlinvar[\sigoriginal]$-invariant sentence $\phi \in \gso[\sigoriginal]$
and given some $n \geq \max\set{\width{\sigoriginal},l}$,
we can construct
$\phi^\mu \in \Lmu[\sigcode{\sigoriginal}{n}]$ such that
for all consistent $\sigcode{\sigoriginal}{n}$-trees $\tree$,
$\tree \models \phi^\mu$ iff
$\decode{\tree} \models \phi$.
\end{lem}

Note that when moving to trees, we must specify the size of the bags (i.e.~the tree width of the corresponding tree decompositions). The $n$ in the lemma can be seen as the desired size of the bags in the tree codes.
To prove Lemma~\ref{lemma:forward-gso}, we use an inductive translation that produces a formula in $\mso$,
and then apply the
Janin-Walukiewicz Theorem~\cite{janinw} to convert this to the required formula in $\Lmu$.
 Applying this conversion from $\mso$ to $\Lmu$
 requires that the trees are (at least) $\sigcode{\sigoriginal}{l}$-bisimilar, which is why we must use $n \geq \max\set{\width{\sigoriginal},l}$ for the size of the bags in the tree codes.

This inductive translation must deal with formulas with free variables,
and hence must use codes that include valuations for these variables.
A valuation for a first-order variable $x$
can be encoded by a valuation of second-order variables
$\tforward{x} = {(Z^x_i)}_{i \in \set{1,\dots,n}}$.
The set $Z^x_i$ consists of the nodes~$v$ in the tree code
where the $i$-th element in $v$ corresponds to the element
identified by $x$.
Likewise, a valuation for a second-order variable~$X$
corresponding to an $r$-ary $\sigoriginal$-guarded relation
(a relation that only includes tuples guarded in~$\sigoriginal$)
can be encoded by a sequence of second-order variables
$\tforward{X} = {(Z^X_{\vec{i}})}_{\vec{i} \in \set{1,\dots,n}^r}$.
The set $Z^{X}_{\vec{i}}$ consists of the nodes $v$ in the tree code
where the tuple of elements coded by $\vec{i}$ in $v$ are in the valuation for $X$.

It is straightforward to construct the following auxiliary formulas
that check whether a tree is consistent,
and whether some tuple of second-order variables
actually encodes a valuation for a first-order variable or
a $\sigoriginal$-guarded relation
in the way we have just described.

\begin{lem}\label{lemma:correct}
Given $\sigma$ and $n$, we can construct the following $\mso[\sigcode{\sigoriginal}{n}]$ formulas:
\begin{itemize}
\item a formula $\consistencyform$ such that for all
$\sigcode{\sigoriginal}{n}$-trees $\tree$,
$\tree \models \consistencyform$
iff
$\tree$ is a consistent $\sigcode{\sigoriginal}{n}$-tree.

\item a formula $\correct{\tforward{x}}$ such that
for all consistent $\sigcode{\sigoriginal}{n}$-trees $\tree$
and for all $\tforward{j} = {(J_i)}_{i \in \set{1,\dots,n}}$,
$\tree \models \correct{\tforward{j}}$
iff
there is some element $a$ in $\decode{\tree}$ such that
for all $i$, we have $J_{i} = \set{ v \in \tree : [v, i] = a}$.

\item  a formula $\correctr{\tforward{X}}$ such that
for all consistent $\sigcode{\sigoriginal}{n}$-trees $\tree$
and for all $\tforward{J} = {(J_{\vec{i}})}_{\vec{i} \in \set{1,\dots,n}^r}$,
$\tree \models \correctr{\tforward{J}}$
iff
there is some $\sigoriginal$-guarded relation~$J$ of arity $r$ on $\decode{\tree}$ and
for all $\vec{i} = i_1 \dots i_r$, $J_{\vec{i}} = \set{ v \in \tree : ([v, i_1],\dots,[v, i_r]) \in J}$.
\end{itemize}
\end{lem}

\noindent
Using these auxiliary formulas,
we can perform the forward translation to $\mso[\sigcode{\sigoriginal}{n}]$.

\begin{lem}\label{lemma:forward}
Let $\psi$ be a formula in $\gso[\sigoriginal]$
with free first-order variables among $x_1,\dots,x_n$,
and free second-order variables among $X_1,\dots,X_m$.
We can construct a formula
\[ \tforward{\psi}(\tforward{x_1},\dots,\tforward{x_n},\tforward{X_1},\dots,\tforward{X_m}) \in \mso[\sigcode{\sigoriginal}{n}] \]
such that
for all consistent $\sigcode{\sigoriginal}{n}$-trees $\tree$,
for all elements $a_1, \dots, a_n$ in $\decode{\tree}$ encoded by $\tforward{j_1},\dots,\tforward{j_n}$
and for all sets of $\sigma$-guarded relations $J_1,\dots,J_m$ encoded by $\tforward{J_1},\dots,\tforward{J_m}$,
\begin{align*}
\decode{\tree},a_1,\dots,a_n,J_1,\dots,J_m \models \psi \quad \text{ iff } \quad
\tree,\tforward{j_1},\dots,\tforward{j_n},\tforward{J_1},\dots,\tforward{J_m} \models \tforward{\psi}.
\end{align*}
\end{lem}

\begin{proof}
\newcommand{\single}[1]{\textup{single}(#1)} 
\newcommand{\match}[2]{\textup{match}_{#1}(#2)} 
The proof is by induction on the structure of $\psi$.
\begin{itemize}
\item Assume $\psi = R x_{i_1} \dots x_{i_r}$.
Then
\[\tforward{\psi} := \exists z . \Bigg( \bigvee_{\rho}
\bigg( z \in R_{\rho(i_1) \dots \rho(i_r)} \wedge \bigwedge_{i \in \set{i_1,\dots,i_r}} z \in Z^{x_{i}}_{\rho(i)}  \bigg) \Bigg)\]
where $\rho$ ranges over maps from $\set{1,\dots,r}$ to $\set{1,\dots,n}$.
This expresses that there is some node in the coded structure
where $R$ holds for elements coded by $\rho(i_1) \dots \rho(i_r)$,
and these elements are precisely $x_{i_1} \dots x_{i_r}$.

Similarly for $\psi = X x_{i_1} \dots x_{i_r}$.
\item Assume $\psi = (x_{i_1} = x_{i_2})$.
Then
\[\tforward{\psi} := \forall z . \Bigg(  \bigwedge_{j \in \set{1,\dots,n}}
\bigg( z \in Z^{x_{i_1}}_j \leftrightarrow z \in Z^{x_{i_2}}_j  \bigg)\Bigg).\]
This expresses that the valuations for the variables $x_{i_1}$ and $x_{i_2}$ are identical.
\item The translation commutes with $\vee$, $\wedge$, and $\neg$.
\item Assume $\psi = \exists x . (\chi$).
Then \[ \tforward{\psi} := \exists \tforward{x} .  \left(\correct{\tforward{x}} \wedge \tforward{\chi} \right) . \]
\item Assume $\psi = \exists X . (\chi)$
for $X$ an $r$-ary relation.
Then \[ \tforward{\psi} := \exists \tforward{X} . \left(\correctr{\tforward{X}} \wedge \tforward{\chi} \right) . \]
\end{itemize}
The proof of correctness is straightforward.
\end{proof}

We now return to the proof of Lemma~\ref{lemma:forward-gso}. Recall that we have a sentence $\phi \in \gso[\sigoriginal]$ that is $\gnlinvar[\sigoriginal]$-invariant and $n \geq \max\set{\width{\sigoriginal},l}$.

We can apply Lemma~\ref{lemma:forward} to produce a sentence $\tforward{\phi} \in \mso[\sigcode{\sigoriginal}{n}]$ such that for all
consistent $\sigcode{\sigoriginal}{n}$-trees $\tree$,
$\tree \models \tforward{\phi}$ iff $\decode{\tree} \models \phi$. This is the property required in Lemma~\ref{lemma:forward-gso}, but the sentence $\tforward{\phi}$ is in $\mso$ rather than $\Lmu$.

Consider the sentence $\tforward{\phi} \wedge \consistencyform$. We claim that this is $\sigcode{\sigoriginal}{n}$-bisimulation invariant.
Let $\tree$ and $\tree'$ be $\sigcode{\sigoriginal}{n}$-bisimilar.
If there is any inconsistency in one tree, then bisimilarity implies that the other is also inconsistent, and $\tforward{\phi} \wedge \consistencyform$ does not hold in either tree.
In the case that $\tree$ and $\tree'$ are both consistent, then their bisimilarity implies that their decodings are both $\gnninvar[\sigoriginal]$-bisimilar and hence $\gnlinvar[\sigoriginal]$-bisimilar (since we have ensured that $n \geq l$). But this $\gnlinvar[\sigoriginal]$-invariance and the property in Lemma~\ref{lemma:forward}, imply that they agree on $\tforward{\phi}$.
Therefore, $\tree$ and $\tree'$ agree on $\tforward{\phi} \wedge \consistencyform$ which is enough to conclude that
$\tforward{\phi} \wedge \consistencyform$ is $\sigcode{\sigoriginal}{n}$-bisimulation invariant.

This means that we can apply the Janin-Walukiewicz Theorem to $\tforward{\phi} \wedge \consistencyform$ to produce an equivalent formula $\phi^\mu \in \Lmu[\sigcode{\sigoriginal}{n}]$ for the forward mapping.
This completes the proof of Lemma~\ref{lemma:forward-gso}.

\subsection{Backward mapping}
The second component will depend on our target sublogic $\logictarget$. It requires
an operation (not necessarily effective)
taking a $\sigoriginal$-structure $\fB$ to a tree structure $\unravellingcode$
such that $\decode{\unravellingcode}$ agrees with $\fB$ on all $\logictarget$ sentences. Informally,
$\unravellingcode$ will be the encoding of some unravelling of $\fB$ appropriate for $\logictarget$, perhaps
with additional properties.
A \emph{backward mapping for $\logictarget$} takes
$\phi'_0 \in \Lmu$ describing tree codes to
 a sentence $\phi_1 \in \logictarget$  such that:
for all $\sigoriginal$-structures~$\fB$,
$\fB \models \phi_1$ iff
 $\unravellingcode \models \phi'_0$.

The formula $\phi_1$ will depend on simplifying the formula
$\phi'_0$ based on the fact that one is working on an unravelling.
For $\logictarget=\gfp[\sigtarget, \sigmag]$ over subsignatures $\sigtarget, \sigmag$ of the original signature $\sigoriginal$,
$\unravellingcode$ will be the appropriate guarded unravelling;
 we will see that results of~\cite{GradelHO02} can easily be refined to give the backward mapping formula $\phi_1$
in  $\gfp[\sigtarget, \sigmag]$.
For $\gnfpk$, providing both the appropriate  unravelling and the formula in the backward mappings will require more work.

\subsection{Definability problem}

The \emph{$\logictarget$ definability problem for logic $\cL$} asks:
given some input sentence $\phi \in \cL$,
is there some $\psi \in \logictarget$
such that
$\phi$ and $\psi$ are logically equivalent?

The forward and backward method of Figure~\ref{fig:bandfeffsem} gives us a generic approach to this problem.
The algorithm consists of applying the forward mapping to get
$\phi'_0$, applying the  backward  mapping to $\phi'_0$ and obtaining the formula component of the mapping,  $\phi_1$,
and then checking if $\phi_1$  is equivalent to $\phi_0$.
We claim $\phi_0$ is $\logictarget$ definable iff $\phi_0$ and $\phi_1$ are equivalent.
If $\phi_0$ and $\phi_1$ are logically equivalent then $\phi_0$ is clearly $\logictarget$ definable using $\phi_1$.
In the other direction, suppose
 that  $\phi_0$ is $\logictarget$-definable.  Fix $\fB$, and let $\unravellingcode$ be given by the backward mapping.
 Then
\begin{align*}
\fB \models \phi_0
&\Leftrightarrow
\decode{\unravellingcode} \models \phi_0 \text{ since $\phi_0$ is equivalent to an $\logictarget$ sentence and} \\
&\qquad\qquad\qquad\qquad\qquad\qquad \text{$\decode{\unravellingcode}$ agrees with $\fB$ on $\logictarget$ sentences} \\
&\Leftrightarrow
\unravellingcode \models \phi'_0 \text{ by \lemmafwd } \\
&\Leftrightarrow
\fB \models \phi_1 \text{ by Backward Mapping for $\logictarget$}
\end{align*}
Hence, $\phi_0$ and $\phi_1$ are logically equivalent, as required.
Thus, we get the following general decidability result:

\begin{prop}%
\label{prop:effectivedefhigh}
Let $\logictarget$ be a subset of $\gnlinvar[\sigoriginal]$-invariant $\gso[\sigoriginal]$
such that we have an effective backward mapping for $\logictarget$.
Then
the $\logictarget$ definability problem is decidable
for $\gnlinvar[\sigoriginal]$-invariant $\gso[\sigoriginal]$.
\end{prop}

Above, we mean that there is an algorithm that decides $\logictarget$ definability for any
input $\gso[\sigoriginal]$ sentence that is  $\gnlinvar[\sigoriginal]$-invariant, with the output
being arbitrary otherwise.
The decidability is relying on the fact that logical equivalence is decidable for
$\gnlinvar[\sigoriginal]$-invariant $\gso[\sigoriginal]$: this follows from performing the forward mapping
and checking equivalence of the corresponding sentences on the encodings.
The approach above gives a definability test in the usual sense for inputs in $\gnfp[\sigoriginal]$, since
these are all $\gnlinvar[\sigoriginal]$-invariant for some $l$. In particular
we will see that we can test whether a $\gnfp^l[\sigoriginal]$ sentence is in $\gfp[\sigtarget]$ or in $\gnfpk[\sigtarget]$ for some subsignature $\sigtarget$ of $\sigoriginal$.
But there are larger $\gnlinvar$-invariant logics (e.g.~the $\gnfpup$ logic in~\cite{gnfpup}), so we can actually apply Proposition~\ref{prop:effectivedefhigh} to decide $\gfp$ or $\gnfp$ definability starting with inputs in these more expressive logics as well.


%% file: gfp.tex
\section{Identifying \texorpdfstring{$\gfp$}{GFP} definable sentences}\label{sec:gfp}

\subsection{Decidability of \texorpdfstring{$\gfp$}{GFP}-definability}

For $\gfp$,
we can instantiate the high-level algorithm by giving a backward mapping.
The backward mapping starts with a $\mu$-calculus formula describing tree codes with some bag size $m$, and produces a $\gfp$-formula describing relational structures. This backward mapping is tuned to a particular subsignature $\sigma'$ of the original signature $\sigma$, with guards taken from $\sigmag \subseteq \sigma'$. Throughout this section, we assume that $\sigmag \subseteq \sigtarget \subseteq \sigoriginal$.

\begin{lem}[$\gfp$-Bwd, adapted from~\cite{GradelHO02}]\label{lemma:backwards-gfp}
Given $\phi^\mu \in \Lmu[\sigcode{\sigoriginal}{m}]$
and $\sigmag \subseteq \sigtarget \subseteq \sigoriginal$,
$\phi^{\mu}$ can be translated into $\psi \in \gfp[\sigtarget,\sigmag]$
such that for all $\sigoriginal$-structures~$\fB$,
$\fB \models \psi$ iff $\gunravel{\fB}{\sigtarget,\sigmag} \models \phi^\mu$.
\end{lem}

As with the forward mapping, the translation proceeds by induction on the structure of the formula $\phi^\mu$.
Since each node in  $\gunravel{\fB}{\sigtarget,\sigmag}$ is strictly $\sigmag$-guarded,
each node is based on at most $\width{\sigmag}$ elements from $\fB$.
To deal with this, the translation of some formula~$\theta$ actually generates a family of formulas: for each $0 \leq k \leq \width{\sigmag}$,  a formula $\tbackward{\theta}_k$ with $k$ free first-order variables is produced such that it correctly captures the meaning of the $\mu$-calculus formula from a node of $\gunravel{\fB}{\sigtarget,\sigmag}$ that represents exactly $k$ elements from $\fB$.
The desired sentence $\psi$ for Lemma~\ref{lemma:backwards-gfp} is $\tbackward{(\phi^\mu)}_0$, since the root of $\gunravel{\fB}{\sigtarget,\sigmag}$ has an empty domain.

For the purposes of the induction, we must also deal with formulas with free second-order variables.
For each fixpoint variable $X$,
each $1 \leq j \leq \width{\sigmag}$,
and each $P \in \sigcode{\sigmag}{\width{\sigmag}}$ with $\indices{P} = \set{1,\dots,j}$,
we introduce
a second-order variable $X_{j,P}$ to represent
nodes of size $j$ whose indices
are strictly $\sigmag$-guarded by $P$ (please refer to the definitions on page~\pageref{encodedguardedness}).
The relation $X_{j,P}$ is a $j$-ary relation.
In order to handle nodes with empty domain or domain of size 1 that are trivially $\sigmag$-guarded,
we also introduce $X_{0,\top}$ and $X_{1,\top}$.
We define $\tbackward{X}$ to be the set of these second-order variables based on~$X$.

Fix some $\sigoriginal$-structure $\fB$ and $\gunravel{\fB}{\sigtarget,\sigmag}$.
We write $\elem{v}$ to denote the ordered tuple of elements
from $\fB$ represented at $v$ in $\gunravel{\fB}{\sigtarget,\sigmag}$.
A set $V$ of nodes in $\gunravel{\fB}{\sigtarget,\sigmag}$ is a \emph{bisimulation-invariant valuation for a free variable $X$}
if it satisfies the following property:
if it contains a node then it contains every node
that is the root of a bisimilar subtree.
We write $\tbackward{V}$ for its representation in $\fB$.
Specifically, $\tbackward{V}$ consists of valuations $V_{j,P}$
for each $X_{j,P}$ in $\tbackward{X}$,
where
\[
V_{j,P} =
\{ \elem{v} : v \in V, \size{\elem{v}} = j, \text{ and the label }\tau \text{ at }v\text{ is strictly }
\sigmag\text{-guarded by }P\} .
\]
We also set $V_{0,\top}$ to $\top$ (respectively, $\bot$)
if $J$ contains all nodes with empty domain (respectively, if $J$ contains no nodes with empty domain),
and $V_{1,\top} = \{ \elem{v} : v \in V, \size{\elem{v}} = 1 \}$.
\newcommand{\kprime}{k}
\begin{lem}%
\label{lemma:gfpbackfree}
Let $\varphi \in \Lmu[\sigcode{\sigoriginal}{m}]$
with free second-order variables $\vec{X}$,
and let $\sigmag \subseteq \sigtarget \subseteq \sigoriginal$.
For each $0 \leq \kprime \leq \width{\sigmag}$,
we can construct a $\gfp[\sigtarget,\sigmag]$-formula
$\tbackward{\varphi}_{\kprime}(x_{1},\dots,x_{\kprime},\tbackward{\vec{X}})$
such that
for all $\sigoriginal$-structures $\fB$,
for all bisimulation-invariant valuations $\vec{V}$ of $\vec{X}$,
and for all nodes $v$ in $\gunravel{\fB}{\sigtarget,\sigmag}$
with $\size{\elem{v}} = \kprime$,
\begin{align*}
&\text{$\fB,\elem{v},\tbackward{\vec{V}} \models \tbackward{\varphi}_{\kprime}$ \quad iff \quad
$\gunravel{\fB}{\sigtarget,\sigmag},v,\vec{V} \models \varphi$} .
\end{align*}
\end{lem}

\begin{proof}[Proof sketch]
We proceed by induction on the structure of $\varphi$.

\begin{itemize}
\item If $\varphi = D_n$, then $\tbackward{\varphi}_{\kprime}$
is $\top$ if $\kprime = n$ and $\bot$ otherwise.
\item If $\phi = R_{i_1 \dots i_l}$ such that $R \notin \sigtarget$
or $\set{i_1,\dots, i_l} \not\subseteq \set{1, \dots, \kprime}$,
then $\tbackward{\phi}_{\kprime} := \bot$.
Otherwise
$\tbackward{\phi}_{\kprime} := R \, x_{i_1} \dots x_{i_l}$.
\item If $\varphi = X$,
then $\tbackward{\phi}_{\kprime} := \bigvee_{\alpha} (\alpha(x_1,\dots,x_{\kprime}) \wedge X_{\kprime,P} \, x_1 \dots x_{\kprime})$
where $\alpha$ ranges over atomic formulas that are strict $\sigmag$-guards for $\set{ x_1, \dots, x_{\kprime}}$,
and $P$ is the encoding of $\alpha$.
\item The translation commutes with $\vee$ and $\wedge$ and $\neg$
for each~$\kprime$.
\item If $\varphi = \dmodality{\rho} \chi$
with $\dom{\rho} = \set{i_1,\dots,i_l} \not\subseteq \set{1,\dots,{\kprime}}$,
then $\tbackward{\phi}_{\kprime} := \bot$.
Otherwise $\tbackward{\phi}_{\kprime}$ is
\[
\bigvee_{l \leq j \leq \width{\sigmag}} \bigvee_{\alpha} \exists y_1 \dots y_j .
	\left( \alpha(y_1,\dots,y_j) \wedge \tbackward{\chi}_j(y_1,\dots,y_j) \wedge \textstyle \bigwedge_{i \in \dom{\rho}} x_i = y_{\rho(i)} \right) \]
where $\alpha$ ranges over atomic formulas that are strict $\sigmag$-guards for $y_1,\dots,y_j$.
\item Finally, if $\varphi = \mu Y . \chi$ then
$\tbackward{\varphi}_{\kprime}$ is \[\bigvee_\alpha \left(\alpha(x_1,\dots,x_{\kprime}) \wedge [\LFPA{Y_{\kprime,P}}{y_1,\dots,y_{\kprime}} . S_{\mu Y.\chi}](x_1,\dots,x_{\kprime}) \right)\]
where $\alpha$ ranges over atomic formulas that are strict $\sigmag$-guards for $y_1,\dots,y_j$,
the relation $P$ is the encoding of $\alpha$,
and $S_{\mu Y . \chi}$
is a system
consisting of equations
\[Y_{j,P}, y_1 \dots, y_j := \tbackward{P}_j(y_1,\dots,y_j) \wedge \tbackward{\chi}_{j}(y_1,\dots,y_j) \]
for each $Y_{j,P}$ in $\tbackward{Y}$.
\end{itemize}

\noindent
The formulas produced by this translation are in $\gfp[\sigtarget,\sigmag]$;
in particular, note the $\sigmag$-guarded existential quantification in the diamond modality translation,
and the $\sigmag$-guarded fixpoints in the fixpoint translation
(we use simultaneous fixpoints here, but these can be eliminated if required).
The correctness of this translation comes from the fact that
every node in $\gunravel{\fB}{\sigtarget,\sigmag}$
represents elements that are strictly $\sigmag$-guarded.
Hence, the translation of a diamond modality that expresses the existence of a neighboring node in the tree
translates into a $\sigmag$-guarded existential quantification.
Likewise, the fixpoint formulas that are defining a set of nodes in the tree
can be translated into fixpoint formulas defining sets of tuples that are all $\sigmag$-guarded.
We omit the formal proof of correctness,
since it is similar to the more complicated proof of correctness for Lemma~\ref{lemma:backwards-gnfp}
that we will give later.
\end{proof}

As mentioned earlier, the desired formula $\psi$ for Lemma~\ref{lemma:backwards-gfp} is $\tbackward{(\phi^\mu)}_0$ obtained
using Lemma~\ref{lemma:gfpbackfree}.

Plugging Lemma~\ref{lemma:backwards-gfp} into our high-level algorithm, with $\gunravel{\fB}{\sigtarget,\sigmag}$
as $\unravellingcode$, we get decidability
of the  $\gfp$-definability problem:

\begin{thm}%
\label{thm:effectivedefgfp}
The $\gfp[\sigtarget,\sigmag]$ definability problem is decidable
for $\gnlinvar[\sigoriginal]$-invariant $\gso[\sigoriginal]$
where $l \geq \width{\sigoriginal}$
and $\sigmag \subseteq \sigtarget \subseteq \sigoriginal$.
\end{thm}

\subsection{Isolating the complexity of \texorpdfstring{$\gfp$}{GFP}-definability}

\input{viewformulagfp}

\subsection{Further applications of the machinery}

Our decidability results give us a corollary on definability
in fragments of $\fo$ when the input is in $\fo$:

\begin{cor}%
\label{cor:gfdef}
The $\gf[\sigtarget,\sigmag]$ definability problem is  decidable
for $\gnlinvar[\sigoriginal]$-invariant $\fo[\sigoriginal]$ where
$l \geq \width{\sigoriginal}$
and $\sigmag \subseteq \sigtarget \subseteq \sigoriginal$.

In the special case that the input is in $\gnf[\sigma]$,
the $\gf[\sigtarget,\sigmag]$ definability problem is $\twoexptime$-complete.
\end{cor}

\begin{proof}
It was known from~\cite{gforig} that if
$\phi$ is in $\fo[\sigma]$ and is guarded bisimulation invariant with respect to $\sigma$, then
it is in $\gf[\sigma]$. By a straightforward refinement of the argument in~\cite{gforig},
 we see that if $\phi$ is in $\fo[\sigma]$ and is $\ginvar[\sigtarget,\sigmag]$-invariant,
then it is in $\gf[\sigtarget, \sigmag]$.

Hence, given an input formula $\phi$, we just use the algorithm of
Theorem~\ref{thm:twoexpdefgfp} to see if $\phi$ is in $\gfp[\sigtarget, \sigmag]$.
If it is, then we conclude that $\phi$ is actually in $\gf[\sigtarget,\sigmag]$.

In the special case that the input is in $\gnf[\sigma]$,
we can use Theorem~\ref{thm:twoexpdefgfp} to get the $\twoexptime$ upper bound.
The lower bound follows from a standard
reduction from satisfiability (see the proof of the lower bound
in Theorem~\ref{thm:twoexpdefgfp}).
\end{proof}

Note that in this work we are characterizing sublogics within fragments of fixpoint
logics and within fragments of first-order logic. We do not deal with
identifying first-order definable formulas within a fixpoint logic, as in~\cite{bowboundedness,bco,boundedness}.

We can also apply our theorem to answer some questions about \emph{conjunctive queries (CQs)}: formulas
built up from relational atoms via $\wedge$ and $\exists$.
When the input $\phi$ to our definability algorithm is a CQ,
$\phi$ can be written as a $\gf$ sentence exactly when it is \emph{acyclic}:
roughly speaking, this means it can be built up from guarded existential quantification
(see~\cite{robbers}).
Transforming a query to an acyclic one could be quite relevant
in practice, since acyclic queries can be evaluated in linear time~\cite{acycliclin}.
There are well-known methods for deciding whether a CQ $\phi$ is acyclic, and recently
these have been extended to the problem of determining whether $\phi$ is acyclic
for all structures satisfying a set of constraints (e.g., Guarded Universal
Horn constraints~\cite{acyclicmodpods} or
Functional Dependencies~\cite{acyclicmodlics}). Using Corollary~\ref{cor:gfdef} above we can
get an analogous result for arbitrary constraints in the guarded fragment:

\begin{thm}%
\label{thm:acyclic}
Given a finite set of $\gf$ sentences $\Sigma$ and a CQ sentence $Q$, we can decide
whether  there is a union of acyclic CQs $Q'$ equivalent to $Q$ for all structures satisfying $\Sigma$.
The problem is $\twoexp$-complete.
\end{thm}

For the purposes of this proof, an acyclic CQ is one built
up from atomic relations by conjunction and guarded existential quantification alone.
\cite{robbers} showed that this is equivalent to the more usual definitions, 
via the associated graph being  chordal and conformal,  or the associated graph being  tree decomposable.
First, we need the following basic result:

\begin{clm} Let $\Sigma$ be a finite set of $\gf$ sentences,
$Q$ a CQ, and suppose there is a $\gf$ sentence $\phi$ such
that $Q$ is equivalent to  $\phi$ for all structures satisfying
$\Sigma$. Then there is a union of  acyclic CQs $Q'$ such that $Q$ is equivalent
to $Q'$ for all structures satisfying $\Sigma$.
\end{clm}

Note that in the case $\Sigma$ is empty, this states that a CQ is in $\gf$ iff
it is a union of acyclic queries. If a CQ is equivalent to a disjunction of CQs, then it is equivalent to one of its disjuncts~\cite{sagivyan},
thus in the case that $\Sigma$ is empty (or more generally, when $\Sigma$ is universal Horn) we can strengthen the conclusion
to be that $Q'$  is a single acyclic CQ\@.
Although the characterization in the claim
is probably well-known, we provide a proof:

\begin{proof}[Proof of Claim]
We apply the ``treeification lemma'' of~\cite{bgo}, which states that
for every CQ sentence $Q$ we have a union of acyclic queries $Q'$
such that:

\begin{itemize}
\item $Q'$ implies $Q$
\item for every $\chi$ in $\gf$: $\chi$ implies $Q'$ if and only if $\chi$ implies $Q$
\end{itemize}

\noindent
Suppose there is a formula
$\phi$ in $\gf$ such that $Q$ is equivalent to  $\phi$ for all structures
satisfying $\Sigma$.
Then clearly
$\Sigma \wedge \phi$ implies $Q$, and hence by the second item
above, $\Sigma \wedge \phi$ implies $Q'$ and
thus for structures satisfying $\Sigma$, $Q$ implies $Q'$.

Therefore for structures satisfying $\Sigma$, $Q$ is equivalent to $Q'$ as required.
\end{proof}

Returning to the proof of Theorem~\ref{thm:acyclic}, the above claim
tells us   that $Q$ is equivalent to an acyclic $Q'$ for structures satisfying $\Sigma$
exactly when $\Sigma \wedge Q$ is equivalent to a sentence in $\gf$.
Since $\Sigma \wedge Q \in \gnf$ when $\Sigma$ is a set of $\gf$ sentences and $Q$ is a CQ,
Corollary~\ref{cor:gfdef} implies that we can decide this
in $\twoexp$ time.
This gives the desired upper bound.
As before, the lower bound follows from a standard
reduction from satisfiability.

Note that if $\Sigma$ consists
of universal Horn constraints (``TGDs''), then a CQ $Q$ is equivalent to a union of CQs $Q'$ relative to  $\Sigma$
implies  that it is equivalent to one of the disjuncts of $Q'$. Thus the result above implies
decidability of acyclicity relative to universal Horn  $\gf$ sentences,
one of the main results of~\cite{acyclicmodpods}. Also note that
the results of~\cite{bgo} imply that
 in Theorem~\ref{thm:acyclic} the quantification ``for all structures'' can
equivalently be replaced by ``for all finite structures''.


%% file: viewformulagfp.tex
We now see if we can get a more efficient $\gfp$-definability test, with
the goal of obtaining a tight bound on the complexity of this problem.

There are two sources of inefficiency in the high-level algorithm.
First, the forward mapping is non-elementary since we pass through
$\mso$ on the way to a $\mu$-calculus formula.
Second, testing equivalence of the original sentence
with the sentence produced by the forward and backward mappings na\"ively
would cause an additional blow-up: we would apply a forward mapping again
in order to produce tree automata, and then
check their  equivalence using an $\exptime$ algorithm.

For the special case of input in $\gnfp$,
we can avoid these inefficiencies and obtain an optimal complexity bound.

\begin{thm}%
\label{thm:twoexpdefgfp}
The $\gfp[\sigtarget,\sigmag]$ definability problem is
$\twoexptime$-complete
for input in $\gnfp[\sigoriginal]$ where $\sigmag \subseteq \sigtarget \subseteq \sigoriginal$.
\end{thm}

The proof of Theorem~\ref{thm:twoexpdefgfp} will require some extra
machinery.
The main idea behind the optimized procedure
is to directly use automata throughout the process.
First, for input $\phi$ in $\gnfp$ it is known from~\cite{boundedness} that
there is a forward mapping
directly producing a tree automaton $\cA_\phi$
with exponentially-many states
that accepts a consistent tree $\tree$ iff $\decode{\tree} \models \phi$;
$\cA_\phi$ accepts exactly the consistent trees that satisfy the formula $\phi^\mu$ from Lemma~\ref{lemma:forward-gso}.
This direct construction avoids passing through $\mso$,
and can be done in $\twoexptime$.
We can then construct an automaton $\cA'_\phi$ from $\cA_\phi$
that accepts a tree $\tree$
iff $\gunravel{\decode{\tree}}{\sigtarget,\sigmag}$ is accepted by $\cA_\phi$;
we call this the \emph{$\ginvar[\sigtarget,\sigmag]$-view automaton},
since it mimics the view of $\cA_\phi$ running on the guarded unravelling of $\decode{\tree}$.
This can be seen as an automaton that
represents the composition of the backward mapping
with the forward mapping.
With these constructions in place,
we have the following improved algorithm
to test definability of $\phi$ in $\gfp$:
construct $\cA_\phi$ from $\phi$,
construct $\cA'_\phi$ from $\cA_\phi$,
and test equivalence of $\cA_\phi$ and $\cA'_\phi$ over consistent trees.
Note that with this improved procedure
it is not necessary to actually construct the backward mapping,
or to pass forward to trees for a second time in order to test equivalence.
Overall, the procedure can be shown to run in $\twoexptime$.
A reduction from $\gfp$-satisfiability testing, which
is known to be $\twoexptime$-hard, yields the lower bound.

\myparagraph{Upper bound}
We now give more details of the upper bound
in Theorem~\ref{thm:twoexpdefgfp}.
As mentioned earlier, there is an improved forward mapping
from
formulas  in $\gnfp[\sigoriginal]$
directly to automata,
without passing through $\mso$.
It is known from prior work
how to do this in $\twoexptime$:

\begin{lem}[$\gnfp$-Fwd Automaton,~\cite{boundedness}]\label{lemma:gnfp-fwd-aut}
Given $\phi \in \gnfp^l[\sigoriginal]$ and given some $m \geq \max\set{l,\width{\sigoriginal}}$,
we can construct in \twoexptime
a 2-way alternating $\mu$-automaton $\cA_\phi$
such that
for all consistent $\sigcode{\sigoriginal}{m}$-trees $\tree$,
$\tree \in L(\cA_\phi)$ iff $\decode{\tree} \models \phi$.

The number of states of $\cA_\phi$ is exponential in $\size{\phi}$,
and the number of priorities is linear in~$\size{\phi}$.
\end{lem}

It is straightforward to construct a 2-way alternating $\mu$-automaton that checks
whether a $\sigcode{\sigtarget}{m}$-tree is consistent.
This is also known from prior work, e.g.~\cite{boundedness}.

\begin{lem}[Consistency Automaton]\label{lemma:consistency-aut}
We can construct in \twoexptime
a 2-way alternating $\mu$-automaton $\cC$
such that
for all $\sigcode{\sigoriginal}{m}$-trees $\tree$,
$\tree \in L(\cA_\phi)$ iff $\tree$ is consistent.

The number of states of $\cA_\phi$ is exponential in $\size{\phi}$,
and the number of priorities is linear in~$\size{\phi}$.
\end{lem}

As mentioned above, we can then
construct a $\ginvar[\sigtarget,\sigmag]$-view automaton,
which can be seen as the composition of the backward mapping with the forward mapping.
This results in an additional blow-up of the state set by a factor of $2^{k+1}$ (for $k = \width{\sigtarget}$)
but no further increase in size.

\begin{lem}[$\gfp$-View Automaton]%
\label{lemma:gfp-view-automaton}
Let $\sigmag \subseteq \sigtarget \subseteq \sigoriginal$.
Given a 2-way alternating $\mu$-automaton $\cA$ over $\sigcode{\sigoriginal}{m}$-trees
with $m \geq \width{\sigtarget}$,
we can construct a $\ginvar[\sigtarget,\sigmag]$-view automaton $\cA'$
such that
$\tree \in L(\cA')$ iff $\gunravel{\decode{\tree}}{\sigtarget,\sigmag} \in L(\cA)$.
The view automaton can be constructed in time
polynomial in the size of $\cA$
and exponential in $k = \width{\sigtarget}$.
The number of states increases by a factor of $2^{k+1}$
and the number of priorities remains the same.
\end{lem}

\begin{proof}
We need to design $\cA'$ so that when it is run on a  consistent $\sigcode{\sigoriginal}{m}$-tree $\tree$,
it mimics the run of $\cA$ on $\gunravel{\decode{\tree}}{\sigtarget,\sigmag}$.
Before we do this, it is helpful to recall
what these tree codes look like,
and what their relationship is.

In $\tree$,
each node represents at most $m$ elements
of $\decode{\tree}$,
and these elements are not necessarily guarded.
On the other hand, each node in $\gunravel{\decode{\tree}}{\sigtarget,\sigmag}$
represents a strictly $\sigmag$-guarded set of elements
of size at most $k = \width{\sigtarget}$.
But there is a strong relationship between
$\tree$ and $\gunravel{\decode{\tree}}{\sigtarget,\sigmag}$:
each node in~$\gunravel{\decode{\tree}}{\sigtarget,\sigmag}$ is
a copy of a strictly $\sigmag$-guarded subset of $\decode{\tree}$, and hence
can be identified with a strictly $\sigmag$-guarded subset of elements. Since
every atom must be represented in at least one node of the tree decomposition, these elements
must occur together in
a \emph{single node} of $\tree$.

The construction of $\cA'$ from $\cA$ reflects this.
We augment each state of $\cA$ to also include
the current \emph{strictly $\sigmag$-guarded view},
which is just some strictly $\sigmag$-guarded subset
of the at most $m$ elements
represented in the current node.
Each strictly $\sigmag$-guarded subset is of size at most $k$.
The view automaton $\cA'$ simulates $\cA$
as if it could only see the strictly $\sigmag$-guarded view of the label.

For each single move of the original automaton,
we allow the view automaton to make a finite (but unbounded) series
of moves before selecting the next strictly $\sigmag$-guarded view.
Why is this?
Observe that each single move of $\cA$ on some guarded unravelling
$\gunravel{\decode{\tree}}{\sigtarget,\sigmag}$
leads to a neighboring node that is based on a new strictly $\sigmag$-guarded set of elements from $\decode{\tree}$.
Although these elements must be represented in a
single node in $\tree$ (since they are guarded),
the node in $\tree$ that represents this new guarded set could be far away from
node that represents the current guarded view.
Hence, we allow the view automaton to navigate to a node in $\tree$
representing this next strictly $\sigmag$-guarded view
before continuing the simulation.

\medskip

We give more details on the construction. It may be helpful at this stage to refer to Section~\ref{sec:automata} for a brief introduction to the $\mu$-automata used here.

Let $\cA = \tuple{\sigcode{\sigoriginal}{m},\QE,\QA,q_0,\delta,\Omega}$,
where $\QE$ and $\QA$ represent existential and universal states,
respectively, $q_0$ and $\delta$ are the initial state and transition
function, while $\Omega$ is a priority function used to define
the acceptance condition.
We construct the $\ginvar[\sigtarget,\sigmag]$-view automaton
$\cA' = \tuple{\sigcode{\sigoriginal}{m},\QE',\QA',q_0',\delta',\Omega'}$ as follows.

Let $\views$ consist of subsets of $\set{1,\dots,m}$ of size at most $k = \width{\sigtarget}$.
Then let $\QE' := \QE \times \views \times \set{\readymode,\waitmode}$ and
$\QA' :=  \QA \times \views \times \set{\readymode,\waitmode}$,
with initial state $q'_0 := (q_0,\emptyset,\readymode)$.

Let $\tau$ denote a node label in a $\sigcode{\sigoriginal}{m}$-tree.
For $\view \in \views$, let $\restrict{\tau}{\sigtarget,\view}$ denote the restriction of the label $\tau$
to the indices in $\view$ and $\sigcode{\sigtarget}{m}$.

In $\readymode$ mode:
$\delta'((q,\view,\readymode),\tau)$ is
the set of moves of the form
$(\dstay,(q',\view',\waitmode))$
such that
$(d,\rho',q') \in \delta(q,\restrict{\tau}{\sigtarget,\view})$
and $\view' = \dom{\rho'}$.
In other words,
the automaton selects the next state in the simulation of $\cA$
based on its view
and then switches to $\waitmode$ mode.

In $\waitmode$ mode:
$\delta'((q,\view,\waitmode),\tau)$ is
the set consisting of \begin{itemize}
\item $(\dup,\rho,(q,\rho(\view),\waitmode))$
for each $\rho \in \EdgeLabels$ with $\dom{\rho} \supseteq \view$,
\item $(\ddown,\rho,(q,\rho(\view),\waitmode))$
for each $\rho \in \EdgeLabels$ with $\dom{\rho} \supseteq \view$,
and
\item
$(\dstay,(q,\view',\readymode))$ for each
$\view' \supseteq \view$ that is strictly $\sigmag$-guarded in $\tau$.
\end{itemize}
Thus, in $\waitmode$ mode,
the automaton can either move to a neighboring node that
contains $\view$ (renamed according to some $\rho$)
and stay in $\waitmode$ mode,
or it can expand to a new strictly $\sigmag$-guarded view and switch to $\readymode$ mode
in order to continue the simulation.

The priority function $\Omega'$ is defined such that states in $\readymode$ mode inherit
the priority of the underlying state from $\cA$.
In $\waitmode$ mode,
$\Omega'((q,\view,\waitmode))$ is $1$ if $q \in \QE$,
and $0$ if $q \in \QA$;
this ensures that the controlling player cannot cheat by forever delaying the next step in the simulation.
\end{proof}

We can use these automata for an improved
\twoexptime decision procedure.
Suppose the input is $\phi \in \gnfp^l[\sigoriginal]$.
Let $m = \max\set{l,\width{\sigtarget}}$.
We start by constructing $\cA_\phi$ and $\cC$ using Lemmas~\ref{lemma:gnfp-fwd-aut}~and~\ref{lemma:consistency-aut},
and then construct the view automaton $\cA'_\phi$ from $\cA_\phi$ using Lemma~\ref{lemma:gfp-view-automaton}.
This can all be done in \twoexptime.

We claim that $\cA_\phi$ is equivalent to $\cA'_\phi$ for $\sigcode{\sigoriginal}{m}$-trees in $L(\cC)$
iff
$\phi$ is $\gfp[\sigtarget,\sigmag]$-definable.
First, suppose $\cA_\phi$ and $\cA'_\phi$ are equivalent with respect to tree codes in $L(\cC)$.
Let $\psi \in \gfp[\sigtarget,\sigmag]$ be the formula obtained by converting $\cA_\phi$ to an equivalent
$\mu$-calculus formula using Theorem~\ref{thm:muaut}, and then applying \lemmagfpbwd{Lemma~$\gfp$-Bwd}.
We show that $\phi$ is $\gfp[\sigtarget,\sigmag]$-definable using this sentence $\psi$.
For all $\sigoriginal$-structures $\fB$:
\begin{align*}
&\fB \models \phi \\
\Leftrightarrow \quad
&\decode{\unravelm{\fB}{\sigoriginal}} \models \phi \tag{by $\gnlinvar[\sigoriginal]$-invariance of $\phi$} \\
\Leftrightarrow \quad
&\unravelm{\fB}{\sigoriginal} \in L(\cA_\phi) \tag{by Lemma~$\gnfp$-Fwd Automaton} \\
\Leftrightarrow \quad
&\unravelm{\fB}{\sigoriginal} \in L(\cA'_\phi) \tag{by language equivalence} \\
\Leftrightarrow \quad
&\gunravel{\decode{\unravelm{\fB}{\sigoriginal}}}{\sigtarget,\sigmag} \in L(\cA_\phi) \tag{by Lemma~$\gfp$-View Automaton} \\
\Leftrightarrow \quad
&\decode{\unravelm{\fB}{\sigoriginal}} \models \psi  \tag{by \lemmagfpbwd{Lemma~$\gfp$-Bwd}}\\
\Leftrightarrow \quad
&\fB \models \psi \tag{by $\gnlinvar[\sigoriginal]$-invariance of $\psi$.}
\end{align*}
Hence, $\phi$ and $\psi$ are logically equivalent, so
$\psi$ witnesses the $\gfp[\sigtarget,\sigmag]$-definability of $\phi$.

In the other direction, suppose that $\phi$ is $\gfp[\sigtarget,\sigmag]$-definable.
We must show that $\cA_\phi$ and $\cA'_\phi$ are equivalent
with respect to $\sigcode{\sigoriginal}{m}$-trees in $L(\cC)$.
For all $\sigcode{\sigoriginal}{m}$-trees $\tree$ in $L(\cC)$:
\begin{align*}
&\tree \in L(\cA_\phi) \\
\Leftrightarrow \quad
&\decode{\tree} \models \phi \tag{by Lemma~$\gnfp$-Fwd Automaton} \\
\Leftrightarrow \quad
&\decode{\gunravel{\decode{\tree}}{\sigtarget,\sigmag}} \models \phi \tag{by $\gfp$-definability of $\phi$} \\
\Leftrightarrow \quad
&\gunravel{\decode{\tree}}{\sigtarget,\sigmag} \in L(\cA_\phi) \tag{by Lemma $\gnfp$-Fwd Automaton} \\
\Leftrightarrow \quad
&\tree \in L(\cA'_\phi) \tag{by Lemma~$\gfp$-View Automaton}.
\end{align*}
Hence, we have shown the equivalence of $\cA_\phi$ and $\cA'_\phi$
with respect to $\sigcode{\sigtarget}{m}$-trees in $L(\cC)$,
which concludes the proof of correctness.

It remains to show that testing equivalence of $\cA_\phi$ and $\cA'_\phi$ with respect to trees in $L(\cC)$
can be done
in $\twoexptime$.
Using standard constructions from automata theory, we can construct 2-way alternating $\mu$-automata recognizing $L(\cC) \cap L(\cA_\phi) \cap \overline{L(\cA'_\phi)}$
and $L(\cC) \cap L(\cA'_\phi) \cap \overline{L(\cA_\phi)}$, with only a constant blow-up in size thanks to the use of alternating automata.
It then suffices to test emptiness of the language
accepted by the automaton
 for $L(\cC) \cap L(\cA_\phi) \cap \overline{L(\cA'_\phi)}$
and the automaton for $L(\cC) \cap L(\cA'_\phi) \cap \overline{L(\cA_\phi)}$.
This can be done in time exponential in the number of states and number of priorities (see~\cite{GradelHO02,Vardi98}).
Overall, this means that the decision procedure is in \twoexptime as claimed.
This completes the upper bound portion of
Theorem~\ref{thm:twoexpdefgfp}.

\myparagraph{Lower bound}
\input{twoexphard}


%% file: twoexphard.tex
Theorem~\ref{thm:effectivedefgfp} also states $\twoexp$ hardness of the $\gfp$-definability problem.
This is a straightforward reduction from satisfiability of $\gfp[\sigoriginal]$ sentences, which
is known to be $\twoexp$-hard~\cite{gfp}.

Fix a sentence $\phi_0$ over a signature $\sigma_0$ that is in $\gnfpk$
over its signature but not in $\gfp$.
Given $\phi \in \gfp[\sigoriginal]$ our reduction produces
$\phi \wedge \phi_0$, where we first modify the signatures so that $\sigoriginal$
is disjoint from $\sigma_0$.

We claim that $\phi \wedge \phi_0$ is definable in $\gfp[\sigma]$ iff
$\phi$ is unsatisfiable.

Clearly if $\phi$ is unsatisfiable $\phi \wedge \phi_0$ is definable in $\gfp[\sigoriginal]$.
In the other direction, suppose for the sake of contradiction that $\phi \wedge \phi_0$ is in $\gfp[\sigoriginal]$
but $\phi$ holds in a model $\fA$.
Then $\phi \wedge \phi_0$ is
$\ginvar[\sigoriginal]$-invariant.

We claim that $\phi_0$ is $\ginvar[\sigma_0]$-invariant.
Consider $\sigma_0$ structures $\fA_1$ and $\fA_2$ that are $\ginvar[\sigma_0]$-bisimilar,
and where $\fA_1$ satisfies $\phi_0$. We can assume $\fA_1 \cup \fA_2$ has a domain that is disjoint from that of $\fA$,
since taking an isomorphic copy does not change either the guarded bisimilarity or the truth of $\phi_0$.
Form the $\sigoriginal \cup \sigma_0$ structures $\fA'_1$ and $\fA'_2$
by interpreting the $\sigoriginal$ relations as in $\fA$.
We can extend the guarded bisimulation of $\fA_1$ and $\fA_2$ over $\sigma_0$ by the identity
mapping for elements in $\fA$, and this clearly gives a guarded bisimulation
of $\fA'_1$ and $\fA'_2$ over $\sigoriginal \cup \sigma_0$. $\fA'_1$ satisfies $\phi_0 \wedge \phi$, so
$\fA'_2$ satisfies $\phi_0 \wedge \phi$ as well, hence
$\fA_2$ satisfies $\phi_0$, completing the proof that $\phi_0$ is
$\ginvar[\sigma_0]$-invariant.

Since $\phi_0$ is $\ginvar[\sigma_0]$-invariant,
$\phi_0$ is definable in $\gfp[\sigma]$ by~\cite{GradelHO02}, a contradiction since $\phi_0$ was chosen to be outside of $\gfp$.


%% file: gnfp.tex
\section{Identifying \texorpdfstring{$\gnfpk$}{GFNPk} and \texorpdfstring{$\unfpk$}{UNFPk} sentences}\label{sec:gnfp}

We now turn to extending the prior results to $\gnfp$ and $\unfp$.
In order to make use of the back-and-forth approach described
in the previous section,
we need to know that it suffices to check for definability on structures of some bounded tree-width.
If we can focus on structures of bounded tree-width, we can make
use of tree automata and other results about regular tree languages in
solving the definability problem, since there is a fixed finite alphabet
for the encodings of such structures.

This was true for definability within $\gfp[\sigtarget,\sigmag]$, where the tree-width depended only on the signature $\sigtarget$.
For definability in $\gnfp$ and $\unfp$, it does not suffice to check structures of some fixed tree-width.
However, it does suffice for $\gnfpk$ and $\unfpk$.
Hence,
in this section, we will consider characterizing and deciding definability within
$\gnfpk$ and $\unfpk$.

The overall approach remains the same:
we apply the high-level algorithm of Proposition~\ref{prop:effectivedefhigh},
using the forward mapping of
Lemma~\ref{lemma:forward-gso}.
However, the unravelling and backward mapping for $\gnfpk$
is more technically challenging than the corresponding
constructions for $\gfp$.
The na\"ive backward mapping from $\Lmu$ into $\lfp$ is a straightforward structural induction, but it
fails to be in $\gnfpk$ for two reasons.
First, the inductive step for negation in the  na\"ive algorithm simply applies negation
to the recursively-produced formula. Clearly this can produce unguarded negation.
Similarly, the recursive step for fixpoints may use unguarded fixpoints.

This section focuses on these issues. Section~\ref{sec:plump} introduces a special unravelling construction called a plump unravelling that is more complicated than the block $k$-width guarded negation unravelling defined earlier, but still preserves all $\gnfpk$-sentences. Section~\ref{sec:backward} then provides the backward mapping for $\gnfpk$ by showing that problematic subformulas in the original $\Lmu$-formula can be eliminated, with the correctness of this simplification holding over tree codes of these plump unravellings.
Section~\ref{sec:gnfpdecidability} then summarizes the definability results that come from this.

As in the previous section, we will be working with signatures
$\sigmag \subseteq \sigtarget \subseteq \sigoriginal$, where $\sigma'$ is the
subsignature of $\sigma$ targeted by the backward mapping, with guards taken from $\sigmag$.

\subsection{Plump unravellings}\label{sec:plump}
We first need an appropriate notion of unravelling.
We use a variant of the block $k$-width guarded negation unravelling discussed
in Section~\ref{sec:prelims}, but we will need to assume that we have a certain copies of pieces of the structure,
over and above the usual copies present in every unravelling construction.
A property like this was defined in~\cite{lics15-gnfpi} for $\unfpk$,  called ``shrewdness'',
but we will need a more subtle property for $\gnfpk$, which we
 call  ``plumpness''.

\input{plump}

\subsection{Backward mapping}\label{sec:backward}
Returning to the components required for the application of Proposition~\ref{prop:effectivedefhigh},
we see that
Proposition~\ref{prop:bisim-plump} says that
the structure $\fB$ is $\sgnkinvar[\sigtarget,\sigmag]$-bisimilar to $\decode{\sunravelk{\fB}{\sigtarget,\sigmag}}$ as required
for an application of Proposition~\ref{prop:effectivedefhigh}.
Plumpness will come into play in the backward mapping, which is given by the following lemma:

\begin{lem}[$\gnfpk$-Bwd]%
\label{lemma:backwards-gnfp}
Given $\phi^{\mu} \in \Lmu[\sigcode{\sigoriginal}{m}]$,
relational signatures $\sigmag$ and $\sigtarget$ with
$\sigmag \subseteq \sigtarget \subseteq \sigoriginal$,
and $k \leq m$,
we can construct $\psi \in \gnfpk[\sigtarget,\sigmag]$
such that for all $\sigoriginal$-structures~$\fB$,
$\fB \models \psi \ \text{iff} \ \sunravelk{\fB}{\sigtarget,\sigmag} \models \phi^{\mu}$.
\end{lem}

There are two main challenges for this backward mapping,
compared to the backward mapping for $\gfp$ described earlier.
First, we must understand where negations occur in the $\mu$-calculus formula,
and ensure that these negations will translate into $\sigmag$-guarded negations.
Second, we must ensure that fixpoints reference only interface positions,
so they can be translated into $\sigmag$-guarded fixpoints.

For example, the original $\mu$-calculus formula $\phi^\mu$ can include subformulas
of the form
$\dmodality{\rho} \ExactLabel{\tau}$
where
$\tau$ is a set of unary relations from $\sigcode{\sigtarget}{k}$,
and $\ExactLabel{\tau}$ asserts $P$ for all $P \in \tau$
and $\neg P$ for all unary relations $P$ not in $\tau$.
This would be problematic in a na\"ive backward mapping,
since the backward translation of some $\neg R_{i_1,\dots,i_n}$ would be converted into an unguarded negation
$\neg R(x_{i_1},\dots,x_{i_n})$.
On the other hand the formula
$\dmodality{\rho} \GNLabel{\tau}$
where $\GNLabel{\tau}$ asserts $P$ for all $P \in \tau$
but only asserts $\neg P$ for unary relations $P$ that are not in $\tau$
but whose indices are $\sigmag$-guarded by some $P' \in \tau$
would be unproblematic,
since this could be translated to a formula with $\sigmag$-guarded negation.
The key observation is that from an interface node in a plump tree,
these two formulas are equivalent:
if $\tree, v \models \dmodality{\rho} \GNLabel{\tau}$ at any interface node $v$,
then plumpness ensures that if there is some $\rho$-child $w'$ of $v$ with label $\tau'$
satisfying $\GNLabel{\tau}$,
then there is a $\rho$-child $w$ of $v$ with label $\tau$ satisfying $\ExactLabel{\tau}$---it can be checked that $\tau$ is a $(\sigmag,\codom{\rho})$-safe restriction of~$\tau'$.
Thus the main technical work in this section is to show that problematic subformulas like
$\ExactLabel{\tau}$ can be eliminated,
with the correctness of the simplification holding at least over plump trees.

Hence, the structure of the proof is as follows.
We first define a fragment of $\Lmu[\sigcode{\sigtarget}{k}]$,
which we call ``$\gnfpk[\sigtarget,\sigmag]$-safe formulas'',
and show that these formulas can be converted into $\gnfpk[\sigtarget,\sigmag]$.
After this, we will show that any $\mu$-calculus formula over $\sigcode{\sigoriginal}{m}$-trees
can be converted to a $\gnfpk[\sigtarget,\sigmag]$-safe form
that is equivalent on plump, $\sigmag$-guarded-interface $\sigcode{\sigtarget}{k}$-trees
like $\sunravelk{\fB}{\sigtarget,\sigmag}$.

We say an $\Lmu[\sigcode{\sigtarget}{k}]$ formula is \emph{$\gnfpk[\sigtarget,\sigmag]$-safe for interface nodes}
(respectively, \emph{$\gnfpk[\sigtarget,\sigmag]$-safe for bag nodes}) if
\begin{itemize}
\item every occurrence of a fixpoint $\lambda X. \chi$ or fixpoint variable $X$ is
in the scope of an even (respectively, odd) number of modalities and has one of the following forms:
\begin{align*}
&P \wedge D_m \wedge \lambda X . \chi &\quad &P \wedge D_m \wedge \neg \lambda X . \chi \\
&P \wedge D_m \wedge X &\quad &P \wedge D_m \wedge \neg X
\end{align*}
for some $P$ encoding a relation in $\sigmag$ with $\indices{P} = \set{1, \dots, m}$;
\item every negation has one of the following forms:
\begin{align*}
&P' \wedge \neg R &\quad &P \wedge D_m \wedge \neg X \\
&P'' \wedge \neg \dmodality{\rho} \chi &\quad &P \wedge D_m \wedge \neg \lambda X. \chi
\end{align*}
for some $P,P',P''$ encoding relations in $\sigmag$ and with
$\indices{P} = \set{1,\dots,m}$,
$\indices{P'} \supseteq \indices{R}$,
$\indices{P''} = \dom{\rho}$;

\item every modality has one of the following forms:
\begin{align*}
&P \wedge \dmodality{\rho} \chi \\
&P \wedge \neg \dmodality{\rho} \chi
\end{align*}
for some $P$ encoding a relation in $\sigmag$ with $\indices{P} = \dom{\rho}$.
\end{itemize}
Note that these safe formulas impose $\sigmag$-guardedness conditions
on fixpoints, negations, and modalities.

We already defined $\indices{\chi}$ for $\chi$ a propositional variable.
We generalize this to all $\gnfpk[\sigtarget,\sigmag]$-safe formulas $\chi$ as follows:
\begin{itemize}
\item $\indices{\top} = \indices{\bot} := \emptyset$;
\item $\indices{D_m} = \set{1,\dots,m}$;
\item $\indices{\chi_1 \wedge \chi_2} = \indices{\chi_1 \vee \chi_2} := \indices{\chi_1} \cup \indices{\chi_2}$;
\item $\indices{P \wedge \neg R} := \indices{P}$;
\item $\indices{P \wedge \dmodality{\rho} \chi} = \indices{P \wedge \neg \dmodality{\rho} \chi} := \dom{\rho}$;
\item $\indices{P \wedge D_m \wedge X} =\indices{P \wedge D_m \wedge \neg X}= \set{1,\dots,m}$;
\item $\indices{P \wedge D_m \wedge \lambda X . \chi} = \indices{P \wedge D_m \wedge \neg \lambda X . \chi} := \set{1,\dots,m}$.
\end{itemize}
These names determine the free first-order variables in the backwards translation.

For each fixpoint variable $X$ we introduce
multiple second-order variables of the form $X_{j,P}$ for $0 \leq j \leq k$ and $P \in \sigcode{\sigtarget}{k}$ with $\indices{P} = \set{1,\dots,j}$,
as we did for the $\gfp$ backward mapping.
Let $\tbackward{X}$ denote the set of these new second-order variables based on $X$.
A set $V$ of nodes in $\sunravelk{\fB}{\sigtarget,\sigmag}$ is a \emph{safe valuation for a free variable $X$}
if it
\begin{enumerate}
\item contains only interface nodes,
and
\item if it contains an interface node then it contains every interface node
that is the root of a bisimilar subtree.
\end{enumerate}
We write $\tbackward{V}$ for its representation in $\fB$
(as we did for $\gfp$).

We will now show that from a $\gnfpk[\sigtarget,\sigmag]$-safe formula in $\Lmu[\sigcode{\sigtarget}{k}]$
we can produce $\gnfpk[\sigtarget,\sigmag]$ formulas described below.
As in the backward mapping for $\gfp$, we must generate a family of formulas
during the inductive translation.
This time we must not only deal with different domain sizes in the nodes of
$\sunravelk{\fB}{\sigtarget,\sigmag}$,
but we must also distinguish between bag nodes and interface nodes.
We use $\tbackward{\phi}_m$ to indicate a formula related to a bag node of size $m$,
and $\tbackward{\phi}_{m,I}$ to indicate a formula related to an interface node with size $m$.

\newcommand{\plumpunravelbshort}{\cU(\fB)}

\begin{lem}\label{lemma:gnfp-backward}
Let $\phi \in \Lmu[\sigcode{\sigtarget}{k}]$ be $\gnfpk[\sigtarget,\sigmag]$-safe for interface nodes
(respectively, bag nodes)
with free second-order variables $\vec{X}$.
For each $m \leq k$,
we can construct a $\gnfpk[\sigtarget,\sigmag]$-formula
$\tbackward{\phi}_{m,I}(x_{1},\dots,x_{m},\tbackward{\vec{X}})$
(respectively, $\tbackward{\phi}_{m}(x_{1},\dots,x_{m},\tbackward{\vec{X}})$)
such that
for all $\sigoriginal$-structures $\fB$,
for all safe valuations $\vec{V}$ of $\vec{X}$,
and for all nodes $v$ in $\plumpunravelbshort = \sunravelk{\fB}{\sigtarget,\sigmag}$ with $\size{\elem{v}} = m$
and label $\tau$:
\begin{align}
\tag{1} &\parbox{4.5cm}{if $v$ is an interface node:}
&\plumpunravelbshort,v,\vec{V} &\models \phi \ &&\Rightarrow \
&\fB,\elem{v},\tbackward{\vec{V}} &\models \tbackward{\phi}_{m,I} \\
\tag{2} &\parbox{4.5cm}{if $v$ is an interface node:}
&\fB,\elem{v},\tbackward{\vec{V}} &\models \tbackward{\phi}_{m,I} \ &&\Rightarrow \
&\plumpunravelbshort,v,\vec{V} &\models \phi \\
\tag{3} &\parbox{4.5cm}{if $v$ is a bag node:}
&\plumpunravelbshort,v,\vec{V} &\models \phi \ &&\Rightarrow \
&\fB,\elem{v},\tbackward{\vec{V}} &\models \tbackward{\phi}_{m} \\
\tag{4} &\parbox{4.5cm}{if $v$ is a bag node and $\tau$ \\ encodes $\atypesig{\elem{v}}{\fB}{\sigtarget}$:}
&\fB,\elem{v},\tbackward{\vec{V}} &\models \tbackward{\phi}_{m} \ &&\Rightarrow \
&\plumpunravelbshort,v,\vec{V} &\models \phi .
\end{align}
Moreover, if $\indices{\phi} = \set{i_1,\dots,i_n} \subseteq \set{1,\dots,m}$
then we have $\free{\tbackward{\phi}_{m,I}} = \free{\tbackward{\phi}_{m}} = \set{x_{i_1}, \dots, x_{i_n}}$,
and any subformula of $\tbackward{\phi}_{m,I}$ or $\tbackward{\phi}_m$
that begins with an existential quantifier and is not directly below another existential quantifier
is strictly $\sigmag$-answer-guarded,
and any negation is strictly $\sigmag$-guarded.
\end{lem}

It is helpful to compare the statement of this lemma to the analogous one for $\gfp$ in Lemma~\ref{lemma:backwards-gfp}. Note that we cannot replace conditions (1)--(4) by, say, the stronger and simpler statement $\sunravelk{\fB}{\sigtarget,\sigmag},v,\vec{V} \models \phi$ iff 
$\fB,\elem{v},\tbackward{\vec{V}} \models \tbackward{\phi}_{m}$.
In particular, if $v$ is a bag node, we cannot guarantee that $\sunravelk{\fB}{\sigtarget,\sigmag},v,\vec{V} \models \phi$ implies
$\fB,\elem{v},\tbackward{\vec{V}} \models \tbackward{\phi}_{m}$,
since bag nodes in a plump unravelling are based on safe restrictions of the atomic information, rather than an exact copy. The weaker condition in (4) is sufficient for our purposes.

Although the following inductive proof is long, it does not use any complicated machinery.
We are translating $\Lmu$-formulas that talk about the tree codes of plump unravellings into
formulas talking about relational structures (e.g.~diamond modalities in the $\mu$-calculus formula become existentially quantified formulas in the translation).
The conditions for $\gnfpk$-safe $\Lmu$-formulas ensure that
this translation always stays within $\gnfpk$.

\begin{proof}[Proof of Lemma~\ref{lemma:gnfp-backward}]
We start with the inductive translation,
and then give the proof of correctness for some illustrative cases.

\paragraph*{Translation}
If $\indices{\phi} \not\subseteq \set{1,\dots,m}$
then
$\tbackward{\phi}_{m,I} = \tbackward{\phi}_{m} = \bot$.
Otherwise,
we proceed by induction on the structure of the $\gnfpk[\sigtarget,\sigmag]$-safe formula $\phi$
to define $\tbackward{\phi}_{m,I}$ and $\tbackward{\phi}_m$.
\begin{itemize}
\item If $\phi = D_j$,
then $\tbackward{\phi}_{m,I} = \tbackward{\phi}_{m}$ is  $\top$ if $m = j$ and $\bot$ otherwise.
\item If $\phi = R_{i_1 \dots i_n}$,
then \begin{align*}
\tbackward{\phi}_{m,I} = \tbackward{\phi}_m &:= R \, x_{i_1} \dots x_{i_n} .
\end{align*}
\item If $\phi = P \wedge \neg R$,
then
$\tbackward{\phi}_{m,I} := \tbackward{(P)}_{m,I} \wedge \neg (\tbackward{(P)}_{m,I} \wedge \tbackward{(R)}_{m,I})$
and
$\tbackward{\phi}_m := \tbackward{(P)}_{m} \wedge \neg (\tbackward{(P)}_{m} \wedge \tbackward{(R)}_m)$.
\item If $\phi = P \wedge D_j \wedge X$ and $m = j$,
then $\tbackward{\phi}_{m,I} :=
\tbackward{(P)}_{j,I} \wedge X_{j,P}(x_1,\dots,x_j)$;
otherwise, if $m \neq j$, then $\tbackward{\phi}_{m,I} := \bot$.
Similarly for $\phi = P \wedge D_j \wedge \neg X$.
\item The translation commutes with $\vee$ and $\wedge$
for each $m$.
\item If $\phi = P \wedge \dmodality{\rho} \chi$
where $\dom{\rho} = \set{i_1,\dots,i_n}$,
then
\begin{align*}
\tbackward{\phi}_{m,I} &:=
\tbackward{(P)}_{m,I}  \wedge \bigvee_{n \leq j \leq k} \exists y_1 \dots y_j .
	\left( \tbackward{(P)}_{m,I} \wedge \tbackward{\chi}_j(y_1,\dots,y_j) [x_i / y_{\rho(i)} : i \in \dom{\rho}] \right) \\
\tbackward{\phi}_{m} &:=
\tbackward{(P)}_{m} \wedge \tbackward{\chi}_{n,I}(y_1,\dots,y_n) [x_i / y_{\rho(i)} : i \in \dom{\rho}].
\end{align*}
Similarly for $\phi = P \wedge \neg \dmodality{\rho} \chi$.
\item If $\phi = P \wedge D_j \wedge \mu Y . \chi$,
then
$\tbackward{\phi}_{m,I} := \tbackward{(P)}_{m,I} \wedge \tbackward{(D_j)}_{m,I} \wedge [\LFPA{Y_{j,P}}{y_1,\dots,y_j} . S_{\mu Y.\chi}](x_1,\dots,x_j)$
where
$S_{\mu Y . \chi}$
is a system
consisting of equations
\[Y_{n,P'}, y_1 \dots, y_n := \tbackward{(P')}_{n,I}(y_1,\dots,y_n) \wedge \tbackward{\chi}_{n,I}(y_1,\dots,y_n) \]
for each $Y_{n,P'}$ in $\tbackward{Y}$.
Similarly for $\phi = P \wedge D_j \wedge \neg \mu Y . \chi$.
\item If $\phi = P \wedge D_j \wedge \nu Y . \chi$,
then
$\tbackward{\phi}_{m,I} := \tbackward{(P)}_{m,I} \wedge \tbackward{(D_j)}_{m,I} \wedge \neg [\LFPA{Y_j}{y_1,\dots,y_j} . S_{\nu Y . \chi}](x_1,\dots,x_j)$
where
$S_{\nu Y . \chi}$
is a system
consisting of equations
\[Y_{n,P'}, y_1 \dots, y_n := \tbackward{(P')}_{n,I}(y_1,\dots,y_n) \wedge \neg \tbackward{\chi}_{n,I}(y_1,\dots,y_n)[\neg Y_{n'',P''} / Y_{n'',P''} : Y_{n'',P''} \in \tbackward{Y}] \]
for each $Y_{n,P'}$ in $\tbackward{Y}$.
Similarly for $\phi = P \wedge D_j \wedge \neg \nu Y . \chi$.
\end{itemize}

\noindent
The formulas $\tbackward{\phi}_m$ and $\tbackward{\phi}_{m,I}$
are in $\gnfp[\sigtarget,\sigmag]$ of width $k$, but are not necessarily in \nf.
However, it can be checked that the negations are all strictly $\sigmag$-guarded;
for instance, in the case of $\phi = P \wedge \neg R$,
then $\tbackward{\phi} = \tbackward{(P)}_{m} \wedge \neg (\tbackward{(P)}_{m} \wedge \tbackward{(R)}_m)$,
which is a strictly $\sigmag$-guarded negation since
$\indices{P} \supseteq \indices{R}$ and $P$ encodes a $\sigmag$-relation
(by definition of $\gnfpk[\sigtarget,\sigmag]$-safety).
Also, every subformula
starting with an existential quantifier
and not immediately below another existential quantifier is
strictly $\sigmag$-answer-guarded.
Thus,
we can convert to \nf using Proposition~\ref{prop:nfwidth}
without increasing the width,
so the formulas are in $\gnfpk[\sigtarget,\sigmag]$ as desired.

\paragraph*{Proof of correctness}

We now give the proof of correctness for some illustrative cases.
We write IH1--IH4 to denote the application of the inductive hypothesis based on properties (1)--(4). 
We will assume that $\indices{\phi} \subseteq \set{1,\dots,m}$.

\myparagraph{Atom}

Consider the base case
when $\phi = R_{i_1 \dots i_n}$.

Recall that in the plump unravelling,
the label at a node $v$ is based on $\atypesig{\fB}{\elem{v}}{\sigtarget}$.
For an interface node, the label is \emph{exactly} the encoding of $\atypesig{\fB}{\elem{v}}{\sigtarget}$,
which is enough to ensure that properties (1) and (2) hold.
For a bag node, the label is based on a safe \emph{restriction} of $\atypesig{\fB}{\elem{v}}{\sigtarget}$,
which can result in the label at $v$ including only a subset of the encoding of $\atypesig{\fB}{\elem{v}}{\sigtarget}$.
Property (3) follows from this.
Not all bag nodes $v$ such that $\fB,\elem{v} \models \tbackward{\phi}_{m}$
would satisfy $\plumpunravelbshort,v \models \phi$.
However, for bag nodes $v$ that do represent $\atypesig{\fB}{\elem{v}}{\sigtarget}$ in its entirety,
$\fB,\elem{v} \models \tbackward{\phi}_{m}$
implies $\plumpunravelbshort,v \models \phi$.
Hence, property (4) holds.

\myparagraph{Negated atom}

Consider the case when
$\phi = P \wedge \neg R$ with $\indices{P} \supseteq \indices{R} = \set{i_1,\dots,i_n}$ and
$\tbackward{\phi}_{m,I} := \tbackward{(P)}_{m,I} \wedge \neg (\tbackward{(P)}_{m,I} \wedge \tbackward{(R)}_{m,I})$
and
$\tbackward{\phi}_m := \tbackward{(P)}_{m} \wedge \neg (\tbackward{(P)}_{m} \wedge \tbackward{(R)}_m)$.

For (1), (2), and (4) correctness essentially follows from the inductive hypothesis.
Consider~(3).
Suppose $\plumpunravelbshort,v \models \phi$
for a bag node $v$.
By the properties of the plump unravelling,
the label at a bag node $v$ is the encoding of a safe restriction of
$\atypesig{\fB}{\elem{v}}{\sigtarget}$---this means that
for any set of elements that is still $\sigmag$-guarded after the restriction,
the label must encode
exactly the atoms in $\atypesig{\fB}{\elem{v}}{\sigtarget}$ about these elements.
Because we know that $i_1,\dots,i_n$ is $\sigmag$-guarded by $P$ in the label at $v$,
this means that the label at $v$ must reflect $\atypesig{\fB}{\elem{v}}{\sigtarget}$ exactly
over the elements corresponding to $i_1,\dots,i_n$.
Hence, $\fB,\elem{v} \models \tbackward{(R)}_m$ iff
$\plumpunravelbshort,v \models R$,
so $\fB,\elem{v} \models \neg \tbackward{(R)}_m$.
By IH3 applied to $P$,
we also have
$\fB,\elem{v} \models \tbackward{(P)}_m$,
so $\fB, \elem{v} \models \tbackward{\phi}_{m}$ as desired.

\myparagraph{Modality}

Consider the case when
$\phi = P \wedge \neg \dmodality{\rho} \chi$
with $\indices{P} = \dom{\rho} = \set{i_1,\dots,i_n}$ and
\begin{align*}
\tbackward{\phi}_{m,I} &:=
\tbackward{(P)}_{m,I}  \wedge \neg \bigvee_{n \leq j \leq k} \exists y_1 \dots y_j .
	\left( \tbackward{(P)}_{m,I} \wedge \tbackward{\chi}_j(y_1,\dots,y_j) [x_i / y_{\rho(i)} : i \in \dom{\rho}] \right) \\
\tbackward{\phi}_{m} &:=
\tbackward{(P)}_{m} \wedge \neg \tbackward{\chi}_{n,I}(y_1,\dots,y_n) [x_i / y_{\rho(i)} : i \in \dom{\rho}].
\end{align*}

For (1), suppose $v$ is an interface node and $\plumpunravelbshort,v \models \phi$.
Let $\elem{v} = a_1 \dots a_m$.
Then by IH1, $\fB,\elem{v} \models \tbackward{(P)}_{m,I}$.
Suppose for the sake of contradiction that
\[\exists y_1 \dots y_j .
	\left( \tbackward{\chi}_j(y_1,\dots,y_j) [x_i / y_{\rho(i)} : i \in \dom{\rho}] \right)\]
	for some $n \leq j \leq k$.
Then there are elements $b_1,\dots,b_j$ in $\fB$ such that
$\fB,b_1,\dots,b_j \models \tbackward{\chi}_j(y_1,\dots,y_j)$,
where $a_i = b_{\rho(i)}$.
By the properties of the unravelling,
there is a $\rho$-child $w$ of $v$ such that $\elem{w} = b_1, \dots, b_j$
and such that the label at $w$ is the encoding of $\atypesig{\fB}{\elem{w}}{\sigtarget}$.
Hence, by IH4, $\plumpunravelbshort,w \models \chi$
(we can apply IH4 since we have chosen $w$ such that it represents the full atomic type of $\elem{w}$ in $\fB$).
This means $\plumpunravelbshort, v \models \dmodality{\rho} \chi$, a contradiction.
Therefore, it must be the case that $\fB, \elem{v} \models \tbackward{\phi}_{m,I}$.

For (2), suppose $v$ is an interface node and $\fB, \elem{v} \models \tbackward{\phi}_{m,I}$.
Let $\elem{v} = a_1 \dots a_m$.
By IH2, $\plumpunravelbshort, v \models P$.
Suppose for the sake of contradiction that
$\plumpunravelbshort, v \models \dmodality{\rho} \chi$.
Then there is some $\rho$-child $w$ of $v$ such that $\plumpunravelbshort, w \models \chi$.
The node $w$ must be a bag node and must have some domain predicate $D_j$.
Let $b_1 \dots b_j$ be the elements from $\fB$ represented there, with $a_i = b_{\rho(i)}$.
By IH3, this means that $\fB, \elem{w} \models \tbackward{\chi}_j(y_1,\dots,y_j)$,
and hence $\fB, \elem{v} \models \exists y_1 \dots y_j . (\tbackward{\chi}_j(y_1,\dots,y_j) [x_i / y_{\rho(i)} : i \in \dom{\rho}] )$,
a contradiction of the fact that $\fB, \elem{v} \models \tbackward{\phi}_{m,I}$.
Hence, $\plumpunravelbshort, v \models \phi$ as desired.

For (3), suppose that $v$ is a bag node and $\plumpunravelbshort, v \models \phi$,
with $\elem{v} = a_1 \dots a_m$.
By IH3, $\fB, \elem{v} \models \tbackward{(P)}_m$.
Suppose for the sake of contradiction that $\fB, \elem{v} \models \tbackward{\chi}_{n,I} [x_i / y_{\rho(i)} : i \in \dom{\rho}]$.
Let $w$ be a $\rho$-child of $v$ with domain of size $\size{\codom{\rho}} = n$,
and satisfying $a_i = \rho(i)$
(this must exist by properties of unravelling).
Hence, $\fB, \elem{w} \models \tbackward{\chi}_{n,I}(y_1,\dots,y_n)$.
By IH2, this means that $\plumpunravelbshort, w \models \chi$,
and hence $\plumpunravelbshort, v \models \dmodality{\rho} \chi$, a contradiction.
Therefore $\fB, \elem{v} \models \tbackward{\phi}_{m}$ as desired.

For (4), suppose $v$ is a bag node representing $\atypesig{\fB}{\elem{v}}{\sigtarget}$
and $\fB, \elem{v} \models \tbackward{\phi}_{m}$.
Let $\elem{v} = a_1 \dots a_m$.
By IH4, $\plumpunravelbshort, v \models P$.
Suppose for the sake of contradiction that $\plumpunravelbshort, v \models \dmodality{\rho} \chi$.
Then there is some $\rho$-child $w$ of $v$ such that $\plumpunravelbshort, w \models \chi$.
The node $w$ must be an interface node and must have domain $D_n$ (where $n = \size{\codom{\rho}}$).
Let $b_1 \dots b_n$ be the elements from $\fB$ represented at $w$, with $a_i = b_{\rho(i)}$.
By IH1, $\fB, \elem{w} \models \tbackward{\chi}_{n,I}(y_1,\dots,y_n)$,
so $\fB, \elem{v} \models \tbackward{\chi}_{n,I}(y_1, \dots, y_n)[x_i / y_{\rho(i)} : i \in \dom{\rho}]$,
a contradiction.
Thus, $\plumpunravelbshort, v \models \phi$.

\myparagraph{Fixpoint}

Consider the case $\phi = P \wedge D_j \wedge \mu Y . \chi$,
where
\[
\tbackward{\phi}_{m,I} := \tbackward{(P)}_{m,I} \wedge \tbackward{(D_j)}_{m,I} \wedge [\LFPA{Y_{j,P}}{y_1,\dots,y_j} . S_{\mu Y.\chi}](x_1,\dots,x_j)
\]
and
$S_{\mu Y . \chi}$
is a system
consisting of equations
\[Y_{n,P'}, y_1 \dots, y_n := \tbackward{(P')}_{j,I}(y_1,\dots,y_n) \wedge \tbackward{\chi}_{n,I}(y_1,\dots,y_n) \]
for each $Y_{n,P'}$ in $\tbackward{Y}$.

\newcommand{\betaapproxchij}{(\tbackward{\chi}_{j,I})^\beta}
\newcommand{\deltaapproxchii}{(\tbackward{\chi}_{i,I})^\delta}

We must prove properties (1) and (2).
For ordinals $\beta$,
we write $\chi^{\beta}$ for the $\beta$-approximant of the fixpoint $\mu Y . \chi$
and $\betaapproxchij$ for the $\beta$-approximant of
$[\LFPA{Y_j}{y_1,\dots,y_j} . S_{\mu Y.\chi}](x_1,\dots,x_j)$.
We first show the result for the fixpoint approximants.
That is, for all interface nodes $v$ with $\size{\elem{v}} = j$,
$\fB, \elem{v} \models \betaapproxchij$
iff
$\plumpunravelbshort,v \models \chi^{\beta}$.
We proceed by induction on the fixpoint approximant~$\beta$;
we will refer to this as the inner induction to distinguish it from
the outer induction on the structure of the formula.

For $\beta = 0$, the result follows by the outer inductive hypothesis IH1 and IH2
applied to the formulas that result from substituting $\bot$ for $Y$ in $\chi$,
and $\bot$ for $Y_0, \dots, Y_k$ in $\tbackward{\chi}_j$.

Now assume $\beta > 0$ is a successor ordinal $\beta = \delta+1$.

For (1), assume $\plumpunravelbshort,v \models \chi^{\beta}$
for $v$ an interface node with $\size{\elem{v}} = j$.
Then $\plumpunravelbshort,v, V_{\delta} \models \chi$
where
$
V_{\delta} := \set{ w : \text{$w$ is an interface node and $\plumpunravelbshort,w \models \chi^{\delta}$} }
$
is the valuation
for $Y$.
Note that this is a safe valuation:
it clearly contains only interface nodes, and if it contains interface node $w$,
then it contains all $w'$ such that the subtrees rooted at $w$ and $w'$ are bisimilar
because $\chi$ is in the $\mu$-calculus,
and the $\mu$-calculus is bisimulation invariant.
By the outer inductive hypothesis,
this implies that $\fB,\elem{v}, \tbackward{V}_{\delta} \models \tbackward{\chi}_j$
for $\tbackward{V}_{\delta} = (V^\delta_0, \dots, V^\delta_k)$
and $V^\delta_i = \set{\elem{w} : \text{$\elem{w} =  i$ and $w \in V$}}$.
However, by the inner inductive hypothesis,
$V^\delta_i
= \sset{\elem{w} : \text{$\elem{w} =  i$ and $\fB, \elem{w} \models \deltaapproxchii$}}$.
This means that $\fB, \elem{v} \models \betaapproxchij$ as desired.

Next, assume $\fB,\elem{v} \models \betaapproxchij$
for $v$ an interface node with $\size{\elem{v}} = j$.
Then $\fB,\elem{v},\tbackward{V}_\delta \models \tbackward{\chi}_j$
where
$V^\delta_i = \{\elem{w} : \text{$\elem{w} =  i$ and $\fB, \elem{w} \models \deltaapproxchii$}\}$.
Define $V_\delta = \set{ w : \text{$w \in V^\delta_i$ for some $i$}}$.
By the inner inductive hypothesis, $V_\delta$ is equivalent to the valuation
$\set{ w : \plumpunravelbshort, w \models \chi^\delta }$.
By the bisimulation-invariance of $\mu$-calculus, this is a safe valuation for $Y$.
Hence, the outer inductive hypothesis implies that
$\plumpunravelbshort, v, V_\delta \models \chi$,
and hence $\plumpunravelbshort, v \models \chi^\beta$ as desired.

The proof is similar when $\beta$ is a limit ordinal.

The overall result for this case follows by
appealing to the fact that the least fixpoint
corresponds to some $\beta$-approximant,
and noting that because the fixpoint references only interface nodes,
it is correct to add a $\sigmag$-guardedness requirement to each formula
in the simultaneous fixpoint.

Note that the greatest fixpoint cases rely on the fact that
a greatest fixpoint can be expressed using negation and least fixpoint,
e.g.~$\nu Y . \chi = \neg \mu Y . \neg \chi [\neg Y / Y]$.
This is reflected in the translation.
The only additional technicality in the translation is that we add an extra $\sigmag$-guard
at the beginning of each formula in the simultaneous fixpoint.
\end{proof}

Lemma~\ref{lemma:backwards-gnfp} follows by
converting $\phi^\mu$ to a $\gnfpk[\sigtarget,\sigmag]$-safe form $\psi$
(see Lemma~\ref{lemma:gnfp-safe} below)
that is equivalent over plump, $\sigmag$-guarded-interface $\sigcode{\sigtarget}{k}$-trees,
and then using Lemma~\ref{lemma:gnfp-backward}, taking $\tbackward{\psi} := \tbackward{\phi}_{0,I}$ as the
desired sentence in $\gnfpk[\sigtarget,\sigmag]$.


\myparagraph{Conversion to $\gnfpk[\sigtarget,\sigmag]$-safe formula}
It remains to show that we can actually
convert a $\mu$-calculus formula $\phi^\mu$
into a $\gnfpk$-safe formula.

\begin{lem}\label{lemma:gnfp-safe}
Given $\phi^\mu \in \Lmu[\sigcode{\sigoriginal}{m}]$,
we can construct $\phi \in \Lmu[\sigcode{\sigtarget}{k}]$
that is $\gnfpk[\sigtarget,\sigmag]$-safe for interface nodes
and such that for all plump, $\sigmag$-guarded-interface $\sigcode{\sigtarget}{k}$-trees $\cT$,
we have
$\cT \models \phi \text{ iff } \cT \models \psi$.
\end{lem}

This requires a series of transformations, described below.
At each stage, we ensure equivalence with $\phi^\mu$,
at least over plump trees.

Before we give these transformations,
we review and introduce some additional notation.

We utilize vectorial (simultaneous) fixpoints, using the standard notation.
We also use the standard notation for the box modality and greatest fixpoint operator,
writing $\bmodality{\rho} \psi$ as an abbreviation for $\neg \dmodality{\rho} \neg \psi$,
and $\nu Y . \psi$ for $\neg \mu Y . \neg \psi[\neg Y / Y]$.

Recall that each node label is a set of propositions $\tau$, and each edge label is a mapping $\rho$ describing
the relationship between names in neighboring nodes.
As before, we write $\EdgeLabels$ to denote the set of functions $\rho$ such that the binary predicate $E_\rho$ is in $\sigcode{\sigtarget}{k}$,
and we write $\NodeLabels$ for the set of internally consistent node labels.
We write $\BagNodeLabels$ and $\IntNodeLabels$ for the subset of $\NodeLabels$
allowed in a bag node and interface node, respectively, in a $\sigmag$-guarded-interface tree
of width $k$.
We write $\bagextend{\tau_0}$ for the subset of labels from $\BagNodeLabels$
that contain the encoded atoms of $\tau_0$,
i.e.~labels that extend $\tau_0$ with additional encoded atoms
(but may differ in domain size).
Finally, for $S \subseteq \EdgeLabels \times Q$,
we write $\EdgeLabelsRestrict{S}$ to denote
the set of edge labels appearing in $S$.

Given $\tau \in \NodeLabels$ and a set of names $I$,
we let $\guardg{I}{\tau}$ denote some $P \in \tau$ that is in $\sigmag$
and satisfies $\indices{P} = I$
($\bot$ if no such $P$ exists, and $\top$ if $I$ is of size at most~1).
We let $\GuardDom{\rho}{\tau}$ denote $\guardg{\dom{\rho}}{\tau}$,
and let $\GuardRng{\rho}{\tau}$ be the result of
replacing the names appearing in $\guardg{\dom{\rho}}{\tau}$ according to $\rho$.
The idea is that these are macros that give a $\sigmag$-guard related to $\tau$.

We also define some auxiliary formulas to improve readability in the
formulas in this section.
For $\tau \in \NodeLabels$,
we define
\begin{align*}
\ExactLabel{\tau} &:=
\bigwedge_{P \in \tau} P \wedge \bigwedge_{P \in \propsneg{\tau}} \neg P
\\
\GNLabel{\tau} &:= \bigwedge_{P \in \tau} P \wedge
\bigwedge_{P \in \propsguardedneg{\tau}} \quad \bigwedge_{\substack{\indices{P} \subseteq I \subseteq \set{1,\dots,k} \\ \text{s.t.} \guardg{I}{\tau} \neq \bot }}
\left(\guardg{I}{\tau} \wedge \neg P \right) .
\end{align*}
where $\propsneg{\tau}$ consists of unary propositions from $\sigcode{\sigtarget}{k}$
that do not appear in $\tau$,
and $\propsguardedneg{\tau}$ consists of unary propositions $P$ from $\sigcode{\sigtarget}{k}$
that do not appear in $\tau$ but use names that are $\sigmag$-guarded by some $P' \in \tau$.
Both $\ExactLabel{\tau}$ and $\GNLabel{\tau}$ assert all of the positive information
about the propositions in $\tau$.
However, $\GNLabel{\tau}$
only asserts some of the negative information,
namely it only asserts that some proposition is missing from $\tau$
if the indices used by that proposition are $\sigmag$-guarded in $\tau$.
Hence, $\GNLabel{\tau}$ can be seen as an approximation of $\ExactLabel{\tau}$:
if $\cT, v \models \ExactLabel{\tau}$ then $\cT, v \models \GNLabel{\tau}$,
but the converse does not always hold.
Note that $\GNLabel{\tau}$ is $\gnfpk[\sigtarget,\sigmag]$-safe
but $\ExactLabel{\tau}$ is not.

We now proceed with the series of transformations taking $\psi$ to a $\gnfpk[\sigtarget,\sigmag]$-safe version
that is equivalent over plump trees.

\paragraph*{Step 1: Refinement to $\sigmag$-guarded-interface $\sigcode{\sigtarget}{k}$-trees}

By Theorem~\ref{thm:muaut}, there is some $\mu$-automaton $\cA'$
that is equivalent to $\phi^\mu$.
This automaton runs on $\sigcode{\sigoriginal}{m}$-trees.
However, for the purposes of the backward mapping,
we are only interested in it running on $\sigcode{\sigtarget}{k}$-trees---and
more specifically on plump $\sigmag$-guarded-interface trees that
encode plump $\sgnkinvar[\sigtarget,\sigmag]$-unravellings of
$\sigoriginal$-structures.

Therefore, in this first step, we make some straightforward modifications to this automaton
that are correct on $\sigmag$-guarded-interface $\sigcode{\sigtarget}{k}$ trees,
and then convert it back into an $\Lmu[\sigcode{\sigtarget}{k}]$-formula.
The shape of the resulting formula is described in Claim~\ref{claim:step2}, but we defer the formal statement until after we have described these modifications.
Along the way, we also introduce some notation and terminology
(e.g.\ bag states, interface states, etc.)
that we will use in the later steps.

Let $\cA'$ be the $\mu$-automaton
with state set $Q'$, transition function $\delta'$, and priority function $\Omega'$
that runs on $\sigcode{\sigoriginal}{m}$-trees and is equivalent to $\phi^\mu$ (see Theorem~\ref{thm:muaut}).

We start by refining this automaton to run on $\sigcode{\sigtarget}{k}$-trees.
This means that we can limit the alphabet for the automaton to just node labels in $\NodeLabels$
and edges in $\EdgeLabels$, and eliminate all of the transition function information
related to labels outside of this.

We then refine it to run on $\sigmag$-guarded-interface trees.
In particular, we can modify the automaton so that states in interface nodes
are disjoint from the states of the automaton in bag nodes.
We will refer to \emph{interface states} and \emph{bag states}
as appropriate.
Because a guarded-interface tree alternates between
interface nodes and bag nodes,
it is possible to enforce that
\begin{enumerate}
\item bag states are assigned priority 0,
\item the initial state is assigned the maximum priority,
\item for all $S \in \delta(s,\tau)$
and $(\rho,s') \in S$,
$\dom{\rho}$ is strictly $\sigmag$-guarded in $\tau$,
and
\item for all bag states $r$,
$S \in \delta(r,\tau)$,
and $\rho \in \EdgeLabels$
such that $\dom{\rho}$ is strictly $\sigmag$-guarded in $\tau$,
there is some $(\rho,q) \in S$.
\end{enumerate}
Making these changes to the automaton results in
only a polynomial blow-up in the size of the automaton.
Let $\cA$ be the resulting automaton with state set $Q$,
transition function $\delta$, and priority function $\Omega$.

We can then write in the usual way (see, e.g.,~\cite{Walukiewicz01})
a vectorial $\Lmu[\sigcode{\sigtarget}{k}]$-formula $\psi$
describing the operation of~$\cA$:
\begin{align*}
\lambda_{n'} X_{s_{n'}} \dots \lambda_1 X_{s_1} .
\begin{pmatrix}
\delta_{s_1} \\
\vdots \\
\delta_{s_{n'}}
\end{pmatrix}
\end{align*}
where $s_1,\dots,s_{n'}$ is an ordering of the states based on the priority
(from least to greatest priority),
$\lambda_i$ is $\mu$ (respectively, $\nu$)
if $s_i$ has an odd (respectively, even) priority,
and $\delta_s$ are transition formulas defined below.
We can assume that the ordering is chosen so that (for some $i$)
$s_1,\dots,s_i$ consists of bag states,
no bag states are present in $s_{i+1},\dots,s_{n'}$,
and $s_{n'}$ is the initial state.
We will refer to $X_{s_i}$ as an \emph{interface predicate} (respectively, \emph{bag predicate})
if $s_i$ is an interface state (respectively, bag state).

The formulas $\delta_s$ describing the transitions from state $s$ are defined as follows:
\begin{align*}
\delta_s &:=
\bigvee_{\tau \in \NodeLabels}
\left(
\ExactLabel{\tau} \wedge
\delta_{s,\tau}
\right)
\\
\delta_{s,\tau} &:=
\bigvee_{S \in \delta(s,\tau)} \left(
\bigwedge_{(\rho,s') \in S} \dmodality{\rho} X_{s'} \wedge
\bigwedge_{\rho \in \EdgeLabels} \bmodality{\rho} \bigvee_{(\rho,s') \in S} X_{s'} \right)
\end{align*}
This captures precisely the meaning of the transition function in a $\mu$-automaton.
The idea is that it picks out exactly the label $\tau$ at the current node,
and then ensures that the successors of this node satisfy the requirements
specified by the transition function when in state $s$ and at a position with label $\tau$.

In order to improve readability,
from now on
we will use $q$ to range over the interface states,
$r$ to range over the bag states,
and $s$ to range over both of these.
Because the priority of bag states is 0
and the automaton alternates between interface and bag states,
we can eliminate all bag predicates by inlining the formulas $\delta_r$
any time a bag predicate $X_r$ appears.
While we are doing this, we can further refine the transition formulas
based on whether it is a bag state or an interface state.

For all bag states $r$,
we construct $\delta^2_{r,\tau}$ from $\delta_{r,\tau}$
by performing the following operations:
\begin{itemize}
\item substitute $\bigwedge_{(\rho,q) \in S}
\dmodality{\rho} (\GuardRng{\rho}{\tau} \wedge D_{\size{\codom{\rho}}} \wedge X_q)$
for $\bigwedge_{(\rho,q) \in S} \dmodality{\rho} X_q$;
\item substitute $\bigwedge_{\rho \in \EdgeLabelsRestrict{S}}
\bmodality{\rho} (\GuardRng{\rho}{\tau} \wedge D_{\size{\codom{\rho}}} \rightarrow \bigvee_{(\rho,q) \in S} X_{q})$
for \\ $\bigwedge_{\rho \in \EdgeLabels} \bmodality{\rho} \bigvee_{(\rho,q) \in S} X_{q}$.
\end{itemize}
Guarding the range of $\rho$-successors
is correct,
since in a $\sigmag$-guarded-interface tree,
every successor of a bag node is an interface node with a strictly $\sigmag$-guarded domain
that is a subset of the bag domain.
It is correct to restrict the conjunction over $\rho \in \EdgeLabels$
to just $\rho \in \EdgeLabelsRestrict{S}$
since we have enforced that $S \in \delta(r,\tau)$ satisfies the property that
$\EdgeLabelsRestrict{S}$ includes every possible outgoing edge label when the node label is $\tau$.

Likewise, for all interface states $q$,
we construct formulas $\delta^2_{q,\tau_0}$ from $\delta_{q,\tau_0}$
using the following operations:
\begin{itemize}
\item substitute $\bigwedge_{(\rho_0,r) \in S} \GuardDom{\rho_0}{\tau_0} \wedge \dmodality{\rho_0} \delta^{2,\tau_0}_r$
for $\bigwedge_{(\rho_0,r) \in S} \dmodality{\rho_0} X_r$;
\item substitute $\bigwedge_{\rho_0 \in \EdgeLabels} \GuardDom{\rho_0}{\tau_0} \wedge
\bmodality{\rho_0} \bigvee_{(\rho_0,r) \in S} \delta^{2,\tau_0}_r$
for \\ $\bigwedge_{\rho_0 \in \EdgeLabels} \bmodality{\rho} \bigvee_{(\rho_0,r) \in S} X_r$;
\end{itemize}
where $\delta^{2,\tau_0}_r$ is obtained from $\delta_{r}$ by
\begin{itemize}
\item substituting $\bigvee_{\tau \in \bagextend{\tau_0}}
\left(
\ExactLabel{\tau} \wedge
\delta^2_{r,\tau}
\right)$
for \\ $\bigvee_{\tau \in \NodeLabels}
\left(
\ExactLabel{\tau} \wedge
\delta_{r,\tau}
\right)$.
\end{itemize}
Guarding the domain of $\rho$-successors is correct,
since in a $\sigmag$-guarded-interface tree,
every interface node has a strictly $\sigmag$-guarded domain,
even though the domain of the successor need not be guarded.
It is correct to replace the disjunction over $\tau \in \NodeLabels$
with $\tau \in \bagextend{\tau_0}$,
since in a $\sigmag$-guarded-interface tree,
any bag node must extend the label $\tau_0$ of its parent.

Finally,
for all interface states $q$,
we construct $\delta^2_{q}$ from $\delta_q$ by
\begin{itemize}
\item substituting
$\bigvee_{\tau_0 \in \IntNodeLabels}
\left(
\GNLabel{\tau_0} \wedge
\delta^2_{q,\tau_0}
\right)$
for \\
$\bigvee_{\tau_0 \in \NodeLabels}
\left(
\ExactLabel{\tau_0} \wedge
\delta_{q,\tau_0}
\right)$.
\end{itemize}
It is correct to replace the disjunction over $\NodeLabels$
with $\IntNodeLabels$ since this is an interface state formula.
Replacing $\ExactLabel{\tau_0}$ with $\GNLabel{\tau_0}$ is correct,
since $\ExactLabel{\tau_0}$ is equivalent to $\GNLabel{\tau_0}$
for all $\tau_0 \in \IntNodeLabels$
(since the domain of $\tau_0$ must be strictly $\sigmag$-guarded).

Note that after these substitutions,
there are no occurrences of bag predicates $X_r$ in the vectorial components,
so these fixpoint variables can be eliminated.
This leaves the interface predicates $X_q$ and vectorial components $\delta^2_q$.

The resulting formula now satisfies the conditions described in the following claim.

\begin{clm}\label{claim:step2}
There is a formula $\psi^2 \in \Lmu[\sigcode{\sigtarget}{k}]$ obtained effectively from $\phi^\mu \in \Lmu[\sigcode{\sigoriginal}{m}]$ of the form
\begin{align*}
\lambda_n X_{q_n} \dots \lambda_1 X_{q_1} .
\begin{pmatrix}
\delta^2_{q_1} \\
\vdots \\
\delta^2_{q_n}
\end{pmatrix}
\end{align*}
such that for all $i$, $\lambda_i \in \set{\mu,\nu}$
and for all $q \in \set{q_1,\dots,q_n}$, $\delta_q^2$ is of the form
\begin{align*}
\delta^{2}_q &:=
\bigvee_{\tau_0 \in \IntNodeLabels}
\left(
\GNLabel{\tau_0} \wedge
\delta^2_{q,\tau_0}
\right)
\\
\delta^2_{q,\tau_0} &:=
\bigvee_{S \in \delta(q,\tau_0)} \Big(
\bigwedge_{(\rho_0,r) \in S} \GuardDom{\rho_0}{\tau_0} \wedge \dmodality{\rho_0} \delta^{2,\tau_0}_r \wedge \\
&\qquad\qquad\qquad \bigwedge_{\rho_0 \in \EdgeLabels} \GuardDom{\rho_0}{\tau_0} \wedge \bmodality{\rho_0} \bigvee_{(\rho_0,r) \in S} \delta^{2,\tau_0}_r \Big) \\
\delta^{2,\tau_0}_r &:=
\bigvee_{\tau \in \bagextend{\tau_0}}
\left(
\ExactLabel{\tau} \wedge
\delta^2_{r,\tau}
\right)
\\
\delta^2_{r,\tau} &:=
\bigvee_{S \in \delta(r,\tau)} \Big(
\bigwedge_{(\rho,q) \in S} \dmodality{\rho} (\GuardRng{\rho}{\tau} \wedge D_{\size{\codom{\rho}}} \wedge X_{q}) \ \wedge \\
&\qquad\qquad\bigwedge_{\rho \in \EdgeLabelsRestrict{S}}
\bmodality{\rho} ( \GuardRng{\rho}{\tau} \wedge D_{\size{\codom{\rho}}} \rightarrow \bigvee_{(\rho,q) \in S} X_{q} ) \Big) .
\end{align*}
Moreover,
for all $\sigmag$-guarded-interface $\sigcode{\sigtarget}{k}$-trees $\tree$,
and for all interface nodes $v$, we have
$\tree,v \models \psi^2$ iff $\tree,v \models \phi^\mu$.
\end{clm}

\paragraph*{Step 2: Refinement to plump trees}

The subformulas
$\dmodality{\rho} \delta^2_r$ and $\bmodality{\rho} \bigvee_{(\rho,r) \in S} \delta^2_r$
both may have negations that are not allowed in $\gnfpk[\sigtarget,\sigmag]$-safe formulas
since the transition function formula $\delta^2_r$
depends on knowing the exact label $\tau$ at the current node:
all of the positive information about which propositions hold
and all negative information about which propositions do not.
In a plump tree, however, we will see that it is not necessary
to know the exact label;
instead, we will make a number of modifications to the formulas,
which will allow us to replace
$\ExactLabel{\tau}$ with $\GNLabel{\tau}$.

We first prove some auxiliary claims that will help with this.
Recall that we say $\vec{V}$ is a \emph{safe valuation}
for predicates $\vec{X}$
if it satisfies the following property:
if $w$ and $w'$ are interface nodes that are roots of bisimilar subtrees,
then
$w$ is in the valuation for $X \in \vec{X}$ iff
$w'$ is in the valuation for $X \in \vec{X}$.

\begin{clm}\label{claim:gn-to-exactlabel}
Let $v$ be an interface node with label $\tau_0$
in a plump $\sigmag$-guarded-interface $\sigcode{\sigtarget}{k}$-tree~$\cT$.

For each $\rho_0$-child $w$ of $v$ and
each $\tau \in \bagextend{\tau_0}$ such that
$\cT, w \models
\GNLabel{\tau}$,
there is some $\rho_0$-child $w'$ of $v$ with label $\tau$
such that
$\cT, w' \models
\ExactLabel{\tau}$.

Moreover,
if $\vec{V}$ is a safe valuation for the interface predicates~$\vec{X}$,
then
$\cT, w, \vec{V} \models \delta^2_{r,\tau}$ iff $\cT, w', \vec{V} \models \delta^2_{r,\tau}$.
\end{clm}

\begin{proof}
Let $w$ be a $\rho_0$-child of $v$ with label $\tau_1$.
It must be the case that $\tau_1 \in \bagextend{\tau_0}$ by the
properties of $\sigmag$-guarded-interface trees.
Let $\tau \in \bagextend{\tau_0}$ such that
$w \models
\GNLabel{\tau}$
(for notational simplicity, we write $w \models \dots$
rather than $\cT, w, \vec{V} \models \dots$).

We first show that $\tau$ is a $(\sigmag,\codom{\rho_0})$-safe restriction of $\tau_1$.

It is clear that $\tau_0 \subseteq \tau \subseteq \tau_1$
since both $\tau$ and $\tau_1$ are in $\bagextend{\tau_0}$,
and if $w \models \GNLabel{\tau}$ then
all $P \in \tau$ must appear in the label $\tau_1$ of $w$.

Consider some proposition $P$ such that $\indices{P}$ is $\sigmag$-guarded in $\tau$.
We must show that $P \in \tau$ iff $P \in \tau_1$.
If $P \in \tau$,
then $P \in \tau_1$ since $\tau \subseteq \tau_1$.
If $P \notin \tau$,
then $\GNLabel{\tau}$ asserts
that $\neg P$ holds (since $\indices{P}$ is $\sigmag$-guarded),
so it must be the case that $P \notin \tau_1$.

Now consider some proposition $P$ using only names in $\codom{\rho_0}$.
Note that $\dom{\rho_0}$ is strictly $\sigmag$-guarded in $\tau_0$,
and since $\tau \supseteq \tau_0$,
$\codom{\tau_0}$ is strictly $\sigmag$-guarded in $\tau$.
Hence, $\indices{P}$ is $\sigmag$-guarded in $\tau$,
and similar reasoning as above implies
that $P \in \tau$ iff $P \in \tau_1$.

This is enough to conclude that $\tau$ is a $(\sigmag,\codom{\rho_0})$-safe restriction of $\tau_1$.
Therefore, by plumpness, there is some $w'$ such that $w' \models \ExactLabel{\tau}$ as desired.

Now make the further assumption that
$w \models \delta^2_{r,\tau}$;
we must show that
$w' \models \delta^2_{r,\tau}$.
If $w \models \delta^2_{r,\tau}$
then there is some $S \in \delta(r,\tau)$
such that
\begin{itemize}
\item (existential requirement) for every $(\rho,q) \in S$, there is some child $u$ of $w$ where $X_q$ holds,
and
\item (universal requirement) for every $\rho$-child $u$ of $w$ such that $\rho \in \EdgeLabelsRestrict{S}$,
there is some $q$ such that $(\rho,q) \in S$ and $X_q$ holds at $u$.
\end{itemize}
We claim that the same property holds at $w'$, using the same $S \in \delta(r,\tau)$.
For each $\rho \in \EdgeLabelsRestrict{S}$,
$\dom{\rho}$ must be strictly $\sigmag$-guarded in $\tau$
(otherwise, the existential requirement would not be fulfilled).
Hence, by the definition of a plump tree,
for each $\rho \in \EdgeLabelsRestrict{S}$ and each $\rho$-child $u$ of $w$,
there is a corresponding $\rho$-child $u'$ in $w'$
such that the subtrees rooted at $u$ and $u'$ are bisimilar.
If $X_q$ holds at $u$, and this is used to satisfy some
existential requirement with $(\rho,q) \in S$,
then $X_q$ also holds at $u'$
so this existential requirement is also satisfied for $w'$
(this relies on the fact that the valuations for $\vec{X}$ are safe).
For the universal requirement at some $\rho$-child $u'$ of $w'$ for $\rho \in \EdgeLabelsRestrict{S}$,
consider the corresponding child $u$ of $w$
such that the subtrees rooted at $u$ and $u'$ are bisimilar,
which is guaranteed by plumpness.
Since the universal requirement is satisfied at $u$,
there is some $X_q$ holding at $u$ with $(\rho,q) \in S$.
Since the valuations for $\vec{X}$ are safe,
this means that $X_q$ also holds at $u'$,
so the universal requirement at $u'$ holds.
Using this reasoning, we can conclude that $w' \models \delta^2_{r,\tau}$.

The reasoning is similar in the other direction,
assuming $w' \models \delta^2_{r,\tau}$.

Thus, we can conclude that
$w' \models \ExactLabel{\tau}$,
and
$w' \models \delta^2_{r,\tau}$ iff $w \models \delta^2_{r,\tau}$.
\end{proof}

Using the previous claim, we will prove that
we can replace
$\ExactLabel{\tau}$ with $\GNLabel{\tau}$
when considering equivalence only over plump trees.
The exact way we do this replacement, however, will depend on whether
$\ExactLabel{\tau}$ is under a diamond modality or a box modality.
We define the following auxiliary formulas to handle these cases
\begin{align*}
\delta^{\lozenge,\tau_0}_{r} &:=
\bigvee_{\tau \in \bagextend{\tau_0}} \left( \GNLabel{\tau} \wedge
\delta^2_{r,\tau} \right)\\
\delta^{\square,\tau_0}_{S} &:=
\bigwedge_{\tau \in \bagextend{\tau_0} } \Big( \GNLabel{\tau} \rightarrow
\bigvee_{(\rho,r) \in S} \delta^2_{r,\tau} \Big)
\end{align*}
and prove the correctness of the following transformation:

\begin{clm}\label{claim:plump}
Let $v$ be an interface node with label $\tau_0$
in a plump $\sigmag$-guarded-interface $\sigcode{\sigtarget}{k}$-tree $\cT$,
and let $\vec{V}$ be a safe valuation for the interface predicates~$\vec{X}$.
Then
\begin{align}
\tag{1} \cT,v,\vec{V} \models \dmodality{\rho} \delta^2_r
&\text{\quad iff \quad}
\cT,v,\vec{V} \models \dmodality{\rho} \delta^{\lozenge,\tau_0}_r
 \,,\label{eq:plump1}
\\
\tag{2} \cT,v,\vec{V} \models \bmodality{\rho} \bigvee_{(\rho,r) \in S} \delta^2_r
&\text{\quad iff \quad}
\cT,v,\vec{V} \models \bmodality{\rho} \delta^{\square,\tau_0}_S
 \,.\label{eq:plump2}
\end{align}
\end{clm}

\begin{proof}[Proof of claim]
Fix a plump $\sigmag$-guarded-interface $\sigcode{\sigtarget}{k}$-tree $\cT$,
an interface node $v$ with label~$\tau_0$,
and a safe valuation $\vec{V}$ for the interface predicates~$\vec{X}$.
We write $v \models \psi$ for $\cT,v,\vec{V} \models \psi$.
\begin{enumerate}
\item[\eqref{eq:plump1}]
We start with the easier left-to-right direction.
Assume $v \models \dmodality{\rho} \delta^2_r$.
Then there is some $\rho$-child $w$ of $v$
such that $w \models \delta^2_r$.
Let $\tau$
be the exact label at $w$;
note that $\tau \in \bagextend{\tau_0}$ by properties of $\sigmag$-guarded-interface trees.
Then
$w \models \ExactLabel{\tau} \wedge \delta^2_{r,\tau}$.
This implies that
$w \models
\GNLabel{\tau} \wedge
\delta^2_{r,\tau}$.
This is enough to conclude that $w \models \delta^{\lozenge,\tau_0}_r$
and $v \models \dmodality{\rho} \delta^{\lozenge,\tau_0}_r$.

Next, we prove the right-to-left direction,
which makes use of Claim~\ref{claim:gn-to-exactlabel} (and hence makes use of plumpness).
Assume that $v \models \dmodality{\rho} \delta^{\lozenge,\tau_0}_r$.
Then there is some $\rho$-child $w$ and some $\tau \in \bagextend{\tau_0}$ such that
$w \models
\GNLabel{\tau} \wedge
\delta^2_{r,\tau}$.
By Claim~\ref{claim:gn-to-exactlabel}, there is some $\rho$-child $w'$ such that
$w' \models
\ExactLabel{\tau} \wedge
\delta^2_{r,\tau}$,
so $v \models \dmodality{\rho} \delta^2_r$ as desired.

\item[\eqref{eq:plump2}]
The easier direction is the right-to-left direction:
assume $v \models \bmodality{\rho} \delta^{\square,\tau_0}_S$.
Let $w$ be a $\rho$-child of $v$,
and let $\tau$ be the label at $w$.
It must be the case that $\tau \in \bagextend{\tau_0}$
by the properties of a $\sigmag$-guarded-interface tree.
Hence, by the definition of $\delta^{\square,\tau_0}_S$,
there is some $(\rho,r) \in S$,
such that $w \models \delta^2_{r,\tau}$.
This means $w \models
\GNLabel{\tau} \wedge
\delta^2_{r,\tau}$,
so $w \models \bigvee_{(\rho,r) \in S} \delta^2_r$.
Overall, this means $v \models \bmodality{\rho} \bigvee_{(\rho,r) \in S} \delta^2_r$
as desired.

Now assume $v \models \bmodality{\rho} \bigvee_{(\rho,r) \in S} \delta^2_r$
for the left-to-right direction.
Let $w$ be a $\rho$-child of $v$.
Consider $\tau \in \bagextend{\tau_0}$ such that
$w \models \GNLabel{\tau}$.
It suffices to show that $w \models \bigvee_{(\rho,r) \in S} \delta^2_{r,\tau}$.
By Claim~\ref{claim:gn-to-exactlabel},
there is a $\rho$-child $w'$ of $v$ such that
$w' \models
\ExactLabel{\tau}$.
Since $v \models \bmodality{\rho} \bigvee_{(\rho,r) \in S} \delta^2_r$
and the label at $w'$ is $\tau$,
this means that
$w' \models \delta^2_{r,\tau}$ for some $(\rho,r) \in S$.
By Claim~\ref{claim:gn-to-exactlabel},
$w \models \delta^2_{r,\tau}$ as well,
so $w \models \bigvee_{(\rho,r) \in S} \delta^2_{r,\tau}$
as required.
\qedhere
\end{enumerate}
\end{proof}

\noindent
This allows us to take $\psi^2$ from the previous step
and refine it further based on the assumption that
we are only interested in plump $\sigmag$-guarded-interface $\sigcode{\sigtarget}{k}$-trees.
The shape of the resulting formula is stated in the following claim:

\begin{clm}
There is a formula $\psi^3$ obtained effectively from $\psi^2$ such that
the vectorial component $\delta^3_q$ for each $q$ is of the form
\begin{align*}
\delta^{3}_q &:=
\bigvee_{\tau_0 \in \IntNodeLabels}
\left(
\GNLabel{\tau_0} \wedge
\delta^3_{q,\tau_0}
\right)
\\
\delta^3_{q,\tau_0} &:=
\bigvee_{S \in \delta(q,\tau_0)} \Big(
\bigwedge_{(\rho_0,r) \in S} \GuardDom{\rho_0}{\tau_0} \wedge \dmodality{\rho_0} \delta^{\lozenge,\tau_0}_r \wedge \\
&\qquad\qquad\qquad\bigwedge_{\rho_0 \in \EdgeLabels} \GuardDom{\rho_0}{\tau_0} \wedge \bmodality{\rho_0} \delta^{\square,\tau_0}_S \Big) \\
\delta^{\lozenge,\tau_0}_{r} &:=
\bigvee_{\tau \in \bagextend{\tau_0}} \left( \GNLabel{\tau} \wedge
\delta^3_{r,\tau} \right)\\
\delta^{\square,\tau_0}_{S} &:=
\bigwedge_{\tau \in \bagextend{\tau_0} } \Big( \GNLabel{\tau} \rightarrow
\bigvee_{(\rho,r) \in S} \delta^3_{r,\tau} \Big)
\\
\delta^3_{r,\tau} &:=
\bigvee_{S \in \delta(r,\tau)} \Big(
\bigwedge_{(\rho,q) \in S} \dmodality{\rho} (\GuardRng{\rho}{\tau} \wedge D_{\size{\codom{\rho}}} \wedge X_{q}) \ \wedge \\
&\qquad\qquad\bigwedge_{\rho \in \EdgeLabelsRestrict{S}}
\bmodality{\rho} ( \GuardRng{\rho}{\tau} \wedge D_{\size{\codom{\rho}}}  \rightarrow \bigvee_{(\rho,q) \in S} X_{q} ) \Big) .
\end{align*}
Moreover, for all plump $\sigcode{\sigtarget}{k}$-trees $\cT$
and for all interface nodes $v$,
we have $\cT,v \models \psi^3$ iff $\cT,v \models \psi^2$.
\end{clm}

\begin{proof}
For all interface states $q$ and bag states $r$:
\begin{itemize}
\item substitute $\dmodality{\rho} \delta^{\lozenge,\tau_0}_r$
for $\dmodality{\rho} \delta^2_r$ in $\delta^2_{q,\tau_0}$;
\item substitute $\bmodality{\rho} \delta^{\square,\tau_0}_{S}$
for $\bmodality{\rho} \bigvee_{(\rho,r) \in S} \delta^2_r$ in $\delta^2_{q,\tau_0}$.
\end{itemize}
Let $\psi^3$ be the resulting formula,
with vectorial components~$\delta^3_q$.

Equivalence over plump $\sigmag$-guarded-interface trees essentially follows from Claim~\ref{claim:plump}.
Technically, one would show that the result is correct by induction on the number of fixpoint operators,
and a transfinite induction on the fixpoint approximant required for each fixpoint.
Since all of the fixpoint approximations give safe valuations for the fixpoint predicates,
Claim~\ref{claim:plump} can be applied at each step in the induction.
\end{proof}

\paragraph*{Step 3: Clean up to obtain $\gnfpk[\sigtarget,\sigmag]$-safe formula}

The formula $\psi^3$ obtained in the previous step is almost $\gnfpk[\sigtarget,\sigmag]$-safe.
We now perform some clean-up operations
in order to get the required $\gnfpk[\sigtarget,\sigmag]$-safe formula.

The first clean-up operation deals with the negations
that are implicit in the box modalities.
The following claim shows the shape of the
formula after we have eliminated box modalities,
and pushed some of these negations inside.

\begin{clm}\label{claim:step4}
There is a formula $\psi^4$ obtained effectively from $\psi^3$ such that
the vectorial component $\delta^4_q$ for each $q$ is of the form:
\begin{align*}
\delta^{4}_q &:=
\bigvee_{\tau_0 \in \IntNodeLabels}
\left(
\GNLabel{\tau_0} \wedge
\delta^4_{q,\tau_0}
\right)
\\
\delta^4_{q,\tau_0} &:=
\bigvee_{S \in \delta(q,\tau_0)} \Big(
\bigwedge_{(\rho_0,r) \in S} \GuardDom{\rho_0}{\tau_0} \wedge \dmodality{\rho_0} \delta^{\lozenge,\tau_0}_r \wedge \\
&\qquad\qquad\qquad\bigwedge_{\rho_0 \in \EdgeLabels} \GuardDom{\rho_0}{\tau_0} \wedge \neg \dmodality{\rho_0} \delta^{\neg \lozenge,\tau_0}_S \Big) \\
\delta^{\lozenge,\tau_0}_{r} &:=
\bigvee_{\tau \in \bagextend{\tau_0}} \left( \GNLabel{\tau} \wedge
\delta^4_{r,\tau} \right)\\
\delta^4_{r,\tau} &:=
\bigvee_{S \in \delta(r,\tau)} \Big(
\bigwedge_{(\rho,q) \in S} \GuardDom{\rho}{\tau} \wedge \dmodality{\rho} (\GuardRng{\rho}{\tau} \wedge D_{\size{\codom{\rho}}} \wedge X_{q}) \ \wedge \\
&\qquad\qquad\bigwedge_{\rho \in \EdgeLabelsRestrict{S}} \GuardDom{\rho}{\tau} \wedge
\neg \dmodality{\rho} (\bigwedge_{(\rho,q) \in S} \GuardRng{\rho}{\tau} \wedge D_{\size{\codom{\rho}}} \wedge \neg X_{q} ) \Big) \\
\delta^{\neg \lozenge,\tau_0}_{S} &:=
\bigvee_{\tau \in \bagextend{\tau_0} } \Big( \GNLabel{\tau} \wedge
\bigwedge_{(\rho,r) \in S} \overline{\delta}^4_{r,\tau} \Big)
\\
\overline{\delta}^4_{r,\tau} &:=
\bigwedge_{S \in \delta(r,\tau)} \Big(
\bigvee_{(\rho,q) \in S} \GuardDom{\rho}{\tau} \wedge \neg \dmodality{\rho} (\GuardRng{\rho}{\tau} \wedge D_{\size{\codom{\rho}}} \wedge X_{q}) \ \vee \\
&\qquad\qquad\bigvee_{\rho \in \EdgeLabelsRestrict{S}} \GuardDom{\rho}{\tau} \wedge
\dmodality{\rho} (\bigwedge_{(\rho,q) \in S} \GuardRng{\rho}{\tau} \wedge D_{\size{\codom{\rho}}} \wedge \neg X_{q} ) \Big) .
\end{align*}
Moreover,
for all $\sigmag$-guarded-interface $\sigcode{\sigtarget}{k}$-trees $\cT$ and for all interface nodes $v$,
we have $\cT,v \models \psi^4$ iff $\cT,v \models \psi^3$.
\end{clm}

\begin{proof}
Recall that
$\bmodality{\rho} \chi$ is equivalent to $\neg \dmodality{\rho} \neg \chi$.
Simply by using this equivalence, and pushing negations inside,
we can rewrite the transition formulas to:
\begin{align*}
\delta^{4}_q &:=
\bigvee_{\tau_0 \in \IntNodeLabels}
\left(
\GNLabel{\tau_0} \wedge
\delta^4_{q,\tau_0}
\right)
\\
\delta^4_{q,\tau_0} &:=
\bigvee_{S \in \delta(q,\tau_0)} \Big(
\bigwedge_{(\rho_0,r) \in S} \GuardDom{\rho_0}{\tau_0} \wedge \dmodality{\rho_0} \delta^{\lozenge,\tau_0}_r \wedge \\
&\qquad\qquad\qquad\bigwedge_{\rho_0 \in \EdgeLabels} \GuardDom{\rho_0}{\tau_0} \wedge \neg \dmodality{\rho_0} \delta^{\neg \lozenge,\tau_0}_S \Big) \\
\delta^{\lozenge,\tau_0}_{r} &:=
\bigvee_{\tau \in \bagextend{\tau_0}} \left( \GNLabel{\tau} \wedge
\delta^4_{r,\tau} \right)\\
\delta^4_{r,\tau} &:=
\bigvee_{S \in \delta(r,\tau)} \Big(
\bigwedge_{(\rho,q) \in S} \dmodality{\rho} (\GuardRng{\rho}{\tau} \wedge D_{\size{\codom{\rho}}} \wedge X_{q}) \ \wedge \\
&\qquad\qquad\bigwedge_{\rho \in \EdgeLabelsRestrict{S}}
\neg \dmodality{\rho} ( \GuardRng{\rho}{\tau} \wedge D_{\size{\codom{\rho}}} \wedge \bigwedge_{(\rho,q) \in S} \neg X_{q} ) \Big) \\
\delta^{\neg \lozenge,\tau_0}_{S} &:=
\bigvee_{\tau \in \bagextend{\tau_0} } \Big( \GNLabel{\tau} \wedge
\bigwedge_{(\rho,r) \in S} \overline{\delta}^4_{r,\tau} \Big)
\\
\overline{\delta}^4_{r,\tau} &:=
\bigwedge_{S \in \delta(r,\tau)} \Big(
\bigvee_{(\rho,q) \in S}  \neg \dmodality{\rho} (\GuardRng{\rho}{\tau} \wedge D_{\size{\codom{\rho}}} \wedge X_{q}) \ \vee \\
&\qquad\qquad\bigvee_{\rho \in \EdgeLabelsRestrict{S}}
\dmodality{\rho} ( \GuardRng{\rho}{\tau} \wedge D_{\size{\codom{\rho}}} \wedge \bigwedge_{(\rho,q) \in S} \neg X_{q} ) \Big) .
\end{align*}

Finally, we perform the following substitutions:
\begin{itemize}
\item substitute $\GuardDom{\rho}{\tau} \wedge \dmodality{\rho} (\GuardRng{\rho}{\tau} \wedge D_{\size{\codom{\rho}}} \wedge X_{q})$ for \\
$\dmodality{\rho} (\GuardRng{\rho}{\tau} \wedge D_{\size{\codom{\rho}}} \wedge X_{q})$ in $\delta^4_{r,\tau}$;
\item substitute $\GuardDom{\rho}{\tau} \wedge \neg \dmodality{\rho} (\GuardRng{\rho}{\tau} \wedge D_{\size{\codom{\rho}}} \wedge X_{q})$ for \\
$\neg \dmodality{\rho} (\GuardRng{\rho}{\tau} \wedge D_{\size{\codom{\rho}}} \wedge X_{q})$ in $\overline{\delta}^4_{r,\tau}$;
\item substitute $\GuardDom{\rho}{\tau} \wedge \neg \dmodality{\rho} ( \bigwedge_{(\rho,q) \in S} \GuardRng{\rho}{\tau} \wedge D_{\size{\codom{\rho}}} \wedge \neg X_{q} )$ for \\ $\neg \dmodality{\rho} ( \GuardRng{\rho}{\tau} \wedge D_{\size{\codom{\rho}}} \wedge \bigwedge_{(\rho,q) \in S} \neg X_{q} )$ in $\delta^4_{r,\tau}$;
\item substitute $\GuardDom{\rho}{\tau} \wedge \dmodality{\rho} ( \bigwedge_{(\rho,q) \in S} \GuardRng{\rho}{\tau} \wedge D_{\size{\codom{\rho}}} \wedge \neg X_{q} )$ for \\ $\dmodality{\rho} ( \GuardRng{\rho}{\tau} \wedge D_{\size{\codom{\rho}}} \wedge \bigwedge_{(\rho,q) \in S} \neg X_{q} )$ in $\overline{\delta}^4_{r,\tau}$;
\end{itemize}
This is correct since the domain of any edge label exiting a bag node must be strictly $\sigmag$-guarded
in a $\sigmag$-guarded-interface tree.
This results in a formula of the desired shape.
\end{proof}

The formula resulting from Claim~\ref{claim:step4} is $\gnfpk[\sigtarget,\sigmag]$-safe
except for the fact that it uses a simultaneous fixpoint.
As a final clean-up step,
we convert the vectorial fixpoint formula $\psi^4$ to a standard $\Lmu$-formula using the \Bekic principle,
with the outermost fixpoint testing membership in the interface state $q_n$ component (recall that the outermost fixpoint operator based on $X_{q_n}$
corresponded to the initial state $q_n$ of the automaton that was equivalent to the original $\Lmu$-formula).
This yields an equivalent formula $\lambda_n X_{q_n} . \chi$, where $\chi$ uses no vectorial fixpoints.
The $\gnfpk[\sigtarget,\sigmag]$-safe formula required for Lemma~\ref{lemma:gnfp-safe}
is just $\top \wedge D_0 \wedge \lambda_n X_{q_n} . \chi$.
This is correct since we are interested in evaluating this starting at the root of $\sunravelk{\fB}{\sigtarget,\sigmag}$,
which has an empty set of names.

This concludes the proof of Lemma~\ref{lemma:gnfp-safe} and Lemma~\ref{lemma:backwards-gnfp}.

\subsection{Decidability of definability}\label{sec:gnfpdecidability}
Using the above lemma and Proposition~\ref{prop:effectivedefhigh}, we obtain the following analog of
Theorem~\ref{thm:effectivedefgfp}.

\begin{thm}%
\label{thm:effectivegnfp}
The $\gnfpk[\sigtarget,\sigmag]$ definability problem is decidable
for $\gnlinvar[\sigoriginal]$-invariant $\gso[\sigoriginal]$ and $k,l \geq \width{\sigoriginal}$.
\end{thm}

Since $\unfpk[\sigtarget]$ is just $\gnfpk[\sigtarget,\emptyset]$,
we obtain the following corollary:

\begin{cor}
The $\unfpk[\sigtarget]$ definability problem is decidable
for $\gnlinvar[\sigoriginal]$-invariant $\gso[\sigoriginal]$
and $k,l \geq \width{\sigoriginal}$.
\end{cor}

We get corollaries
for  fragments of $\fo$, analogous to Corollary~\ref{cor:gfdef}:

\begin{cor}%
\label{cor:gnfkunfdef}
The $\gnfk[\sigtarget,\sigmag]$ and  $\unfk[\sigtarget]$ definability problems are  decidable
for $\gnlinvar[\sigoriginal]$-invariant $\fo[\sigoriginal]$ and $k,l \geq \width{\sigoriginal}$.
\end{cor}

We can also apply the backward and forward mappings to get  a semantic characterization for $\gnfpk$,
analogous to the Janin-Walukiewicz theorem. The following  extends a result of~\cite{gnfj}
characterizing $\gnfk$ formulas as the $\sgnkinvar$-invariant fragment
of $\fo$.

\begin{thm}\label{thm:gnfpk-characterization}
$\gnfpk[\sigtarget,\sigmag]$ is the $\sgnkinvar[\sigtarget,\sigmag]$-invariant fragment of $\gso[\sigtarget]$.
\end{thm}

The proof is similar to the characterizations of Janin-Walukiewicz and~\cite{GradelHO02},
and can also be seen as a variant of Proposition~\ref{prop:effectivedefhigh}, where we
use $\sgnkinvar[\sigtarget,\sigmag]$-invariance rather than equivalence to a $\gnfpk[\sigtarget,\sigmag]$ sentence
in justifying that the input formula is equivalent to the result of the composition of backward and forward mappings.


%% file: plump.tex
In order to define the property that this special unravelling has,
we need to define how we can modify copies of certain parts of the structure
in a way that is not distinguishable by $\gnfpk$.
Let $\tau$ and $\tau'$ be sets of $\sigtarget$-atoms over some set $A$ of elements.
Let $I, J \subseteq A$.
We say $\tau$ and $\tau'$ agree on $J$
if for all $\sigtarget$-atoms $\alpha(a_1,\dots,a_l)$ with $\set{a_1,\dots,a_l} \subseteq J$,
$\alpha(a_1,\dots,a_l) \in \tau$ iff $\alpha(a_1,\dots,a_l) \in \tau'$.
We say $\tau'$ is an \emph{$(\sigmag,I)$-safe restriction} of $\tau$ if
\begin{enumerate}
\item $\tau' \subseteq \tau$;
\item $\tau'$ agrees with $\tau$ on $I$;
\item $\tau'$ agrees with $\tau$ on every $J \subseteq A$ that is $\sigmag$-guarded in~$\tau'$.
\end{enumerate}
We will use the same terminology and notation when
$\tau$ and $\tau'$ are encoded sets of atoms (rather than sets of atoms)
and $I, J$ are sets of indices encoding elements (rather than elements themselves).

Note that $\tau$ itself is considered a trivial $(\sigmag,I)$-safe restriction of $\tau$.
Here is another example:

\begin{exa}\label{ex:safe-restriction}
Consider signatures $\sigtarget = \set{ U, R, T}$  and $\sigmag = \set{R}$,
where $U$ is a unary relation,
$R$ is a binary relation, and $T$ is a ternary relation.
For readability in this example, we will write, e.g., $R(u,v)$ instead of $R uv$.
Consider $I = \set{1,2}$ and
\[
\tau =
\left\{ \begin{array}{c}
U(1), U(3) \\
R(1,2), R(2,3), R(3,1) \\
T(3,2,2)
\end{array} \right \}.
\]
Then the possible $(\sigmag,I)$-safe restrictions of $\tau$
are $\tau$ itself and
\begin{align*}
\tau'_1  =
\left\{ \begin{array}{c}
U(1), U(3) \\
R(1,2), R(2,3) \\
T(3,2,2)
\end{array} \right \} &
\qquad
\tau'_2  =
\left\{ \begin{array}{c}
U(1), U(3) \\
R(1,2), R(3,1) \\
T(3,2,2)
\end{array} \right \} &&
\quad
\tau'_4  =
\left\{ \begin{array}{c}
U(1), U(3) \\
R(1,2) \\
T(3,2,2)
\end{array} \right \} &
\\
&
\qquad
\tau'_3  =
\left\{ \begin{array}{c}
U(1), U(3) \\
R(1,2), R(3,1)
\end{array} \right \} &&
\quad
\tau'_5  =
\left\{ \begin{array}{c}
U(1), U(3) \\
R(1,2)
\end{array} \right \}\enspace. &
\end{align*}

Note that in some of the restrictions, we drop some
atoms that use relations from $\sigtarget$ or even $\sigmag$.
However, we cannot drop atoms over unary relations
(since these are always trivially guarded),
and we can never drop atoms using elements from the set~$I$.
Furthermore, the $\sigmag$-atoms that we keep
restrict what other atoms that we can drop, since for any $\sigmag$-guarded
set that remains we must preserve atoms over that set.
\end{exa}

We are almost ready to define the plumpness property that a plump unravelling will exhibit. The idea is that a plump tree is a special type of $\sigmag$-guarded interface tree. Recall that a $\sigmag$-guarded interface tree alternates between interface nodes that are $\sigmag$-guarded and bag nodes of some bounded width (see the description on page~\pageref{interfacenodes}).
The additional property that a plump tree must satisfy is that
that for every interface node $u$, all safe restrictions of the atoms represented at $u$
are realized by siblings of $u$.

By a \emph{$(\sigmag,I)$-safe restriction of a node $v$} in a tree code,
we mean a $(\sigmag,I)$-safe restriction of the atoms represented
by the node $v$.

A $\sigcode{\sigtarget}{k}$-tree
has the \emph{$\sigmag$-plumpness property}
if for all interface nodes $v$:
if $w$ is a $\rho_0$-child of $v$ over names $J$ with $I = \codom{\rho_0}$
and
$\tau$ is the encoded set of $\sigtarget$-atoms that hold at $w$,
then for any $(\sigmag,I)$-safe restriction $\tau'$ of $\tau$,
there is a $\rho_0$-child $w'$ of $v$ such that
\begin{enumerate}
\item $\tau'$
is the encoded set of $\sigtarget$-atoms that hold at $w'$;
\item for each $\rho$-child $u'$ of $w'$,
there is a $\rho$-child $u$ of $w$ such that the subtrees rooted at $u$ and $u'$ are bisimilar; and
\item for each $\rho$-child $u$ of $w$
such that $\dom{\rho}$ is strictly $\sigmag$-guarded in $\tau'$,
there is a $\rho$-child $u'$ of $w'$ such that
the subtrees rooted at $u'$ and $u$ are bisimilar.
\end{enumerate}
We say a tree code is \emph{$\sigmag$-plump} if it satisfies this property.

\begin{exa}
Let $\sigtarget,\sigmag$ be as in Example~\ref{ex:safe-restriction}.
Let $\tree$ be a $\sigmag$-plump tree.
Suppose there is an interface node $v$ in $\tree$ with label encoding $\tau_0 = \set{ U(1), R(1,2)}$,
and there is a $\rho_0$-child $w$ of $v$ such that the label of $w$ encodes $\tau = \set{ U(1), U(3),
R(1,2), R(2,3), R(3,1), T(3,2,2)}$,
and $\rho_0$ is the identity function with domain $\set{1,2}$.
Then by plumpness there must also be $\rho_0$-children $w_1, \dots, w_5$ of $v$
with labels encoding
$\tau'_1, \dots, \tau'_5$ from Example~\ref{ex:safe-restriction}.
\end{exa}

The following proposition shows that one can obtain unravellings that are plump:

\begin{prop}\label{prop:bisim-plump}
Let $\fB$ be a $\sigoriginal$-structure, $k \in \N$, and $\sigmag \subseteq \sigtarget \subseteq \sigoriginal$.
There is a consistent, plump, $\sigmag$-guarded-interface tree $\sunravelk{\fB}{\sigtarget,\sigmag}$
such that $\fB$ is $\sgnkinvar[\sigtarget,\sigmag]$-bisimilar to $\decode{\sunravelk{\fB}{\sigtarget,\sigmag}}$.
We call $\decode{\sunravelk{\fB}{\sigtarget,\sigmag}}$ the \emph{plump unravelling} of $\fB$.
\end{prop}

The proof of the proposition will take up the remainder of this section.
Before proceeding with the proof, we would recommend reviewing the notation and definitions on page~\pageref{unravellings}.

\input{bisimplump}


%% file: bisimplump.tex
\newcommand{\underlying}[1]{}

\input{plumpdef}
\noindent This completes the definition of $\sunravelk{\fA}{\sigtarget,\sigmag}$.

One worry with a definition like this is that the resulting tree is not consistent;
in particular, one could worry that information about shared elements is not propagated correctly between neighboring nodes
because the labels come from encoding safe restrictions of the atomic type (rather than just always encoding the exact atomic type).
Suppose that we have some encoded atom $R_{i_1,\dots,i_n}$ in some node.
We can show that this information is propagated to all appropriate nodes,
by induction on the length of the propagation path.
The base case (for a path of length 0) is trivial.
Now suppose that $R_{i_1,\dots,i_n}$ is in an interface node $v$ and there is some $\rho$-child $w$
such that $\indices{R} \subseteq \dom{\rho} = \codom{\rho}$.
The label at $w$ could correspond to a $(\sigmag,\codom{\rho})$-safe restriction of the atomic type.
But because such a restriction must agree on $\codom{\rho}$,
this means that $R_{\rho(i_1), \dots, \rho(i_n)}$ must appear in $w$ as required.
If $R_{i_1,\dots,i_n}$ is in a bag node $w$ and there is some $\rho$-child $v$
such that $\indices{R} \subseteq \dom{\rho} = \codom{\rho}$,
then $\dom{\rho}$ must be strictly $\sigmag$-guarded in $w$ (by the definition of the plump unravelling).
Because $w$ must agree exactly with the atomic type on any $\sigmag$-guarded set in $w$,
this means that $R_{\rho(i_1), \dots, \rho(i_n)}$ will also appear in
$v$ as required.
Similar reasoning can be used for propagation to a parent node as well.

Thus, $\sunravelk{\fA}{\sigtarget,\sigmag}$ is consistent.
It is also straightforward to check that it is a plump, $\sigmag$-guarded-interface tree.

\myparagraph{Proof of bisimilarity}

The proof that this unravelling is $\sgnkinvar[\sigtarget,\sigmag]$-bisimilar to the original structure
is fairly standard.
It suffices to show that
Duplicator has a winning strategy in the $\sgnkinvar[\sigtarget,\sigmag]$-game
between
$\fU := \decode{\sunravelk{\fA}{\sigtarget,\sigmag}}$
and
$\fA$.

We build up the strategy inductively,
ensuring that every partial play of even length in the strategy ends
in a position $f$ that is \emph{good}.
We say a position $f$ with active structure $\fA$ is good
if it is a partial $\sigtarget$-isomorphism and
there is an interface node $v$ in $\cU := \sunravelk{\fA}{\sigtarget,\sigmag}$
where $\codom{f}$ is represented,
and this node is based on $\dom{f}$ from $\fA$.
Likewise, we say a position $f$ with active structure $\fU$ is good
if it is a partial $\sigtarget$-isomorphism and
there is an interface node $v$ in $\cU$
where $\dom{f}$ is represented,
and this node is based on $\codom{f}$ from~$\fA$.

The empty play is trivially good, since
the starting position has empty domain,
and is represented at the root of $\cU$.

Assume that a partial play consistent with the strategy ends in a good position $f : A \to U$ with active structure $\fA$,
with $\codom{f}$ represented in node $v$ in $\cU$.
We show how to extend the strategy in the block $k$-width game.
\begin{itemize}
\item If Spoiler switches structures and then collapses to some strictly $\sigmag$-guarded set $U'$ of $U$,
then the resulting position $f'$ is clearly still a partial $\sigtarget$-isomorphism.
In the plump unravelling, we know that $v$ is a node based on a $\sigmag$-guarded set $A$ (i.e.~a sequence ending in $A$) and
there is a successor $w$ of $v$ such that $w$ is also based on $A$ (i.e.~a sequence ending in some $(A,\tau)$).
Moreover, both $v$ and $w$ must represent exactly the atomic $\sigtarget$-type of $\dom{f}$,
since interface nodes in $\cU$ represent the exact $\sigmag$-type of the underlying elements of $\fA$,
and this is propagated to a successor that shares the same elements.
Hence, $w$ has a successor $v'$ based on the restriction of $A$ to $\codom{f'}$, which represents the strictly $\sigmag$-guarded subset $U'$ of $U$.
Thus, we have extended the play to another good position~$f'$.

\item Now consider the case when Spoiler extends to $A' \supseteq A$ with $\size{A'} \leq k$.
Consider the successor $w$ of $v$ that is based on $A'$ and the exact $\sigmag$-type of $A'$;
this exists in $\cU$ since $v$ was based on $A$, $A' \subseteq A$ and $\size{A'} \leq k$.
Let $U'$ be the elements from $\fU$ that come from this node $w$.
Then $f' : A' \to U'$ is a partial $\sigtarget$-isomorphism based on how we selected $w$.
Spoiler must then select some strictly $\sigmag$-guarded $A''$ that is a subset of $A'$ to collapse to. The resulting position $f'' : A'' \to U''$ is also a partial $\sigtarget$-isomorphism. Moreover, by the definition of the unravelling, $\codom{f''}$ must be represented at some interface node $v'$ that is a successor of $w$ and is based on $A''$.
This means we have extended the strategy so that the partial play ends in a good position.
\end{itemize}

\noindent
Now assume that the play ends in a good position $f : U \to A$ with active structure $\fU$,
with $\dom{f}$ represented in node $v$ in $\cU$.
We show how to extend the strategy.
\begin{itemize}
\item If Spoiler switches structures and then collapses to some strictly $\sigmag$-guarded set,
then we can use similar reasoning as above to argue that the resulting position is good.
\item So assume Spoiler extends to $U' \supseteq U$ with $\size{U'} \leq k$.
Let Duplicator select the elements $A'$ which were the basis for $U'$.
The resulting $f' : U' \to A'$ is guaranteed to be a partial $\sigtarget$-homomorphism since the atomic information
about $U'$ must be a subset of the atomic type of the underlying elements $A'$ from $\fA$.

We cannot guarantee at this stage that $f'$ is a partial $\sigtarget$-isomorphism. First, multiple elements from $U'$ might be derived from a single element in $A'$. This is the case in any unravelling construction, since the unravelling creates `copies' of pieces of the structure.
Second, elements from $U'$ might come from a node that is a restriction of the atomic type of
the underlying elements $A'$ from $\fA$. This is a particular feature of the plump unravellings, since we include variations of the atomic type of the elements $A'$ based on safe restrictions of the atomic type of those elements.
We also cannot guarantee that these elements $U$ are all represented at a single node in $\cU$.

However, when Spoiler collapses to a strictly $\sigmag$-guarded set $U'' \subseteq U'$,
this strictly $\sigmag$-guarded set must be represented in at least one node $w$ in $\cU$;
suppose this is a bag node.
This node $w$ could represent a restriction of the type.
But the key property of the plump unravelling is that we only allow safe restrictions of the type,
which must agree exactly with the underlying elements from $\fA$ on any $\sigmag$-guarded set.
Hence, this restriction to a strictly $\sigmag$-guarded set must be a partial $\sigtarget$-isomorphism,
and the domain of this partial $\sigtarget$-isomorphism
is represented in a successor $v'$ of $w$.
The reasoning is easier if $w$ is an interface node, since an interface node must be based
on the exact type of the underlying elements from $\fA$.
In any case, the resulting partial plays end in good positions, as desired.
\end{itemize}

\noindent
This means that we can build up a winning strategy for Duplicator in the $\sgnkinvar[\sigtarget,\sigmag]$-game
between $\fA$ and its plump unravelling, which concludes the proof of Proposition~\ref{prop:bisim-plump}.


%% file: plumpdef.tex
\myparagraph{Construction of plump unravelling}

Let $\fA$ be a $\sigoriginal$-structure and let $k \in \N$.
Consider the set $\Pi'_k$ of finite sequences of the form
$X_0 (Y_1,\tau_1) X_1 \dots (Y_m,\tau_m)$ or
$X_0 (Y_1,\tau_1) X_1 \dots (Y_m,\tau_m) X_m$,
where $X_0 = \emptyset$ and for all $1 \leq i \leq m$,
\begin{itemize}
\item $X_i$ is a set of elements of $\fA$ that is strictly $\sigmag$-guarded by an atom in $\tau_i$;
\item $Y_i$ is a set of elements of $\fA$ of size at most $k$;
\item $Y_i$ contains both $X_{i-1}$ and $X_{i}$;
\item $\tau_{i}$ is a $(\sigmag,X_{i-1})$-safe restriction of $\atypesig{Y_i}{\fA}{\sigtarget}$.
\end{itemize}

\noindent
Let $\sunravelk{\fA}{\sigtarget,\sigmag}$ denote the $\sigcode{\sigtarget}{k}$-tree of sequences of $\Pi'_k$,
arranged based on prefix order.
Roughly speaking, the node labels
indicate the encoding of the atomic formulas holding at each position---this is based on $\atypesig{X_i}{\fA}{\sigtarget}$ if the node corresponds to a sequence ending in $X_i$,
and $\tau_i$ if the node corresponds to a sequence ending in $(Y_i,\tau_i)$.
The edge labels $\rho$ indicate the shared elements
between $X_i$ and $Y_{i+1}$
or $Y_{i+1}$ and $X_{i+1}$.
This is similar to $\bunravelk{\fA}{\sigtarget,\sigmag}$,
except it includes all of the variations to the labels coming from safe restrictions
of the bag nodes
that are needed to make it a plump tree.

Formally, we build up the labels inductively based on the depth of the tree.
As we go, we also define for each $v \in \Pi'_k$ ending in $Z_i$ or $(Z_i,\tau_i)$,
a bijective function $\nu_v$ from $Z_i$ to $\set{1,\dots,\size{Z_i}}$
that defines the element index assigned at that node to each element in $Z_i$
in the tree encoding.
\begin{itemize}
\item The label at the root $v_0$ consists only of $D_0$,
and we have the empty map $\nu_{v_0}$ (since $X_0 = \emptyset$).
\item Consider the node
$v = X_0 (Y_1,\tau_1) X_1 \dots X_{m-1} (Y_m,\tau_m)$
and its parent of the form $u = X_0 (Y_1,\tau_1) X_1 \dots X_{m-1}$
with inductively defined $\nu_{u}$.
Fix some bijective $\nu_v$ from $Y_m$ to $\set{1,\dots,\size{Y_m}}$
that agrees with $\nu_u$ on $X_{m-1}$.
For each $R(a_1,\dots,a_l) \in \tau_m$,
add $R_{\nu_v(a_1),\dots,\nu_v(a_l)}$ to the label of $v$.
Add $D_{\size{Y_m}}$ to the label at $v$.
The edge label $\rho$ between $u$ and $v$ is defined to be the identity map from
$\set{1,\dots,\size{X_{m-1}}}$ to $\set{1,\dots,\size{X_{m-1}}}$.
\item Consider the node
$v = X_0 (Y_1,\tau_1) X_1 \dots (Y_m,\tau_m) X_m$
and its parent of the form $u = X_0 (Y_1,\tau_1) X_1 \dots (Y_m,\tau_m)$
with inductively defined $\nu_{u}$.
Fix some bijective $\nu_v$ from $X_m$ to $\set{1,\dots,\size{X_m}}$.
Then for each $R(a_1,\dots,a_l) \in \atypesig{X_m}{\fA}{\sigtarget}$,
add $R_{\nu_v(a_1),\dots,\nu_v(a_l)}$ to the label at $v$.
Add $D_{\size{X_m}}$ to the label at $v$.
The edge label between $u$ and $v$ is the function $\rho$
with domain $\valuation_u(Y_m \cap X_m)$
such that for each $a \in Y_{m} \cap X_m$,
$\rho(\valuation_u(a)) := \valuation_v(a)$.
\end{itemize}


%% file: interp.tex
\section{Interpolation}%
\label{sec:interp}

\subsection{Positive results}

The forward and backward mappings utilized for
the definability questions can also be used to prove
that $\gfp$ and $\gnfpk$ have a form of interpolation.

Let $\phi_{\L}$ and $\phi_{\R}$ be sentences
over signatures $\sigma_{\L}$ and $\sigma_{\R}$
such that $\phi_{\L} \models \phi_{\R}$ ($\phi_{\L}$ entails $\phi_{\R}$).
An \emph{interpolant} for such a validity is a formula $\theta$
for which $\phi_{\L} \models \theta$ and $\theta \models \phi_{\R}$,
and $\theta$ mentions only relations appearing in both $\phi_{\L}$ and $\phi_{\R}$ (their \emph{common signature}).
We say a logic $\cL$ has \emph{Craig interpolation} if for all $\phi_{\L},\phi_{\R} \in \cL$ with
$\phi_{\L} \models \phi_{\R}$, there is an interpolant $\theta \in \cL$ for it.
We say a logic $\cL$ has the stronger \emph{uniform interpolation} property if one can
obtain $\theta$ from $\phi_\L$ and a signature $\sigtarget$, and $\theta$ can serve
as an interpolant for any $\phi_{\R}$ entailed by $\phi_{\L}$ and
such that the common signature of $\phi_{\R}$ and
$\phi_{\L}$ is contained in $\sigtarget$. A uniform interpolant can be thought
of as the best approximation from above of $\phi_{\L}$ over $\sigtarget$.

Uniform interpolation holds for the $\mu$-calculus~\cite{interpolationmucalc},
and also for
$\unfpk$~\cite{lics15-gnfpi}.
Unfortunately, $\gfp[\sigoriginal]$ and $\gnfpk[\sigoriginal]$ both fail to have uniform interpolation
and Craig interpolation~\cite{HooglandMO99,lics15-gnfpi}.
However, if we disallow subsignature restrictions that change the guard signature,
then we can regain this interpolation property. This ``preservation of guard'' variant
was investigated first by Hoogland, Marx, and Otto in the context of Craig interpolation~\cite{HooglandMO99}. The uniform interpolation variant was introduced by D'Agostino and Lenzi~\cite{DAgostinoL15}, who called it \emph{uniform modal interpolation}.
Formally, we say a guarded logic $\cL[\sigoriginal,\sigmag]$
with guard signature $\sigmag \subseteq \sigoriginal$
has \emph{uniform modal interpolation} if
for any $\phi_{\L} \in \cL[\sigoriginal,\sigmag]$ and any subsignature $\sigtarget \subseteq \sigoriginal$ containing $\sigmag$,
there exists a formula $\theta \in \cL[\sigtarget,\sigmag]$ such that
 $\phi_{\L}$ entails $\theta$ and
for any  $\sigma''$ containing~$\sigmag$ with $\sigma'' \cap \sigoriginal \subseteq \sigtarget$
and any $\phi_{\R} \in \cL[\sigma'',\sigmag]$ entailed by $\phi_{\L}$,
$\theta$ entails $\phi_{\R}$.
It was shown in~\cite{DAgostinoL15} that $\gf$ has uniform modal interpolation.
  We strengthen
this to $\gfp$ and~$\gnfp^k$.

\begin{thm}%
\label{thm:modalinterp}
For $\sigoriginal$ a relational signature, $\sigmag \subseteq \sigoriginal$, and $k \in \N$:
$\gfp[\sigoriginal,\sigmag]$ and
$\gnfpk[\sigoriginal,\sigmag]$ sentences have uniform modal interpolation, and the interpolants
can be found effectively.
\end{thm}

Theorem~\ref{thm:modalinterp} also implies that $\unfpk$ has the traditional uniform interpolation property:
since the guard signature is empty for $\unfpk$, uniform modal interpolation and uniform interpolation coincide.
Another corollary is that $\unfp$ (not just $\unfpk$) has Craig interpolation.
Consider sentences $\phi_{\L}$ and $\phi_{\R}$ in $\unfp$ such that $\phi_{\L} \models \phi_{\R}$ and with $k$ the maximum width of $\phi_{\L}$ and $\phi_{\R}$.
Then the $\unfpk$ uniform interpolant for $\phi_{\L}$ with respect to $\sigma_{\L} \cap \sigma_{\R}$
can serve as a Craig interpolant for $\phi_{\L} \models \phi_{\R}$.
Note that uniform interpolation for $\unfpk$ and Craig interpolation for $\unfp$ were known already from~\cite{lics15-gnfpi}.

\begin{cor}\label{cor:unfpinterp}
For $\sigoriginal$ a relational signature and $k \in \N$: $\unfpk[\sigoriginal]$
has uniform interpolation and $\unfp[\sigma]$ has Craig interpolation. In both cases, the interpolants can be found effectively.
\end{cor}

The idea for the proof of Theorem~\ref{thm:modalinterp}
is to use the back-and-forth method from before, together with the uniform interpolation property of the $\mu$-calculus.
To illustrate this, we sketch the argument for $\gfp[\sigoriginal,\sigmag]$, before giving the formal proof for $\gnfpk$ below.

Consider  $\phi_{\L} \in \gfp[\sigoriginal, \sigmag]$ of width $k$
and subsignature $\sigtarget \subseteq \sigoriginal$ containing
$\sigmag$. We apply \lemmafwd to get a formula
$\phi_{\L}^{\mu} \in \Lmu[\sigcode{\sigoriginal}{k}]$ that captures codes of tree-like models of $\phi_{\L}$.

We want to go backward now, to get a formula over the subsignature $\sigtarget$.
We saw that the backward mapping from earlier can do this:
it can start with a $\mu$-calculus formula
over $\sigcode{\sigoriginal}{k}$,
and produce a formula in $\gfp[\sigtarget,\sigmag]$.
The formula produced by this backward mapping
has a nice property related to definability: it is equivalent to $\phi_{\L}$ exactly when
$\phi_{\L}$ is definable in $\gfp[\sigtarget,\sigmag]$.

In general, however, we do not expect $\phi_{\L}$ to be equivalent to a formula
over the subsignature---for uniform interpolation we just want to \emph{approximate} the formula over this subsignature.
The backward mapping of $\phi_{\L}^\mu$ (e.g.~using \lemmagfpbwd{Lemma $\gfp[\sigtarget,\sigmag]$-Bwd} from Section~\ref{sec:gfp} or \lemmagnfpbwd{Lemma~$\gnfpk[\sigtarget,\sigmag]$-Bwd} from Section~\ref{sec:gnfp}), does not always do this.
Hence, it is necessary to add one additional step before taking the backward mapping:
we apply uniform interpolation for the $\mu$-calculus~\cite{interpolationmucalc},
obtaining $\theta^\mu \in \sigcode{\sigtarget}{k}$ which is entailed by $\phi_{\L}^{\mu}$
and entails each $\Lmu[\sigcode{\sigtarget}{k}]$-formula implied by $\phi_{\L}^{\mu}$.
Finally, we apply \lemmagfpbwd{Lemma $\gfp[\sigtarget,\sigmag]$-Bwd} to $\theta^\mu$
to get $\theta \in \gfp[\sigtarget, \sigmag]$. We can check
that $\theta \in \gnfpk[\sigtarget,\sigmag]$ is the required uniform modal interpolant for $\phi_{\L}$ over subsignature $\sigtarget$.

We emphasize that although our uniform interpolation results and definability decision procedures
both use this back-and-forth approach,
the definability questions are easier in the sense that
they do not require interpolation for the $\mu$-calculus.

We now give the proof of Theorem~\ref{thm:modalinterp}.

\begin{proof}[Proof of Theorem~\ref{thm:modalinterp}]
We prove this for $\gnfpk[\sigoriginal,\sigmag]$,
but the proof is similar for $\gfp[\sigoriginal,\sigmag]$
(using the guarded unravelling $\gunravel{\fB}{\sigoriginal,\sigmag}$
instead of the plump unravelling $\sunravelk{\fB}{\sigoriginal,\sigmag}$,
and \lemmagfpbwd{Lemma~$\gfp[\sigtarget,\sigmag]$-Bwd} instead of \lemmagnfpbwd{Lemma~$\gnfpk[\sigtarget,\sigmag]$-Bwd}).

We construct the interpolant for $\phiL$ as follows:
\begin{enumerate}
\item apply \lemmafwd to get $\phiL^\mu \in \Lmu[\sigcode{\sigoriginal}{k}]$;
\item let $\cform{\sigoriginal}{k}$ be the $\Lmu[\sigcode{\sigoriginal}{k}]$-formula that expresses that a tree is consistent with respect to $\sigcode{\sigoriginal}{k}$;
\item get the uniform interpolant $\chi \in \Lmu[\sigcode{\sigtarget}{k}]$ for $\phiL^\mu \wedge \cform{\sigoriginal}{k}$
and subsignature $\sigcode{\sigtarget}{k}$ (using~\cite{interpolationmucalc});
\item apply Lemma~$\gnfpk[\sigtarget,\sigmag]$-Bwd to $\chi$ to get $\theta \in \gnfpk[\sigtarget,\sigmag]$.
\end{enumerate}

\noindent
We can see that $\theta$ is a formula over the subsignature $\sigtarget$ by the properties of the backward mapping.
We must show that $\theta$ satisfies the other properties required of a uniform interpolant.

\myparagraph{Original sentence entails interpolant}
First, we prove that
$\phiL \models \theta$.
Let $\fB$ be a $\sigoriginal$-structure and assume $\fB \models \phiL$.
Then $\decode{\sunravelk{\fB}{\sigoriginal,\sigmag}} \models \phiL$
since $\phiL$ is $\sgnkinvar[\sigoriginal,\sigmag]$-invariant.
Hence, by \lemmafwd,
we have $\sunravelk{\fB}{\sigoriginal,\sigmag} \models \phiL^\mu$.
Since $\sunravelk{\fB}{\sigoriginal,\sigmag}$ is a consistent $\sigcode{\sigoriginal}{k}$-tree,
this means that
$\sunravelk{\fB}{\sigoriginal,\sigmag} \models \phiL^\mu \wedge \cform{\sigoriginal}{k}$.
We can now use the fact that $\chi$ is a uniform interpolant for $\phiL^\mu \wedge \cform{\sigoriginal}{k}$,
to conclude that
$\sunravelk{\fB}{\sigoriginal,\sigmag} \models \chi$.
But $\chi$ is in $\Lmu[\sigcode{\sigtarget}{k}]$,
so the restriction of $\sunravelk{\fB}{\sigoriginal,\sigmag}$ to the subsignature $\sigcode{\sigtarget}{k}$
also satisfies $\chi$.
Moreover, the restriction of $\sunravelk{\fB}{\sigoriginal,\sigmag}$ to the subsignature is
$\sigcode{\sigtarget}{k}$-bisimilar to the unravelling with respect to this subsignature $\sigcode{\sigtarget}{k}$
(this relies on the fact that the guard signature $\sigmag$ is the same in both cases).
Hence, $\sunravelk{\fB}{\sigtarget,\sigmag} \models \chi$,
which by the backward mapping means
that $\fB \models \theta$.

\myparagraph{Interpolant entails appropriate sentences in subsignature}
Next, assume that $\phiL \models \phiR$
for some $\phiR \in \gnfpk[\sigR,\sigmag]$ with $\sigR \cap \sigoriginal \subseteq \sigtarget$. Let $\sigma'' = \sigoriginal \cup \sigR$.

We need to show that $\theta \models \phiR$.
In order to prove this, our reasoning will go back and forth between relational and tree structures.

We start by applying \lemmafwd to $\phiR$ to get $\phiR^\mu$.
Let $\cform{\sigR}{k}$ be the $\Lmu[\sigcode{\sigR}{k}]$-formula
that expresses that a tree is consistent with respect to $\sigcode{\sigR}{k}$.
We now want to show that $\cform{\sigoriginal}{k} \wedge \phiL^\mu$ entails $\cform{\sigR}{k} \rightarrow \phiR^\mu$.

\begin{clm}
$\cform{\sigoriginal}{k} \wedge \phiL^\mu \models \cform{\sigR}{k} \rightarrow \phiR^\mu$
over all $\sigcode{\sigma''}{k}$-structures.
\end{clm}
\begin{proof}[Proof of claim]
We first show that $\cform{\sigoriginal}{k} \wedge \phiL^\mu \models \cform{\sigR}{k} \rightarrow \phiR^\mu$
over all $\sigcode{\sigma''}{k}$-trees.
Suppose that $\tree$ is a $\sigcode{\sigma''}{k}$-tree and
$\tree \models \cform{\sigoriginal}{k} \wedge \phiL^\mu$.
Then $\tree$ must be consistent with respect to the subsignature $\sigcode{\sigoriginal}{k}$.
If $\tree$ is not consistent with respect to $\sigcode{\sigR}{k}$,
then $\tree$ trivially satisfies $\cform{\sigR}{k} \rightarrow \phiR^\mu$ and we are done.
Otherwise, $\tree$ is consistent with respect to both $\sigcode{\sigoriginal}{k}$
and $\sigcode{\sigR}{k}$, which is enough to conclude that it is a consistent $\sigcode{\sigma''}{k}$-tree.
Hence, by \lemmafwd, we have $\decode{\tree} \models \phiL$.
Since $\phiL \models \phiR$, this means that $\decode{\tree} \models \phiR$.
Another application of \lemmafwd allows us to conclude that $\tree \models \phiR^\mu$,
and hence by weakening, $\tree \models \cform{\sigR}{k} \rightarrow \phiR^\mu$ as desired.

We can use this to prove that $\cform{\sigoriginal}{k} \wedge \phiL^\mu \models \cform{\sigR}{k} \rightarrow \phiR^\mu$
over all $\sigcode{\sigma''}{k}$-structures.
Assume not. Then there is some $\sigcode{\sigma''}{k}$-structure $\fG$ such that
$\fG \models \cform{\sigoriginal}{k} \wedge \phiL^\mu \wedge \neg( \cform{\sigR}{k} \rightarrow \phiR^\mu)$.
By the tree-model property of $\Lmu$~\cite{BradfieldS07}, this means
there is some $\sigcode{\sigma''}{k}$-tree $\tree$ that witnesses this,
which contradicts the previous paragraph.
Therefore $\cform{\sigoriginal}{k} \wedge \phiL^\mu \models \cform{\sigR}{k} \rightarrow \phiR^\mu$
over all $\sigcode{\sigma''}{k}$-structures as required.
\end{proof}

Since $\cform{\sigoriginal}{k} \wedge \phiL^\mu$ entails $\cform{\sigR}{k} \rightarrow \phiR^\mu$ by the previous claim
and $\chi$ is a uniform interpolant for $\cform{\sigoriginal}{k} \wedge \phiL^\mu$ over the subsignature $\sigcode{\sigtarget}{k}$,
we know that $\chi \models \cform{\sigR}{k} \rightarrow \phiR^\mu$ over all $\sigcode{\sigma''}{k}$-structures.
We can use this to show that $\theta \models \phiR$.

Let $\fB$ be a $\sigma''$-structure such that $\fB \models \theta$.
Then $\decode{\sunravelk{\fB}{\sigtarget,\sigmag}} \models \theta$ since $\theta$ is $\sgnkinvar[\sigtarget,\sigmag]$-invariant
(since it is in $\gnfpk[\sigtarget,\sigmag]$).
Hence, $\sunravelk{\fB}{\sigtarget,\sigmag} \models \chi$ by properties of the backward mapping.
But $\sunravelk{\fB}{\sigma'',\sigmag} \models \chi$ as well,
since $\sunravelk{\fB}{\sigma'',\sigmag}$ is $\sigcode{\sigtarget}{k}$-bisimilar to $\sunravelk{\fB}{\sigtarget,\sigmag}$.
Hence, by the previous paragraph, we must have $\sunravelk{\fB}{\sigma'',\sigmag} \models \cform{\sigR}{k} \rightarrow \phiR^\mu$.
Since $\sunravelk{\fB}{\sigma'',\sigmag}$ is a consistent $\sigcode{\sigma''}{k}$-tree,
it is also $\sigcode{\sigR}{k}$-consistent.
Hence, $\sunravelk{\fB}{\sigma'',\sigmag} \models \phiR^\mu$ and
$\decode{\sunravelk{\fB}{\sigma'',\sigmag}} \models \phiR$.
Since $\fB$ and $\decode{\sunravelk{\fB}{\sigma'',\sigmag}}$ are $\sgnkinvar[\sigma'',\sigmag]$-bisimilar,
and $\phiR \in \gnfpk[\sigR,\sigmag]$ with $\sigR \subseteq \sigma''$,
this means that $\fB \models \phiR$ as desired.

This completes the proof that $\theta$ entails $\phiR$,
and hence completes the proof that $\theta$ is a uniform modal interpolant.
\end{proof}

\input{failure}


%% file: failure.tex
\subsection{Failure of uniform interpolation}\label{sec:failure}

In this section we will see that some natural extensions
and variants of our main interpolation theorems fail.

Although we have shown that \unfp has Craig interpolation,
it fails to have uniform interpolation.

\begin{prop}
Uniform interpolation fails for {\unfp}.
In particular,
there is a \unf antecedent with no uniform interpolant
in \lfpl, even when the consequents are restricted to sentences
in {\unf}.
The variant of uniform interpolation
where entailment is considered  only over finite structures also fails for $\unfp$.
\end{prop}

\begin{proof}
There is a $\unf$ sentence $\varphi$ that expresses that
unary relations $R$,$G$,$B$ form a 3-coloring of a graph with edge relation $E$.
This is because we only need unary negation to say:
\begin{itemize}
\item every node has a color: $\neg \exists x . (\neg R x \wedge \neg G x \wedge \neg B x )$,
\item neighbouring nodes do not share a color:
$\neg \exists x y . \left( E xy \wedge \big( (R x \wedge R y)  \vee \ldots \big) \right)$.
\end{itemize}
For readability, we have omitted the trivial guards $x=x$ in the unary negations above.

Consider a  uniform interpolant $\theta$ (in any logic) for the \unf sentence $\varphi$
with respect to its $\unf$-consequences in the signature $\set{E}$.
We claim that there cannot be an $\lfpl$ formula equivalent to $\theta$.

For all finite graphs $G$ that are not 3-colorable, let
$\psi_G$ be the \unf sentence corresponding to the
canonical conjunctive query  of $G$ over relation $E$---that  is, if
$G$ consists of edges $E$ mentioning vertices $v_1 \ldots v_n$, $\psi_G$
is $\exists v_1 \ldots v_n . ( \bigwedge_{e \in G, e=(v_i,v_j) } E v_i v_j )$.
Then $\varphi$ must entail $\neg \psi_G$, since the 3-coloring of a graph $G'$ satisfying
$\psi_G$
is easily seen to induce a 3-coloring on~$G$.

Now consider a finite graph $G$.
If $G$ is 3-colorable, then $G \models \varphi$,
and hence $G \models \theta$.
On the other hand, if $G$ is not 3-colorable,
then $G \models \psi_G$,
so $G \models \neg \theta$
because $\varphi$ entails $\neg \psi_G$
and thus, by the  assumption on $\theta$, $\theta$ entails $\neg \psi_G$.
Therefore, $\theta$ holds in $G$ iff
$G$ is 3-colorable.

Dawar~\cite{Dawar98} showed that
3-colorability is not expressible in the infinitary logic $\mathcal{L}^{\omega}_{\infty\omega}$ over
finite structures.
Since \lfpl can be translated into $\mathcal{L}^{\omega}_{\infty\omega}$ over finite structures,
this implies that $\theta$ cannot be in \lfpl.

The above argument only makes use of the properties of $\theta$ over finite structures, and
thus demonstrates the failure of the variant of uniform interpolation in the finite.
\end{proof}

We have
trivial uniform interpolants in existential second-order logic,
i.e.\ in NP\@.
The previous arguments shows that interpolants for $\unfp$ express
NP-hard problems, and thus cannot be in any  \ptime language
if \ptime is not equal
to NP\@.
We remark that one could still hope to find uniform interpolants for $\unfp$
by allowing the interpolants to live in a larger
fragment that is still ``tame'',
but we leave this as an open question.

Uniform interpolation also fails for \gso.

\begin{prop}\label{prop:failure-gso}
Uniform interpolation fails for \gso.
In particular, there is a \gf antecedent with
no uniform interpolant in \gso,
even when the consequents are restricted to sentences of \gf (or \unf)
of width 2.
\end{prop}

\begin{proof}
Let $\varphi \in \gf[\sigma]$ for $\sigma = \set{G,P,Q,R_1,R_2,S}$ be
\begin{align*}
&\forall z . [Q z \rightarrow \exists x y . (G z z x y \wedge S x y \wedge R_1 \, z x \wedge R_2 \, z y)]
\ \wedge
\\
&\forall x y . \Big[S x y \rightarrow
\exists x' y' . \Big( Gxyx'y' \wedge S x' y' \wedge R_1 \, x x' \wedge R_2 \, y y' \ \wedge \\
&\phantom{\forall x y . \; S x y \rightarrow \exists x' y' . \ } \big( ( P x' \wedge P y') \vee (\neg P x' \wedge \neg P y') \big)\Big) \Big]
\end{align*}
which implies that there is a ``ladder'' starting at every $Q$-node
(where $S$ connects pairs of elements on the same rung,
and $R_i$ connects corresponding elements on different rungs)
and the pair of elements on each rung agree on $P$.
The relation $G$ is used as a dummy guard to ensure that the formula is in $\gf$.

Then for each $n$, we can define over $\sigma' = \set{P,Q,R_1,R_2}$ a formula
$\psi_n$
\begin{align*}
&\bigg(\exists x . \Big(Q x \wedge \forall x_1 \dots x_n .\big( (\textstyle \bigwedge_i R_1 \, x_i x_{i+1} \wedge x_1 = x) \rightarrow P x_n \big)  \Big) \bigg)
\rightarrow \\
&\Big(\exists y . \big(Q y \wedge \exists y_1 \dots y_n .(\textstyle \bigwedge_i R_2 \, y_i y_{i+1} \wedge y_1 = y \wedge P y_n)\big)\Big)
\end{align*}
which expresses that if there is some $Q$-position $x$ such that
every $R_1$-path of length $n$ from $x$ ends in a position satisfying $P$,
then there is an $R_2$-path of length $n$ from some $Q$-position $y$
that ends in a position satisfying $P$.
Note that for all $n$,
$\psi_n$ can be written
in either \gf or \unf of width 2,
and $\varphi \models \psi_n$.

Assume for the sake of contradiction that there
is some uniform interpolant $\theta$
in \gso.

Over trees,
\gso coincides with \mso
(\cite{Courcelle97}, as cited in~\cite{GradelHO02}).
Hence, there is an equivalent $\theta'$ in \mso
over tree structures.
This means we can construct from $\theta'$
a nondeterministic parity tree automaton $\cA$
that recognizes precisely the language of trees
(with branching degree at most 2, say)
where $\theta'$ holds.

Let $m$ be the number of states in $\cA$.
Consider the ladder structure $\fA_m$
consisting of
a single element $a$ from which there is
an infinite $R_1$-chain of distinct elements
and an infinite $R_2$-chain of distinct elements,
where the $i$-th elements on each chain are connected by $S$,
elements on level $i$ and $i+1$ are guarded by $G$,
$P$ holds only at the $(m+1)$-st element in each chain,
and $Q$ holds only at $a$.

Because $\fA_m \models \varphi$,
we have $\fA_m \models \theta$.
But over $\sigma'$,
$\fA_m$ is a tree with branching degree at most 2, so
$\fA_m \models \theta'$.
Hence, there is an accepting run of $\cA$ on $\fA_m$.
Using a pumping argument,
we can pump a section of the $R_2$ branch before the $P$-labelled element
in order to generate an accepting run of $\cA$ on a new tree $\fA'_m$
where $P$ holds at the $(m+1)$-st element in the $R_1$-chain and
$P$ does not hold at that position in the $R_2$ chain.
Hence, this new tree $\fA'_m$
is a model for both $\theta'$ and $\theta$.
But $\fA'_m \not\models \psi_m$,
contradicting the fact that $\theta$ is a uniform interpolant.
\end{proof}

It has been known for some time that \gf fails to have even
ordinary Craig interpolation~\cite{HooglandMO99},
and hence fails to have uniform interpolation.
The previous proposition shows that we cannot get uniform interpolants for \gf
even when we allow the uniform interpolants
to come from~\gso.

\subsection{Failure of Craig interpolation for \gnfp}

It is natural to try to extend our results about $\gnfpk$ to the logic $\gnfp$.
Unfortunately, Craig interpolation fails for \gnfp.

\begin{prop}%
\label{prop:craigfailgnfp}
Craig interpolation fails for \gnfp.
In particular, there is an entailment of \gfp sentences
with no \gnfp interpolant,
even over finite structures.
\end{prop}

\begin{proof}
Define the $\gfp[\sigma]$ sentence $\varphi$ over signature $\sigma = \set{G,Q,R}$ to be
$\forall x . (Q x \rightarrow \varphi'(x))$ where $\varphi'(x)$ is
\begin{align*}
[\LFPA{X}{xy} . Gxyy \wedge (R xy \vee \exists y' . (G x y' y \wedge R xy' \wedge X y' y)) ](xx)
\end{align*}
Note that $\varphi'(x)$ implies that  $x$ has an $R$-path to itself.
Thus $\varphi$ implies
that every element where $Q$ holds has an $R$-path to itself.

Define the $\gfp[\sigma']$ sentence $\psi$ over signature $\sigma' = \set{P,Q,R}$ to be
\[
\forall x . \Big( (Q x \wedge P x) \rightarrow [\LFPA{X}{x} . \exists y . \big(R x y \wedge (P y \vee X y)\big) ] (x) \Big) .
\]
The sentence $\psi$ expresses that for all $Q$ and $P$ labelled elements $x$,
there is an $R$-path from $x$ leading to some node $y$ with $P y$.

We first argue that $\varphi \models \psi$.
Assume $\varphi$ holds in some $(\sigma \cup \sigma')$-structure, and consider some element $x_0$; we must show that $x_0$ satisfies $(Q x_0 \wedge P x_0) \rightarrow [\LFPA{X}{x} . \exists y . \big(R x y \wedge (P y \vee X y)\big) ](x_0)$.
If $Q$ and $P$ do not hold at $x_0$, then the condition is trivially satisfied here.
Otherwise, if $Q$ and $P$ do hold at $x_0$, then $\varphi$ ensures that there is an $R$-path from $x_0$ to itself, and hence there is an $R$-path from $x_0$ to a node where $P$ holds as required by $\psi$.

Now suppose for the sake of contradiction that there is a $\gnfp[\sigma \cap \sigma']$-interpolant~$\chi$
for $\varphi \models \psi$.
Note that $\chi$ only uses relations $Q$ and $R$.
Let $k$ be the width of $\chi$ in \nf.

\input{diagramforinterp}

We now define two $(\sigma \cup \sigma')$-structures, $\fA$ and $\fB$, that we will use to obtain a contradiction.
The graph structures for $\fA$ and $\fB$ are pictured in Figure~\ref{fig:interp}
(i.e.~this shows the structure in terms of relation $R$ only).

Let $\fA$ be the structure consisting of elements $\set{a_1, \dots, a_{k+1}}$ arranged in an $R$-cycle (i.e.~$Ra_1a_2, Ra_2a_3, \ldots, Ra_{k+1}a_1$). $Q$ holds of all elements, and $G$ holds of all triples of elements (in particular, $Ga_1a_2a_2$, $Ga_2a_3a_3$, etc.). The only element satisfying $P$ is $a_1$.

Let $\fB$ be the structure with elements $\set{b_0,b_1,\ldots,b_{k},b_{k+1}} \cup \set{c_1,\dots,c_{k+1}}$ where the elements $c_1,\ldots,c_{k+1}$ are arranged in an $R$-cycle,
$b_1, \ldots, b_{k+1}$ are arranged in an $R$-cycle, and $R c_1 b_0$ and $R b_0 b_1$. As in $\fA$, $Q$ holds of all elements and $G$ holds of all triples of elements. $P$ holds only at $b_0$.

Notice that $\fA$ satisfies $\varphi$,  and hence satisfies $\psi$.
But $\fB$ does not satisfy $\psi$,
because $b_0$ does not have an $R$-path to a node labelled with
$P$.

We claim that $\fA$ and $\fB$
are indistinguishable by $\gnfp[\sigma \cap \sigma']$ sentences of width $k$.
We must define a winning strategy for Duplicator in the $\sgnkinvar[\sigma \cap \sigma']$-bisimulation game between $\fA$ and~$\fB$. Because the structures
agree on $Q$ (since $Q$ holds of every element in both structures),
it suffices to show that they are indistinguishable with respect to relation~$R$.

Consider a position that is strictly guarded by $R$ or $Q$. Such a position consists of at most two elements (and if there are two elements, $u$ and $v$ must satisfy $Ruv$).
The initial position in the game (consisting of the empty partial homomorphism) is like this, so we must show that Duplicator has a strategy to ensure that she is always gets back to a position like this.

Suppose the active structure is $\fA$. Without loss of generality, we can assume
the  partial isomorphism $f$ in the position
has domain $a_1, a_2$. It must be the case that $R f(a_1) f(a_2)$ holds in $\fB$.
Because Spoiler can only extend his selection to at most $k$ elements, it is not possible for him to select all of the elements in the $R$-cycle in $\fA$. We
can assume, again without loss of generality, that he extends to a partial homomorphism that includes
all but $1$ element, $a_{i+1}$, in the $R$-cycle of $\fA$.
That is, we suppose
that Spoiler extends his selection to elements
$a_1,a_2,\dots,a_i$ and $a_{i+2},\dots,a_{k+1}$.
The sequence
$a_1, a_2,\dots,a_i$ forms an $R$-successor chain, so Duplicator
responds by mapping $a_3 \ldots a_i$ so that the images $f(a_2), \ldots, f(a_i)$
form  a chain of $R$-successors starting at
 $f(a_2)$. The chain is unique unless $f(a_2)$ is $c_1$; if $f(a_2)=c_1$,
she can choose to obtain either the successor chain
$c_1, c_2, c_3, \ldots$ or the chain $c_1, b_0, b_1, \dots$.
Likewise $a_{i+2},\dots,a_{k+1} a_1$ forms a successor chain
leading to $a_{1}$. Duplicator maps $a_{i+2},\ldots,a_{k+1}$ so
that $f(a_{i+2}),\ldots,f(a_{k+1}),  f(a_1)$
forms an $R$-successor chain leading to $f(a_1)$. The
chain is unique unless  $f(a_1)$ is $b_1$, and in this case
 she can choose either of the two candidate chains.
This is a new partial homomorphism with respect to $R$, and when Spoiler collapses to a single element or pair of elements satisfying $Ruv$, we
have a partial isomorphism as required.

Now consider the case where
the active structure is $\fB$, with  elements $u'$ and $v'$ with
$R u' v'$, and Spoiler extends to a set of elements $E$.
 Note that the  subgraph on $E$ induced by $R$ in $\fB$ is acyclic, due
to the size of $E$.
Let $V^+$ be the maximal subset of $E$ that contains $v'$ and is closed
under $R$.
The set $V^+$ could be the union of two chains of $R$-successors, or a single chain of $R$-successors.
Similarly let $U^-$ be the maximal subset of $E$ that contains $u'$ and is closed under $R$-predecessors.
$U^-$ can consist of two chains, or it can be a single
$R$-chain.
Note that an element of $V^+$ is the $i$-th successor of $u'$ for some $i$,
and thus Duplicator has no choice but to play the unique
$i$-th $R$-successor element  of $f(v')$ in her response.
Similarly  on $U^-$ Duplicator must play
the corresponding $R$-predecessor of  $f(u')$.
On the remaining elements $O$ of $E$, Duplicator can choose any homomorphism into
$\fA$. Such a homomorphism can be found by breaking the subgraph induced
on $O$ into connected components: for
each component $C$  choose an element $e_0 \in C$ and
map it to an $f(e_0)$ arbitrarily; each other element in $C$ is the $i$-th predecessor
or $i$-th successor of $e_0$, so we can map it to the unique $i$-th predecessor
or successor of $f(e_0)$.  The acylicity of $E$ guarantees that this mapping
is a homomorphism.
Although it is not injective, two elements $e_1$ and $e_2$ of $E$ map to the same element
in $\fA$ only if there is some $i$ such that either the $i$-th successor of
$e_1$ and $e_2$ are equal or the
$i$-th predecessor of $e_1$ and $e_2$
are equal. Thus acyclicity of $E$ guarantees that we cannot
have $R e_1 e_2$ for such an $e_1$ and $e_2$.
Hence, when Spoiler collapses to a strictly guarded position in his next move,
the resulting position is
a partial isomorphism as required.

Playing like this, Duplicator can continue to play indefinitely, so we she wins the bisimulation game. This shows that $\fA$ and $\fB$ are indistinguishable by \nf $\gnfpk[\sigma \cap \sigma']$-sentences, so they must agree on $\chi$.

Since $\fA \models \varphi$, we have $\fA \models \chi$.
Hence, $\fB \models \chi$.
But this implies that $\fB \models \psi$,
which is a contradiction.
\end{proof}


%% file: diagramforinterp.tex
\begin{figure}
\begin{tikzpicture}[]
\tikzset{vertex/.style = {shape=circle,draw,minimum size=2em}}
\node[align=center,font=\bf,text width=3cm,text badly centered] (aa) at (-1,2) {Structure $\fA$};
\node[align=center,font=\bf,text width=3cm,text badly centered] (bb) at (7,2) {Structure $\fB$};
\node[vertex] (a) at (-1,0) {$a_1$};
\draw (a) edge [->,dashed,out=0, in=90,looseness=10] node[font=\small,xshift=-2cm] {$k+1$ edges} (a);
\node[vertex] (b1) at (5,0) {$c_{1}$};
\node[vertex] (b2) at (7,0) {$b_0$};
\node[vertex] (b3) at (9,0) {$b_{1}$};
\draw (b1) edge [->,dashed,out=180, in=90,looseness=10] node[font=\small,xshift=2cm] {$k+1$ edges} (b1);
\draw (b3) edge [->,dashed,out=0, in=90,looseness=10] node[font=\small,xshift=-2cm] {$k+1$ edges} (b3);
\draw (b1) edge [->] (b2);
\draw (b2) edge [->] (b3);
\end{tikzpicture}
\caption{Structures used in the proof of Proposition~\ref{prop:craigfailgnfp}.}%
\label{fig:interp}
\end{figure}


%% file: conc.tex
\section{Conclusions}%
\label{sec:conc}
In this paper we explored effective characterizations
of definability in expressive fixpoint logics.
In the process, we have extended and refined the approach of going back and forth between relational structures and trees.
Boot-strapping from results about trees also allowed us to obtain results about interpolation for these logics.
We did not allow constants in the formulas in this paper,
but we believe that similar effective characterization and interpolation results
hold for guarded fixpoint logics with constants.

There are a number of open questions related to this work.
For $\gnfp^k$-definability, we proved only decidability results in this paper.
It would be interesting to determine the exact complexity of this problem,
perhaps using a direct automaton construction in the spirit of the construction given for $\gfp$.
We also leave open the question of deciding definability in $\gnfp$ and $\unfp$, without any width restriction.
For this question, one natural way to proceed is to try to bound the width of a defining sentence in terms
of some parameter of the input (e.g., its length).
For example, if we could show that a sentence of length $n$ in some larger logic $\cL$ is definable in $\gnfp$ iff
it is definable in $\gnfp^{f(n)}$ for some fixed function $f$, then we could test for membership in $\gnfp$ using the results of this paper.
We also note that our results on fixpoint logics
hold only when equivalence is considered over all structures, leaving open the corresponding questions
over finite structures.

In Corollary~7 of the conference version of this paper (\cite{icalp17}), we claimed to have proven that it was possible to decide membership in alternation-free $\gfp$,
a restriction of $\gfp$ to formulas with no nesting of both least and greatest fixpoints.
However, the proof of this claim was incorrect, and hence this question is open.
It is desirable to know if a sentence is in this alternation-free fragment of $\gfp$
since it has better computational properties: for instance, model checking for this alternation-free fragment can be done in linear time~\cite{dataloglite}.
This alternation-free fragment also corresponds to another previously studied logic called $\dloglite$~\cite{dataloglite}.
Hence, deciding membership in alternation-free~$\gfp$ (equivalently, $\dloglite$) remains an interesting open problem.

Finally, we showed that $\gnfp$ fails to have Craig interpolation.
This leaves open the question of whether there is
a decidable fixpoint logic that contains $\gnfp$ and has interpolation.
One candidate for this larger logic is called $\gnfpup$~\cite{gnfpup},
but it is not clear whether the methods in this paper could be adapted to prove such a result.

%% file: ack.tex
\section*{Acknowledgment}
Benedikt and Vanden Boom were funded by the EPSRC grants
PDQ (EP/M005852/1), ED$^3$ (EP/N014359/1), and DBOnto (EP/L012138/1).
